\newif\ifshortappendix
\shortappendixfalse

\ifshortappendix
  \documentclass[acmsmall]{acmart}

  \acmJournal{PACMMOD}
  \acmYear{2025} \acmVolume{3} \acmNumber{2 (PODS)}
  \acmArticle{XXX} \acmMonth{5} \acmPrice{}
  \acmDOI{10.1145/XXXXXXX}
  \received{December 2024}
  \received[revised]{February 2025}
  \received[accepted]{March 2025}
\else
  \documentclass[acmsmall,nonacm]{acmart}
\fi

\AtBeginDocument{%
  }

\usepackage{tikz-cd,enumerate}
\usepackage[bibliography=common]{apxproof} %
\usepackage{bbding} %

\usepackage{xcolor}

\definecolor{Dark Ruby Red}{HTML}{6a211f}
\definecolor{Dark Blue Sapphire}{HTML}{164a59}
\definecolor{Dark Gamboge}{HTML}{be7c00}

\definecolor{Desire}{HTML}{eb3b5a} %
\definecolor{Boyzone}{HTML}{2d98da} %
\definecolor{Royal Blue}{HTML}{3867d6} %
\definecolor{NYC Taxi}{HTML}{f7b731} %
\definecolor{Algal Fuel}{HTML}{20bf6b} %
\definecolor{Innuendo}{HTML}{a5b1c2} %
\definecolor{Twinkle Blue}{HTML}{d1d8e0} %
\definecolor{Gloomy Purple}{HTML}{8854d0} %

\colorlet{cBlue}{Royal Blue}
\colorlet{cYellow}{NYC Taxi}
\colorlet{cGreen}{Algal Fuel}
\colorlet{cRed}{Desire}
\colorlet{cGrey}{Innuendo}
\colorlet{cLightGrey}{Twinkle Blue}
\colorlet{cPurple}{Gloomy Purple}

\usepackage[xcolor, hyperref, cleveref, notion, quotation, electronic]{knowledge}
\knowledgeconfigure{quotation, protect quotation={tikzcd}}
\knowledgeconfigure{diagnose line=true, diagnose bar=true}

\IfKnowledgePaperModeTF{
}{
    \knowledgestyle{intro notion}{color={Dark Ruby Red}, emphasize}
    \knowledgestyle{notion}{color={Dark Blue Sapphire}}
    \hypersetup{
        colorlinks=true,
        breaklinks=true,
        linkcolor={Dark Blue Sapphire}, %
        citecolor={Dark Blue Sapphire}, %
        filecolor={Dark Blue Sapphire}, %
        urlcolor={Dark Blue Sapphire},
    }
    \IfKnowledgeElectronicModeTF{
    }{
        \knowledgeconfigure{anchor point color={Dark Ruby Red}, anchor point shape=corner}
        \knowledgestyle{intro unknown}{color={Dark Gamboge}, emphasize}
        \knowledgestyle{intro unknown cont}{color={Dark Gamboge}, emphasize}
        \knowledgestyle{kl unknown}{color={Dark Gamboge}}
        \knowledgestyle{kl unknown cont}{color={Dark Gamboge}}
    }
}

\newrobustcmd\decisionproblem[3]{
	\begin{center}\AP
	\fbox{\begin{tabular}{rl}
	\multicolumn{2}{l}{#1:} \\
	{\emph{Input}}: & \parbox[t]{.8\linewidth}{#2} \\   
	{\emph{Question}}: & \parbox[t]{.8\linewidth}{#3}
	\end{tabular}} 
	\end{center}
}

\usepackage[backgroundcolor=orange!20, textcolor={Dark Ruby Red}, textsize=tiny, disable]{todonotes}
\definecolor{green}{RGB}{0,120,0}
\definecolor{hlyellow}{RGB}{250, 250, 190}
\definecolor{diegoeditcolor}{RGB}{210,210,255}
\definecolor{remieditcolor}{RGB}{210,255,210}
\definecolor{migueleditcolor}{RGB}{255,204,153}
\newcommand{\sidediego}[1]{}
\newcommand{\sideremi}[1]{}
\newcommand{\remi}[1]{}

\newcommand{\diego}[1]{}
\newcommand{\miguel}[1]{}
\newcommand{\sidemiguel}[1]{}
\definecolor{light-gray}{gray}{0.9}

\knowledgenewrobustcmd{\eqseg}[1][\gamma]{\mathrel{\cmdkl{\sim_{#1}}}}
\knowledgenewrobustcmd{\neqseg}[1][\gamma]{\mathrel{\cmdkl{\not\sim_{#1}}}}

\usepackage{tikz}
\usetikzlibrary{calc,decorations.pathmorphing,shapes}

\newcounter{sarrow}
\newcommand\xrsquigarrow[1]{%
\stepcounter{sarrow}%
\mathrel{\begin{tikzpicture}[baseline= {( $ (current bounding box.south)$ )}]
\node[inner sep=.5ex] (\thesarrow) {$\scriptstyle #1$};
\path[draw,<-,decorate,
  decoration={zigzag,amplitude=0.7pt,segment length=1.2mm,pre=lineto,pre length=4pt}] 
    (\thesarrow.south east) -- (\thesarrow.south west);
\end{tikzpicture}}%
}

\usepackage{adforn}
\newcommand{\proofcase}[1]{\adforn{39}~\emph{#1}~}
\newcommand{\proofsubcase}[1]{\adforn{39}\adforn{39}~\emph{#1}~}

\knowledgenewrobustcmd{\vars}{\cmdkl{\textit{vars}}} %

\newcommand{\tup}[1]{\langle #1 \rangle}

\newcommand{\anexpansion}{\xi}

\knowledgenewrobustcmd\subaut[3]{#1\cmdkl{[#2,#3]}}

\newcommand{\orig}{\textit{orig}}
\newcommand{\cont}{\textit{contr}}

\newcommand{\changes}[1]{\color{black}#1\color{black}}

\usepackage{mathcommand}

\ExplSyntaxOn
\RenewDocumentCommand\withkl{mm}{
\int_gincr:N\knowledge_inner_modifier_count_int
\cs_gset:cpx
{\knowledge_inner_command:}
{\exp_not:N\cs_gset:Npn
\exp_not:c{\knowledge_inner_command:}
{\knowledge_inner_modifer_re_tl\knowledge_kl_modifiers_tl\exp_not:n{#1}}
\knowledge_kl_modifiers_tl\exp_not:n{#1}}
\knowledge_kl_modifiers_reset:
#2
\int_gdecr:N\knowledge_inner_modifier_count_int
}
\ExplSyntaxOff

\newrobustcmd{\+}[1]{\mathcal{#1}}
\newrobustcmd{\A}{\mathbb{A}}
\newrobustcmd{\B}{\mathbb{B}}
\newrobustcmd{\N}{\mathbb{N}}
\newrobustcmd{\Np}{\mathbb{N}_{>0}}
\newrobustcmd{\Z}{\mathbb{Z}}
\renewcommand{\phi}{\varphi}
\renewcommand{\leq}{\leqslant}
\renewcommand{\geq}{\geqslant}
\renewcommand{\emptyset}{\varnothing}
\newcommand{\set}[1]{\{#1\}}
\newrobustcmd{\equivclass}[2][]{[#2]^{#1}}

\newrobustcmd{\defeq}{\mathrel{\hat{=}}}

\newrobustcmd{\fun}{f}

\knowledgenewrobustcmd\vertex[1]{\cmdkl{V}(#1)}
\knowledgenewrobustcmd\edges[1]{\cmdkl{E}(#1)}
\knowledgenewrobustcmd{\core}{\mathop{\cmdkl{\textrm{core}}}}

\knowledgenewrobustcmd\qvar{\footnotesize\bullet} %
\knowledgenewrobustcmd{\contained}{\mathrel{\cmdkl{\subseteqq}}}
\knowledgenewrobustcmd{\semequiv}{\mathrel{\cmdkl{\equiv}}}

\newrobustcmd\smashxrightarrow[1]{%
  \raisebox{-.04em}{$%
    \xrightarrow{\smash{\raisebox{-.1em}{%
      \tiny{#1}%
    }}}%
  $}%
}%
\newrobustcmd\smashxleftarrow[1]{%
  \raisebox{-.04em}{$%
    \xleftarrow{\smash{\raisebox{-.1em}{%
      \tiny{#1}%
    }}}%
  $}%
}%

\knowledgenewrobustcmd{\atom}[1]{\,\cmdkl{\smashxrightarrow{\ensuremath{#1}}}\,}
\knowledgenewrobustcmd{\coatom}[1]{\,\kl[\atom]{\smashxleftarrow{\ensuremath{#1}}}\,}

\knowledgenewrobustcmd{\homto}{\mathrel{\cmdkl{\smashxrightarrow{hom}}}}
\newrobustcmd{\cohomto}{\mathrel{\kl[\homto]{\smashxleftarrow{hom}}}}
\RequirePackage{centernot}
\newrobustcmd{\nothomto}{\mathrel{\kl[\homto]{\centernot{\smashxrightarrow{hom}}}}}
\newrobustcmd{\notcohomto}{\mathrel{\kl[\homto]{\centernot{\smashxleftarrow{hom}}}}}

\newcommand{\xrightarrowdbl}[2][]{%
  \xrightarrow[#1]{#2}\mathrel{\mkern-14mu}\rightarrow
}
\knowledgenewrobustcmd{\surjto}{\mathrel{\cmdkl{\xrightarrowdbl{\smash{\textit{\tiny hom}}}}}}

\knowledgenewrobustcmd{\atompath}[1]{\,\xrsquigarrow{\,\smash{#1}\,}\,}
\knowledgenewrobustcmd{\atoms}{\cmdkl{\textit{atoms}}}

\knowledgenewrobustcmd{\size}[2][]{\cmdkl{\|}#2\cmdkl{\|^{#1}}}
\knowledgenewrobustcmd{\nbvar}[2][]{\cmdkl{\|}#2\cmdkl{\|^{#1}_{\textrm{var}}}}
\knowledgenewrobustcmd{\nbatoms}[2][]{\cmdkl{\|}#2\cmdkl{\|^{#1}_{\textrm{at}}}}
\knowledgenewrobustcmd{\nbseg}[2][]{\cmdkl{\|}#2\cmdkl{\|^{#1}_{\textrm{seg}}}}
\knowledgenewrobustcmd{\nbsegmentsOw}[2][]{\cmdkl{\|}#2\cmdkl{\|^{#1}_{\textrm{1seg}}}}
\knowledgenewrobustcmd{\nbsegmentsTw}[2][]{\cmdkl{\|}#2\cmdkl{\|^{#1}_{\textrm{2seg}}}}

\knowledgenewrobustcmd{\contr}[1]{\cmdkl{#1^{\textsf{ctr}}}}
\knowledgenewrobustcmd{\contrhom}[1]{\cmdkl{#1^{\textsf{ctr}}}}

\knowledgenewrobustcmd{\cdb}{
  \mathrel{\cmdkl{%
    \vDash^{\star}
  }}
}

\knowledgenewrobustcmd{\seggraph}{\mathop{\cmdkl{\+{S\!G}}}}

\newrobustcmd{\weakequiv}[1]{
  \mathrel{\withkl{\kl[\weakequiv]}{\cmdkl{%
    \sim_{#1}
  }}}
}
\knowledge{\weakequiv}{notion}

\newrobustcmd{\true}{\textit{tt}}
\newrobustcmd{\false}{\textit{ff}}

\knowledgenewrobustcmd{\Unfold}{\mathop{\cmdkl{\+U}}}

\def\ifempty#1{\ifx\empty#1\empty} %
\knowledgenewrobustcmd{\CQ}[1][]{\ifempty{#1}\textnormal{"CQ"}\else\cmdkl{\textnormal{CQ}(#1)}\fi}
\knowledgenewrobustcmd{\infUCRPQ}[1][]{\ifempty{#1}\cmdkl{\textnormal{UCRPQ}^\infty}\else\cmdkl{\textnormal{UCRPQ}^\infty(#1)}\fi}
\knowledgenewrobustcmd{\CRPQ}[1][]{\ifempty{#1}\textnormal{"CRPQ"}\else\cmdkl{\textnormal{CRPQ}(#1)}\fi}
\knowledgenewrobustcmd{\UCRPQ}[1][]{\ifempty{#1}\textnormal{"UCRPQ"}\else\cmdkl{\textnormal{UCRPQ}(#1)}\fi}
\knowledgenewrobustcmd{\underlying}[1]{\cmdkl{G_{#1}}}
\knowledgenewrobustcmd{\Exp}[1][]{\cmdkl{\textnormal{Exp}^{\smash{#1}}}}
\knowledgenewrobustcmd{\Refin}[1][]{\cmdkl{\textnormal{Ref}^{\smash{#1}}}}
\knowledgenewrobustcmd{\AppInf}[2]{\cmdkl{\textup{App}_{#2}^{\infty}(#1)}}
\knowledgenewrobustcmd{\App}[3]{\cmdkl{\textup{App}_{#2}^{\smash{\leq #3}}(#1)}}
\knowledgenewrobustcmd{\type}[2]{\cmdkl{\textup{type}_{#1}(#2)}}

\knowledgenewrobustcmd{\shrink}[1]{\mathop{\cmdkl{S_{#1}}}}

\knowledgenewrobustcmd{\classCRPQ}{\cmdkl{\+Q}}

\newrobustcmd{\marking}{\triangledown}

\makeatletter
\newcommand\incircbin
{%
  \mathpalette\@incircbin
}
\newcommand\@incircbin[2]
{%
  \mathbin%
  {%
    \ooalign{\hidewidth$#1#2$\hidewidth\crcr$#1\bigcirc$}%
  }%
}
\makeatother
\knowledgenewrobustcmd{\disconj}{\mathbin{\cmdkl{\incircbin{\land}}}}

\knowledgenewrobustcmd{\erasingmorphism}[2][]{\cmdkl{\phi^{#1}_{-#2}}}
\knowledgenewrobustcmd{\autoclass}[1]{\cmdkl{\+A_{#1}}}
\knowledgenewrobustcmd{\InclusionPb}[1][\autoclass{\classCRPQ}]{\cmdkl{\textsc{NFA Inclusion}(#1)}}

\knowledgenewrobustcmd{\atomneighbourhoodin}[2]{\cmdkl{\+N^{\smash{\textrm{in}}}_{\smash{#2}}(#1)}}
\knowledgenewrobustcmd{\atomneighbourhoodout}[2]{\cmdkl{\+N^{\smash{\textrm{out}}}_{\smash{#2}}(#1)}}
\knowledgenewrobustcmd{\edgeneighbourhoodin}[2]{\cmdkl{\+N^{\smash{\textrm{in}}}_{\smash{#2}}(#1)}}
\knowledgenewrobustcmd{\edgeneighbourhoodout}[2]{\cmdkl{\+N^{\smash{\textrm{out}}}_{\smash{#2}}(#1)}}

\knowledgenewrobustcmd{\alphabetmarking}{\cmdkl{\mathbb{M}}}

\knowledgenewrobustcmd{\axiomsCanon}{\cmdkl{\textup{(\textsf{Cnz})$_*$}}}
\knowledgenewrobustcmd{\axiomCanonMonotonicity}{\cmdkl{\textup{(\textsf{Cnz})$_\textsf{monotonic}$}}}
\knowledgenewrobustcmd{\axiomCanonContracted}{\cmdkl{\textup{(\textsf{Cnz})$_\textsf{contracted}$}}}
\knowledgenewrobustcmd{\axiomCanonCore}{\cmdkl{\textup{(\textsf{Cnz})$_\textsf{str-onto}$}}}
\knowledgenewrobustcmd{\axiomCanonNonRed}{\cmdkl{\textup{(\textsf{Cnz})$_\textsf{non-red}$}}}
\knowledgenewrobustcmd{\axiomCanonContainment}{\cmdkl{\textup{(\textsf{Cnz})$_\textsf{containment}$}}}
\knowledgenewrobustcmd{\axiomCanonMarking}{\cmdkl{\textup{(\textsf{Cnz})$_\textsf{marking}$}}}
\knowledgenewrobustcmd{\axiomStrongCanonCore}{\cmdkl{\textup{(\textsf{SCnz})$_\textsf{str-onto}$}}}

\knowledgenewrobustcmd{\axiomVarMarkingLoop}{\cmdkl{\textup{(\textsf{VM})$_\textsf{loop}$}}}
\knowledgenewrobustcmd{\axiomVarMarkingOut}{\cmdkl{\textup{(\textsf{VM})$_\textsf{out}$}}}
\knowledgenewrobustcmd{\axiomVarMarkingIn}{\cmdkl{\textup{(\textsf{VM})$_\textsf{in}$}}}

\knowledgenewrobustcmd{\Encoding}{\mathop{\cmdkl{\textrm{Enc}}}}

\usepackage{subcaption}
\AtEndPreamble{%
  \theoremstyle{acmdefinition}
  \newtheorem{remark}[theorem]{Remark}\crefname{remark}{Remark}{Remarks}
  \newtheorem{claim}[theorem]{Claim}\crefname{claim}{Claim}{Claims}
  \newtheorem{fact}[theorem]{Fact}\crefname{fact}{Fact}{Facts}
  \crefname{open}{Open question}{Open questions}
  \newtheoremrep{lemma}[theorem]{\sc Lemma}
  \newtheoremrep{proposition}[theorem]{\sc Proposition}
    \crefname{proposition}{proposition}{propositions}
    \Crefname{proposition}{Proposition}{Propositions}
  \newtheoremrep{corollary}[theorem]{\sc Corollary}
  \newtheoremrep{thm}[theorem]{\sc Theorem}
    \crefname{thm}{theorem}{theorems}
    \Crefname{thm}{Theorem}{Theorems}
  
}

\input{kl/abbreviations.kl}
\input{kl/complexity.kl}
\input{kl/databases.kl}
\input{kl/minimization.kl}
\input{kl/strongminimality.kl}
\input{kl/treepatterns.kl}
\input{kl/forestshaped.kl}
\input{kl/queries.kl}
\input{kl/lowerbounds.kl}

\begin{document}
\title{Minimizing Conjunctive Regular Path Queries}

\author{Diego Figueira}
\email{diego.figueira@cnrs.fr}
\orcid{0000-0003-0114-2257}
\affiliation{%
  \institution{Univ. Bordeaux, CNRS,  Bordeaux INP, LaBRI, UMR 5800}
  \city{F-33400, Talence}
  \country{France}
}
\author{Rémi Morvan}
\email{remi.morvan@u-bordeaux.fr}
\orcid{0000-0002-1418-3405}
\affiliation{%
  \institution{Univ. Bordeaux, CNRS,  Bordeaux INP, LaBRI, UMR 5800}
  \city{F-33400, Talence}
  \country{France}
}
\author{Miguel Romero}
\email{mgromero@uc.cl}
\orcid{0000-0002-2615-6455} %
\affiliation{%
  \institution{Dept.\ of Computer Science, Universidad Católica de Chile \& CENIA} %
  \city{Santiago}
  \country{Chile}
}

\begin{abstract}
  We study the \emph{minimization problem} for Conjunctive Regular Path Queries (CRPQs) and unions of CRPQs (UCRPQs).
This is the problem of checking, given a query and a number $k$, whether the query is  equivalent to one of size at most $k$. 
For CRPQs we consider the size to be the number of atoms, and for UCRPQs the maximum number of atoms in a CRPQ therein, motivated by the fact that the number of atoms has a leading influence on the cost of query evaluation.

We show that the minimization problem is decidable, both for CRPQs and UCRPQs.
We provide a 2ExpSpace upper-bound for CRPQ minimization, based on a brute-force enumeration algorithm, and an ExpSpace lower-bound. For UCRPQs, we show that the problem is ExpSpace-complete, having thus the same complexity as the classical containment problem. The upper bound is obtained by defining and computing a notion of maximal under-approximation.
Moreover, we show that for UCRPQs using the so-called \emph{simple regular expressions} consisting of concatenations of expressions of the form $a^+$ or $a_1 + \dotsb + a_k$, the minimization problem becomes "PiP2"-complete, again matching the complexity of containment.

\end{abstract}



\maketitle
\noindent
\raisebox{-.4ex}{\HandRight}\hspace{.2cm}This pdf contains internal links: clicking on a "notion@@notice" leads to its \AP ""definition@@notice"".%

\ifshortappendix
\noindent\raisebox{-.4ex}{\HandRight}\hspace{.2cm}A long version of this paper with all proofs is available at \url{https://arxiv.org/abs/XXXXXXX}
\else
\noindent\raisebox{-.4ex}{\HandRight}\hspace{.2cm}This article is based on a PODS'25 paper \cite{thispaper}. %
\fi

\section{Introduction}
\AP\label{sec:intro}
 Conjunctive Regular Path Queries (CRPQs) and unions of CRPQs (UCRPQs) form the backbone of graph database query languages, including the new ISO standard Graph Query Language (GQL) \cite{isoGQL} and the SQL extension for querying graph-structured data SQL/PGQ \cite{isoPGQ} (see also \cite{DBLP:conf/icdt/FrancisGGLMMMPR23,DBLP:conf/pods/FrancisGGLMMMPR23}).
 These extend the well-known classes of Conjunctive Queries (CQs) and unions of CQs (UCQs), with the ability to reason about paths in a graph. 
  Optimizing and understanding the fundamental properties of such queries has then become a major topic in graph database theory.

  Static optimization for CRPQs has received considerable attention. The basic study of containment and equivalence problems for CRPQs, possibly with unions and inverses, was initiated over 25 years ago \cite{Florescu:CRPQ,four-italians}, where they were shown to be "ExpSpace"-complete. These problems have also been investigated under different scenarios: restrictions on the shape of queries \cite{figueira_containment_2020}, restrictions on their regular languages \cite{FigueiraGKMNT20}, alternative semantics~\cite{FigueiraRomero23}, or under schema information~\cite{GGIM-kr22,GGIM-pods24}.
  This has enabled the study of more advanced static analysis problems motivated by the following general question: \emph{Can a given query be equivalently rewritten as one from a target fragment (which enjoys desirable properties)?}
        In the literature the problem has been studied where the target fragment are queries which either (i) avoid having infinite languages, or (ii) have a tree-like structure. This gives rise to the so-called (i) \emph{boundedness problem} for CRPQs ("ie", whether a CRPQ is equivalent to a UCQ) \cite{BarceloF019,FigueiraAnanthaAl24}, and (ii) \emph{semantic treewidth problem} for CRPQs ("ie", whether a CRPQ is equivalent to one of a given treewidth) \cite{BarceloRV16,FM2023semantic,FeierGM24}.

  \paragraph{Minimization of queries.}
  Minimization -- that is, the problem of transforming a query into a strictly smaller equivalent query -- is perhaps the most fundamental query optimization question. This problem corresponds to the question posed above, where the target fragment are queries of bounded sizes. 
  For CQs (and UCQs), minimization is well understood, and there exists a canonical unique minimal query, called the \emph{core}.
  The mechanism for obtaining such minimal query is simple: eliminate any atom from the query that results in an equivalent query ("ie", any atom which is `redundant'  in the sense of equivalence).
In contrast, minimization of CRPQs is poorly understood from a theoretical perspective. In this case, the situation is more challenging: there is no natural notion of `core’, and it is not clear whether a notion of `canonical’ smallest query may even be possible. In particular, eliminating redundant atoms of a CRPQ as done for CQs, in general results in a query which is neither minimal nor canonical.%

  In this paper we study the "minimization problem" for CRPQs and UCRPQs. In the case of CRPQs, we aim at minimizing the number of atoms of a CRPQ, and hence we formulate the problem as follows ($\semequiv$ denotes query equivalence, \changes{"ie", the fact that the queries output the same answer for all databases}):\footnote{\changes{Note that we can always assume $k$ to be smaller than the number of "atoms" of the input query, since otherwise the instance of the "minimization problem" is trivially solvable by answering `yes'. So, whether $k$ is given in unary or binary does not affect the size of the input.}}
 
\decisionproblem{""Minimization problem for CRPQs@Minimization problem""}
{A finite alphabet $\A$, a "CRPQ" $\gamma$ over $\A$ and $k \in \N$.}
{Is there a "CRPQ" $\delta$ over $\A$ with at most $k$ "atoms" 
such that $\gamma \semequiv \delta$?}
\medskip

On the other hand, in the case of UCRPQs, we minimize the maximum number of atoms of the CRPQs participating in a UCRPQ:
  
\decisionproblem{\reintro[Minimization problem]{Minimization problem for UCRPQs}}
{A finite alphabet $\A$, a "UCRPQ" $\Gamma$ over $\A$ and $k \in \N$.}
{Is there a "UCRPQ" $\Delta$ over $\A$ whose every "CRPQ" has at most $k$ "atoms" 
"st" $\Gamma \semequiv \Delta$?}
  \medskip

  Observe that the minimization problem for CRPQs and UCRPQs are two different problems: an algorithm for the minimization problem for UCRPQs (where the equivalent query may have unions) in principle does not imply any bound on the minimization for CRPQs (where we insist in being only one CRPQ).

  \paragraph{Contributions}
  We investigate the minimization problem for CRPQs and UCRPQs, and present several fundamental results. More concretely:
  \begin{itemize}
    \item We show that the minimization problem for CRPQs and UCRPQs are both decidable. As explained before, these are different problems  and we give two very different algorithms. Contrary as what happens for CQs, minimizing a CRPQ by unions of CRPQs may result in smaller queries, hence in a sense UCRPQ minimization may be seen as a strictly more powerful approach (\Cref{prop:unionsmatter}).
    \item For the "minimization of CRPQs", the algorithm is essentially by brute-force. By carefully bounding the sizes of the automata involved, we show that the algorithm can be implemented in "2ExpSpace" in \Cref{thm:2expspace-min-crpqs}. We also show an "ExpSpace"-hard lower bound in \Cref{thm:minimization-lowerbound}, leaving an exponential gap.
    \item For the minimization of UCRPQs we can apply a more elegant solution, and in fact we show how to compute `maximal under-approximations' of a query by UCRPQs of a given size (\Cref{lemma:approximation-for-finclass}). The minimization then follows by testing whether the given query is equivalent to its approximation of size $k$, yielding an "ExpSpace" upper bound (\Cref{coro:upperbound-ucrpqs}), which is tight with the lower bound (\Cref{coro:lowerbounds}).
    \item We consider subclasses of UCRPQs restricted to some commonly used regular expressions as observed in practice, namely, the so-called \AP""simple regular expressions"" (or \reintro{SRE}). These are concatenations of expressions of the form
      (i) $a^+$ for some letter over the alphabet $a \in \A$, or 
      (ii) $a_1 + \dotsb + a_m$ for some $a_1, \dotsc, a_m \in \A$.\footnote{Note that "eg" \cite[\S 6]{FM2023semantic} uses $a^*$ instead of $a^+$. This restriction is needed in \Cref{lem:canonization-SREs}.}
    We show that minimization of UCRPQs having such simple regular expressions is
    "PiP2"-complete (\Cref{thm:minimization-SRE}). %
   \item We explore some necessary and sufficient conditions for minimality. In particular, we show that non-redundancy ("ie", the fact that removing any atom results in a non-equivalent query) is necessary but not sufficient for minimality (also known to be the case for tree patterns \cite{min-tree-patterns}). We also investigate a notion of `"strong minimality"' which implies minimality (\Cref{coro:strong-min}), and can be used as a theoretical tool to prove minimality of queries.
   This result is based on \Cref{thm:structure-theorem}, which may be of independent interest, providing a tool to extract
    lower bounds on the number of "atoms" (and more generally properties on the underlying structure of queries, such as tree-width, path-width, etc.) that is necessary to express a "UCRPQ".
   \item We also discuss an alternative definition of size, where instead of the number of atoms we count the number of variables: we obtain upper bounds for the "variable-minimization problem" of "CRPQs" and "UCRPQs" in \Cref{sec:discussion}.
  \end{itemize}
  
  \paragraph{On the chosen size measure}
  A na\"ive algorithm for the "evaluation" of a union of $t$ "CRPQs" with $k$ atoms on a graph database $G$ gives a rough bound of $O\big(t  k  (|\vertex{G}||\edges{G}| r) + t  |\vertex{G}|^{2k}\big)$, where $r$ is the maximum size of the regular expressions it contains, and $\edges{G}$, $\vertex{G}$ are the set of edges and vertices of $G$, respectively.\footnote{This is obtained by first materializing a table with the answers to each RPQ atom $x \atom{L} y$ of the query. For each vertex $u\in V(G)$, we can compute the answers to $u \atom{L} y$, by a BFS traversal on the product of $G$ and   the NFA $\+A_L$ for the regular language $L$, taking roughly $O(|\edges{G}| r)$. Then we can evaluate each "CRPQ" as if it were a "conjunctive query" on the computed tables (each table having size at most $|\vertex{G}|^2$), in $O((|\vertex{G}|^2)^k) = O(|\vertex{G}|^{2k})$.}
  As we see, the most costly dependence is on $k$, since $G$ is the largest object ("ie", the database, several orders of magnitude larger than the remaining parameters in practice). The size of regular expressions and the number of unions have a less predominant multiplicative influence on the cost.
  Further, unions can be executed in parallel, which justifies the choice of taking the maximum size of the number of atoms of the CRPQs therein.
  However, other measures may also be reasonable. For example, taking the size to be the number of variables instead of the number of atoms is explored in \Cref{sec:varmin}.
  More complex measures including the size of regular expressions and the number of unions would need to take into account the drastically different roles of the parameters in the evaluation in view of the previous discussion ("eg", a simple sum of the parameters would not be a reasonable choice).

Our size measure of number of atoms is also natural from a practical perspective. In practice, systems typically evaluate CRPQs by combining   on-the-fly ``materialization'' of CRPQ atoms  with  relational database techniques, in particular using join algorithms (see "eg" \cite{milleniumDB24,eswc-crpqs24,cucumides-icdt23}).  The number of atoms (or joins) plays an important role in these algorithms.
  
  \paragraph{Related work}

  Minimization is well-understood for CQs and corresponds to the \reintro{core} in the Chandra-Merlin theory~\cite{DBLP:conf/stoc/ChandraM77}, that is, the smallest homomorphically-equivalent query, which is unique up to renaming of variables. The "minimization problem" is then "NP"-complete.
  For UCQs, the canonical minimal query consists of minimizing each CQ and removing redundant queries ("ie", removing a CQ disjunct $q$ if there is another disjunct $q'$ such that $q \subseteq q'$), which remains "NP".
  
  However, for unbounded homomorphism-closed queries, such as CRPQs, the existence of such unique minimal queries (even seen as infinitary unions of CQs) remains rather elusive. In particular, what breaks is the ``redundancy removal'', because there could be infinite chains of ever growing queries, as for instance in the Boolean CRPQ $q() = x \atom{a^+} x$.

  Minimization has also been studied for the class of \emph{tree patterns}~\cite{FFM08,KS08,min-tree-patterns}. Tree patterns are simple yet widely used tree-like queries for tree-like databases such as XML. These queries allow mild recursion in the form of descendent edges, that is, atoms of the form $x \atom{a^+} y$, where $x$ is the parent of $y$. 
  Minimization of tree patterns is now well-understood \cite{min-tree-patterns}:  it is known that non-redundancy is not the same as minimality, and that the "minimization problem" is "SigmaP2"-complete, the lower bound being highly non-trivial.
  
\medskip

Due to space constraints, some proofs and details are deferred to the Appendix.

\section{Preliminaries}
\label{sec-prelims}
\paragraph{Graph databases.}
\AP ""Graph databases"" are abstracted as edge-labelled directed graphs
$G = \langle \vertex{G}, \edges{G} \rangle$, 
where nodes of $\intro*\vertex{G}$ represent entities and labelled edges $\intro*\edges{G} \subseteq \vertex G \times \A \times \vertex G$
represent relations between these entities, with $\A$ being a fixed finite alphabet.

\paragraph{Conjunctive regular path queries (CRPQs) and unions of CRPQs (UCRPQs).}
\AP A ""CRPQ"" $\gamma$ is defined as a tuple $\bar z = (z_1,\hdots,z_n)$
of ""output variables""\footnote{For technical reasons (see the definition of "equality atoms") we allow for a variable to appear multiple times.},
together with a conjunction of ""atoms"" of the form
\AP$\bigwedge_{j=1}^m x_j \intro*\atom{L_j} y_j$, where each $L_j$ is a regular language %
 and where $m \geq 0$.
The set of all variables occurring in $\gamma$, namely\footnote{We neither assume 
disjointness nor inclusion between $\{z_1,\hdots,z_n\}$ and $\{x_1,y_1,\hdots,x_m,y_m\}$.}
$\{z_1,\hdots,z_n\}\cup\{x_1,y_1,\hdots,x_m,y_m\}$, is denoted by
$\intro*\vars(\gamma)$. Variables in $\vars(\gamma)\setminus \{z_1,\hdots,z_n\}$ are existentially quantified. 
We denote by $\intro*\atoms(\gamma)$ the set of "atoms" of $\gamma$.
Given a "database" $G$, we say that a tuple of nodes $\bar u = (u_1,\hdots,u_n)$
\AP""satisfies"" $\gamma$ 
on $G$ if there is a mapping
$\fun\colon \vars(\gamma) \to \vertex{G}$ such that $u_i = \fun(z_i)$ for all
$1 \leq i \leq n$, and for each $1 \leq j \leq m$,
there exists a (directed) path from $\fun(x_j)$ to $\fun(y_j)$ in $G$, labelled by
a word from $L_j$ (if the path is empty, the label is $\varepsilon$). The \AP""evaluation"" of $\gamma$ on $G$ is then the set of all tuples that "satisfy" $\gamma$ on $G$.

A ""union of CRPQs"" (\reintro{UCRPQs})
is defined as a finite set of "CRPQs", called \AP""disjuncts"", whose tuples of "output variables" have all the same arity. %
The "evaluation" of a union is defined as the union of its "evaluations". 
If a query has no "output variables" we call it ""Boolean"", and
its "evaluation" can either be the set $\set{()}$, in which case we say that $G$
\reintro{satisfies} the query, or the empty set $\set{}$.

\AP Given two "UCRPQs" $\Gamma$
and $\Gamma'$ whose "output variables" have the same arity,
we say that $\Gamma$ is \AP""contained"" in $\Gamma'$,
denoted by $\Gamma \intro*\contained \Gamma'$, if
for every "graph database" $G$, for every tuple $\bar u$ of $\vertex{G}$,
if $\bar u$ "satisfies" $\Gamma$ on $G$, then so does $\Gamma'$. We will hence reserve the symbol `$\subseteq$' for set inclusion.
The \AP""containment problem"" for "UCRPQs" is the problem of, given
two "UCRPQs" $\Gamma$ and $\Gamma'$, to decide if $\Gamma \contained \Gamma'$.
When $\Gamma \contained \Gamma'$ and $\Gamma' \contained \Gamma$  we say that
$\Gamma$ and $\Gamma'$ are \AP""equivalent"", denoted by
$\Gamma \intro*\semequiv \Gamma'$. 

\AP A ""conjunctive query"" (\reintro{CQ}) is in this context a "CRPQ" whose every atom is of the form $x \atom{a} y$ for $a \in \A$ ("ie", every language is a singleton $\set{a}$).
\AP A ""union of CQs"" (\reintro{UCQs}) is defined as a "UCRPQ" with the same property.
\begin{toappendix}
    A \AP""canonical database"" $G$ of a "CRPQ" $\gamma$ is any "canonical database" associated
    to an "expansion" of $\gamma$, see \cite[Definition 3.1]{Florescu:CRPQ}
    for a formal definition. We denote it by \AP$G \intro*\cdb \gamma$.
    A \reintro{canonical database} of a "UCRPQ" is a "canonical database" of one
    of its "disjuncts".

    An \AP""evaluation map"" from a "CRPQ" $\gamma$ to a "graph database" $G$
    in a function $f$ from variables of $\gamma$ to $G$ "st"
    for any atom $x \atom{L} y$ in $\gamma$, there is path from $f(x)$ to $f(y)$ in $G$
    labelled by a word of $L$.

    The "containment" between "UCRPQs" $\Gamma_1 \contained \Gamma_2$ is exactly characterized
    by the fact that for all "canonical database" $G_1 \cdb \Gamma_1$,
    there exists a "disjunct" $\gamma_2$ of $\Gamma_2$ "st" there is an "evaluation map"
    from $\gamma_2$ to $G_1$.
\end{toappendix}

\paragraph{Homomorphisms}
\AP
A ""homomorphism"" $\fun$ from a "CRPQ" $\gamma(x_1, \dotsc, x_m)$ to a "CRPQ" $\gamma'(y_1, \dotsc, y_m)$ is a mapping from $\vars(\gamma)$ to $\vars(\gamma')$ such that $\fun(x) \atom{L} \fun(y)$ is an "atom" of $\gamma'$ for every "atom" $x \atom{L} y$ of $\gamma$, and further $\fun(x_i)=y_i$ for every $i$.
Such a "homomorphism" $\fun$ is \AP""strong onto"" if for every "atom" $x' \atom{L} y'$ of $\gamma'$ there is an "atom" $x \atom{L} y$ of $\gamma$ such that $\fun(x)=x'$ and $\fun(y)=y'$.
We write $\gamma \intro*\homto \gamma'$ if there is a "homomorphism" from $\gamma$ to $\gamma'$, and $\gamma \intro*\surjto \gamma'$ if there is a "strong onto homomorphism".
In the latter case, we say that $\gamma'$ is a \AP""homomorphic image"" of $\gamma$.
A \reintro{homomorphism} $\fun$ from a graph database $G$ to a graph database $G'$ is a mapping from $\vertex{G}$ to $\vertex{G'}$ such that  for every edge $u \atom{a} v$ of $G$, it holds that $\fun(u) \atom{a} \fun(v)$ is an edge in $G'$. A "homomorphism" from a "CQ" to a graph database is defined analogously.  

\AP
It is easy to see that if $\gamma \homto \delta$ then $\delta \contained \gamma$, and in the case where $\gamma,\delta$ are "CQs" this is an ``if and only if'' \cite[Lemma 13]{DBLP:conf/stoc/ChandraM77}. 
\AP
Two "CQs" $\gamma,\delta$ are ""hom-equivalent"" if there are "homomorphisms" $\gamma \homto \delta$ and $\delta \homto \gamma$.
Hence, for any two "CQs" $\gamma, \delta$, we have $\gamma \semequiv \delta$ if, and only if, they are "hom-equivalent".
\AP The ""core"" of a "CQ" $\gamma$, denoted by $\intro*\core(\gamma)$
is the result of repeatedly removing any atom which results in an equivalent query. It is unique up to isomorphism (see, "eg", \cite{DBLP:conf/stoc/ChandraM77}). We say that a "CQ" is `a core' if it is isomorphic to its "core". If $\gamma$ and $\delta$ are "hom-equivalent" then they have
the same "core". Moreover, there is always an \AP""embedding""---"ie", a "homomorphism" which is injective both on variables and "atoms"---of $\core(\gamma)$ into $\gamma$.

\paragraph*{Refinements and expansions of (U)CRPQs}
\AP For an NFA $\+A$ and two states $q,q'$ thereof, we denote by $\intro*\subaut{\+A}{q}{q'}$ the ""sublanguage"" of $\+A$ recognized  when considering $\set{q}$ as the set of initial states and $\set{q'}$ as the set of final states.
\AP An ""atom $m$-refinement"" of a "CRPQ" "atom" $\gamma(x,y) = x \atom{L} y$ where $m\geq 1$ and $L$ is given by the NFA $\+A_L$ is any "CRPQ" of the form 
\begin{equation}
    \AP\label{eq:refinement}
    \rho(x,y) = x \atom{L_1} t_1 \atom{L_2} \hdots \atom{L_{n-1}} t_{n-1} \atom{L_n} y
\end{equation}
where $1 \leq n \leq m$, $t_1,\hdots,t_{n-1}$ are fresh (existentially quantified) variables,
and $L_1,\hdots,L_n$ are such that there exists a sequence $(q_0,\dotsc,q_n)$ of states of $\+A_L$
such that $q_0$ is initial, $q_n$ is final, and for each $i$, $L_i$ is either of the form
\begin{enumerate}[(i)]
	\item $\subaut{\+A_L}{q_{i-1}}{q_{i}}$, or 
	\item $\{a\}$ if the letter $a\in \A$ belongs to $\subaut{\+A_L}{q_{i-1}}{q_{i}}$.
\end{enumerate}
Additionally, if $\varepsilon \in L$, the \AP""equality atom"" ``$x = y$'' is also an \reintro{atom $m$-refinement}. Thus, an \reintro{atom $m$-refinement} can be either of the form \eqref{eq:refinement} or ``$x=y$''.
By definition, note that the concatenation
$L_1\cdots L_n$ is a subset of  $L$ and hence $\rho \contained \gamma$ for any "atom $m$-refinement" $\rho$ of $\gamma$.
An \AP""atom refinement"" is an "atom $m$-refinement" for some $m$.

Given a natural number $m$, an \AP""$m$-refinement"" of a "CRPQ" $\gamma(\bar x) = \bigwedge_{i} x_i \atom{L_i} y_i$ is any query resulting from: 1) replacing every "atom" by one of its "$m$-refinements@@atom", and 2)
should some "$m$-refinements@@atom" have "equality atoms",
collapsing the variables (and removing the identity atoms `$x=x$').
\AP A ""refinement"" is an "$m$-refinement" for some $m$.
Note that %
in a "refinement" of a "CRPQ"
the "atom refinements" need not have the same length.
For instance, both $\rho(x,x) = x \atom{c} x$ and $\rho'(x,y) = x \atom{a} t_1 \atom{a} y \coatom{c} y$ are "refinements" of $\gamma(x,y) = x \atom{a^*} y \coatom{c} x$.
\AP
We write $\intro*\Refin(\gamma(\bar x))$ to denote the set of all "refinements" of $\gamma(\bar x)$ and $\reintro*\Refin[\leq m](\gamma(\bar x))$ to the $m$-refinements. %

The set of \AP""expansions"" of a "CRPQ" $\gamma$ is the set $\intro*\Exp(\gamma)$ of all "CQs" which are "refinements" of $\gamma$.
In other words, an "expansion" of $\gamma$ is any "CQ" obtained from $\gamma$
by replacing each "atom" $x \atom{L} y$ by a path $x \atom{w} y$ for some
word $w \in L$. The "expansions" (resp.\ "refinements") of a "UCRPQ" are the "expansions" (resp.\ "refinements") of the "CRPQs" it contains.
We define \AP""atom expansions"" analogously to "atom refinements". For "UCRPQs" we use  $\Exp(\Gamma)$, $\Refin(\Gamma)$ and $\Refin[\leq m](\Gamma)$ as for "CRPQs".

Any "UCRPQ" is equivalent to the infinitary union of its "expansions". In light of this, 
the semantics for "UCRPQs" can be rephrased as follows. 
Given a "UCRPQ" $\Gamma(\bar x)$ and a graph database $G$, 
the "evaluation" of $\Gamma(\bar x)$ on $G$, denoted by $\Gamma(G)$, is the set of tuples 
$\bar{v}$ of nodes for which there is $\anexpansion \in \Exp(\Gamma)$ such that there is a "homomorphism" $\anexpansion \homto G$ that sends $\bar x$ onto $\bar v$. 

"Containment" of "UCRPQs" can also be characterized in terms of "expansions".
\begin{proposition}[Folklore, see e.g. {\cite[Proposition 3.2]{Florescu:CRPQ}} or
    {\cite[Theorem 2]{four-italians}}]
    \AP\label{prop:cont-char-exp-st} 
    Let $\Gamma_1$ and $\Gamma_2$ be "UCRPQs". Then the following are equivalent:
        (i) $\Gamma_1 \contained \Gamma_2$;
        (ii) for every $\anexpansion_1\in \Exp(\Gamma_1)$, $\anexpansion_1 \contained \Gamma_2$;
        (iii) for every $\anexpansion_1\in \Exp(\Gamma_1)$ there is $\anexpansion_2\in \Exp(\Gamma_2)$ such that $\anexpansion_2\homto \anexpansion_1$. 
\end{proposition}

\paragraph*{General assumptions}
To simplify proofs, we often assume that the regular languages are described via non-deterministic finite automata (NFA) instead of regular expressions,
which does not affect any of our complexity bounds.
However, for readability all our examples will be given in terms of regular expressions.
\AP
We denote by $\intro*\nbatoms{\gamma}$ the number of "atoms" of a "CRPQ" $\gamma$ and by $\intro*\nbvar{\gamma}$ the number of variables.
We extend these notations to a "UCRPQ" $\Gamma$ by letting
$\nbatoms{\Gamma} = \max_{\gamma \in \Gamma} \nbatoms{\gamma}$
and $\nbvar{\Gamma} = \max_{\gamma \in \Gamma} \nbvar{\gamma}$. %
We denote by $\size{\Gamma}$ the size (of a reasonable encoding) of a  "UCRPQ". 
For a CRPQ $\gamma$, we define its  \AP""underlying graph"" $\intro*\underlying{\gamma}$ of $\gamma$ as the directed multigraph obtained from $\gamma$ by ignoring the regular languages labelling the atoms of $\gamma$. 

We assume familiarity with basic concepts of directed multigraphs.
For simplicity, thorough the paper, by `graph' we mean a directed multigraph. We also adapt implicitly in the natural way, concepts defined for "CRPQs" to (directed multi)graphs (such a "homomorphisms", "embeddings", etc.).
A \AP""minor"" of a graph is any graph that can be obtained by removing
    edges, removing vertices---and their adjacent edges---, and \AP""contracting edges""---meaning that we identify the two endpoints of the edge and remove the edge from the graph.%
    \footnote{This definition is a trivial generalization of the notion of minors for undirected graphs.}

\section{Necessary \& Sufficient Conditions for Minimality}

This section explores some necessary and sufficient conditions for a query to be "minimal".
We start with some necessary definitions.
\AP
An ""internal variable"" of a "CRPQ" $\gamma$ is any non-"output variable" with both in-degree and out-degree 1.
A non-"internal variable" is called \AP""external"".
A \AP""one-way internal path""\footnote{This definition comes from \cite[\S 7]{FM2023semantic}, 
under the name `one-way internal path'; there is also an equivalent
notion for C2RPQs.}, or "internal path" for short, from $x_0$ to $x_n$ of a "CRPQ" $\gamma$ is a \emph{simple}\footnote{Meaning that all nodes are pairwise disjoint,
except that potentially $x_0 = x_n$.} path
\begin{align}
    x_0 \atom{L_1} x_1 \atom{L_2} \cdots \atom{L_n} x_n 
    \AP\label{eq:path-internal}
\end{align}
in $\gamma$ where $n > 0$ and every $x_i$ is "internal" with $i \in \lBrack 1,n-1\rBrack$.
A \AP""segment"" of a "CRPQ" $\gamma$ is a maximal "internal 
path" in $\gamma$, where ``maximal'' means that it cannot be extended on the left
or on the right. We say that a "segment" is cyclic if $x_0 = x_n$.
We identify two cyclic "segments" if they are equal up to circular permutation.
We say that a "segment" as in \eqref{eq:path-internal} is \AP""incident@@segment""
to a variable $y$ if $y = x_i$ for some $i$.%

\begin{figure}
  \centering
  \begin{minipage}{0.6\textwidth}
    \centering
    \includegraphics[width=\linewidth]{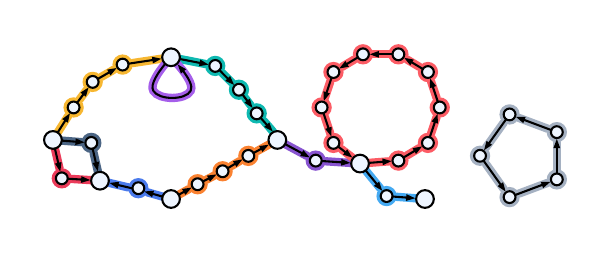}
    \subcaption{%
      \AP\label{fig:segments}%
      The "segments" of $\gamma$---labels are omitted. Each "segment" has a different color.  "Internal variables" are the smaller circles.
    }
  \end{minipage}%
  \hfill%
  \begin{minipage}{0.38\textwidth}
    \centering
    \includegraphics[width=\linewidth]{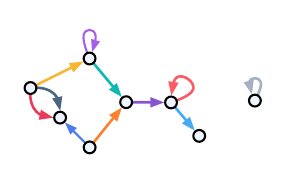}
    \subcaption{%
      \AP\label{fig:segment-graph}%
      The "segment graph" of $\gamma$.
    }
  \end{minipage}%
  \caption{"Segments" and "segment graph" of a "CRPQ" $\gamma$.}
\end{figure}

\begin{fact}
  \AP\label{fact:partition-into-segments}
  The "segments"---seen as sets of "atoms"---of $\gamma$ form a partition on its set of "atoms".
\end{fact}

See \Cref{fig:segments} for an example of a decomposition into "segments".
Note that each "segment" is "incident@@segment" to either zero "external variable" (isolated cycles),
to one (non-isolated cycles) or to two (non-cycles).
We denote by \AP$\intro*\nbseg{\gamma}$ the number of "segments" of a "CRPQ" $\gamma$,
and we extend this notation to "UCRPQs" by letting $\nbseg{\Gamma} \defeq
\max_{\gamma \in \Gamma}{\nbseg{\gamma}}$. By Fact~\ref{fact:partition-into-segments}, $\nbseg{\Gamma}\leq \nbatoms{\Gamma}$ always holds.

\subsection{Necessary Conditions: Contractions and Redundancy}
\paragraph{Contractions}
One simple (and tractable) way to make a query smaller is to `contract' any two consecutive atoms in a path with just one atom having the concatenation of the languages.
Formally, a
\AP ""contraction@@var"" of an "internal variable" $y$ in a "CRPQ" $\gamma$ is the result of replacing any pair of distinct atoms $x \atom{L} y$ and $y \atom{L'} z$ with $x \atom{L \cdot L'} z$ for $L \cdot L' \defeq \set{w \cdot w' : w \in L, w' \in L'}$.\footnote{Note that "contraction of internal variable" is a particular case of "edge contraction".} Observe that this results in an "equivalent" query.
\AP A ""contraction"" of $\gamma$ is any "CRPQ" obtained by repeatedly "contracting@@var" "internal variables".
A "CRPQ" $\gamma$ is \AP""fully contracted"" if it cannot be "contracted".
In other words, $\gamma$ is "fully contracted" "iff" $\nbseg{\gamma} = \nbatoms{\gamma}$.
A "contraction" of a "UCRPQ" is a "contraction" of a "CRPQ" therein (obtaining the "UCRPQ" where one "CRPQ" $\gamma$ was replaced by a "contraction" of $\gamma$). A "UCRPQ" is then \reintro{fully contracted} if each "CRPQ" therein is "fully contracted". %
\begin{fact}
  \AP\label{fact:produce-fully-contracted}
  "Contractions" preserve "semantic equivalence". Further, from a "UCRPQ" $\Gamma$ one can produce, in polynomial time, an "equivalent" one that is "fully contracted" with $\nbseg{\Gamma}$ "atoms".
\end{fact}

In particular, if a "UCRPQ" $\Gamma$ is "minimal" then $\Gamma$ is "fully contracted" and $\nbseg{\Gamma} = \nbatoms{\Gamma}$; in other words, $\nbseg{\Gamma}$ is an upper bound on the number of atoms of the "minimal" equivalent query.

\paragraph{Redundancy}
Another way to reduce the number of atoms of a query is to remove any \AP""redundant atom"", that is, any "atom" whose removal results in an "equivalent" query. 
When there are no such redundant atoms, we say that the query is \reintro{non-redundant}.
However, this is a more difficult problem, since it involves testing for query equivalence, an "ExpSpace"-complete problem.

\begin{proposition}
  \AP\label{prop:lowerbound-non-redundant}
  Testing whether a "(U)CRPQ@UCRPQ" is "non-redundant" is "ExpSpace"-complete.
\end{proposition}
\begin{proof}
  The upper bound is trivial: for every "atom" we remove it and check "equivalence".
  
  For the lower bound, we use the construction of \Cref{prop:variation-figueira} for containment. Let $\delta() 
  \defeq x' \atom{K} y' \land \bigwedge_i x \atom{L_i} y$ be the "disjoint conjunction" of $\gamma_1()$ and $\gamma_2()$ as defined in \Cref{prop:variation-figueira}.
  We first strengthen the construction to ensure the following two properties:
  \begin{enumerate}
    \item \textbf{$K$ cannot be mapped inside any $L_i$}: There is no word of $K$ which appears as factor of a word from some $L_i$. For this, it suffices to add a special letter at the beginning and the end of every word of $K$ which is not in any of the $L_i$'s. That is, we can define a new $K^{\textit{new}} \defeq \# \cdot K^{\textit{old}} \cdot \#$ for a new symbol $\#$.
    \item \textbf{For every $j$ there is $w_j \in L_j$ such that $w_j \not\in L_i$ for every $i \neq j$}: it suffices to add a special word ("eg" using a new alphabet letter) to each $L_i$. For example, we can define $L_{i}^{\textit{new}} \defeq L_{i}^\textit{old} \cup \set{@_i}$, where $@_i$ is a fresh alphabet letter.
  \end{enumerate}
  It is easy to see that these modifications preserve all the properties needed for the "containment problem" to still be "ExpSpace"-hard.
  
  We show that $\delta()$ is "non-redundant" "iff" $x' \atom{K} y' \contained \bigwedge_i x \atom{L_i} y$.
  
  \proofcase{$K$ cannot be removed.} We first show that removing the atom $x' \atom{K} y'$ from $\delta()$ results in a non-"equivalent" query $\delta'()$. Indeed, if it is removed then for any expansion of $\delta'()$ there will not be any word from $K$ that can be used to map into the expansion due to the first point above.

  \proofcase{If some $L_j$ is redundant, then containment holds.} Consider the result $\delta'()$ of removing an atom $x \atom{L_j} y$ from $\delta()$. Consider all the expansions of $\delta'()$ that choose $w_i$---defined in the second point above---as the "atom expansion" for $L_i$ for every $i \neq j$. It follows that for any such expansion there must be an expansion of $\delta()$ that maps necessarily $x \atom{L_j} y$ to the "atom expansion" of $x' \atom{K} y'$ in $\delta'()$. Otherwise, we would be mapping some word of $L_j$ to some $w_i$ with $i\neq j$, which we know it is not possible due to the second point above.
  This means that $x' \atom{K} y' \contained \bigwedge_i x \atom{L_i} y$.

  \proofcase{If containment holds, then all $L_i$'s redundant.} Finally, observe that if $x' \atom{K} y' \contained \bigwedge_i x \atom{L_i} y$ then the query is "equivalent" to $x' \atom{K} y'$.

  \medskip

  Overall, we obtained that the following are equivalent:
  \begin{enumerate}[(a)]
    \item $\delta()$ is "redundant",
    \item an atom $x \atom{L_i} y$ of $\delta()$ is "redundant",
    \item the containment $x' \atom{K} y' \contained \bigwedge_i x \atom{L_i} y$ holds,
    \item all atoms $x \atom{L_i} y$ of $\delta()$ are "redundant".\qedhere
  \end{enumerate}
\end{proof}

While in the case of "conjunctive queries" "non-redundancy" is the same as "minimality", for "CRPQs" and "UCRPQs" this is not the case, even if the query is "fully contracted". 

\begin{proposition}
  There are "fully contracted" "non-redundant" "CRPQs" which are not "minimal".
\end{proposition}

\begin{proof}
  \begin{figure}
    \includegraphics[scale=.8]{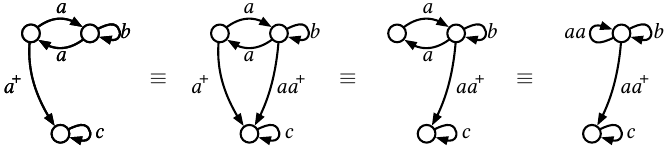}
    \caption{"Equivalent" "CRPQs". The leftmost query is "fully contracted", "non-redundant" and not "minimal", since it is "equivalent" to the rightmost query.}
    \AP\label{fig:ex-equiv-queries}
  \end{figure}
  \Cref{fig:ex-equiv-queries} contains a simple witnessing example. It is trivial to see that the "CRPQs" in the figure are all "equivalent" and that the leftmost query is "non-redundant" and "fully contracted". However, it is not "minimal" since it is "equivalent" to the rightmost "CRPQ".
\end{proof}

\subsection{A Sufficient Condition: Strong Minimality}

We have seen some sound ways to reduce the size of queries. But how can we ensure that a query is actually minimal? Here we give a theoretical tool which can ensure this by means of finding some "expansion" which is a witness for the query to have ``many atoms''. We call this "strong minimality".

Say that an "expansion" $\anexpansion$ of a "UCRPQ" $\Gamma$ is \AP""hom-minimal"" when, for every "expansion" $\anexpansion'$ of $\Gamma$, if $\anexpansion' \homto \anexpansion$ then $\anexpansion'$ and $\anexpansion$ are "hom-equivalent". 

A "UCRPQ" $\Gamma$ is \AP""strongly minimal"" if 
it has a "hom-minimal" "expansion" $\anexpansion \in \Exp(\Gamma)$ "st" $\nbseg{\core(\anexpansion)} = \nbatoms{\Gamma}$. We will next show that this is a sufficient condition for minimality.

The \reintro{segment graph} $\reintro*\seggraph(\gamma)$ of a "CRPQ" $\gamma$ is the directed multigraph obtained by replacing segments of $\gamma$ with edges in its "underlying graph", as illustrated in \Cref{fig:segment-graph} (a formal definition can be found in \Cref{defn:segmentgraph}).

\begin{toappendix}
\begin{definition}\AP\label{defn:segmentgraph}
  Given a "CRPQ" $\gamma$, we define its \AP""segment graph"" $\intro*\seggraph(\gamma)$
  to be the directed multigraph defined by:
  \begin{itemize}
    \item every "external variable" of $\gamma$ is a vertex of $\seggraph(\gamma)$,
      and moreover, for every "cyclic segment" $\sigma$
      that is not "incident@@segment" to any "external variable", we create a new variable
      $x_\sigma$;
    \item for every "cyclic segment" $\sigma$, we have a self-loop around $x_{\sigma}$,
      and for any "external variable" $x$ and $y$ of $\gamma$, we have an edge
      from $x$ to $y$ in $\seggraph(\gamma)$ for any "segment" starting at $x$ and
      ending at $y$.
  \end{itemize}
\end{definition}

The notion is illustrated in \Cref{fig:segment-graph}.
The motivation behind "segments" is that they are essentially the
dual to "atom refinements".
\end{toappendix}

\begin{fact}
  \AP\label{fact:segment-graph-is-contraction}
  The "segment graph" of $\gamma$ can always be obtained from its "underlying
  graph" by "contracting internal variables". 
\end{fact}
\begin{figure}
  \centering
  \begin{minipage}{0.38\textwidth}
    \centering
    \includegraphics[width=\linewidth]{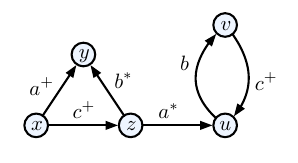}
    \subcaption{%
      \AP\label{fig:example-crpq}%
      A "CRPQ" $\gamma$.
    }
  \end{minipage}%
  \hfill%
  \begin{minipage}{0.28\textwidth}
    \centering
    \includegraphics[width=\linewidth]{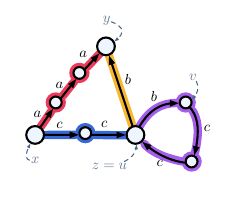}
    \subcaption{%
      \AP\label{fig:segments-expansion}%
      An "expansion" $\xi$ of $\gamma$, together with its "segments".
    }
  \end{minipage}%
  \hfill%
  \begin{minipage}{0.3\textwidth}
    \centering
    \includegraphics[width=.66\linewidth]{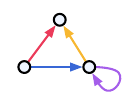}
    \subcaption{%
      \AP\label{fig:segment-graph-expansion}%
      The "segment graph" of $\xi$.
    }
  \end{minipage}%
  \caption{%
  \AP\label{fig:seggraph-of-expansion}
    Illustration of \Cref{prop:seggraph-of-expansion}:
    \Cref{fig:segment-graph-expansion} can be obtained as a "minor" of 
    \Cref{fig:example-crpq}.
  }
\end{figure}

\begin{propositionrep}
  \AP\label{prop:seggraph-of-expansion}
  Let $\gamma$ be a "CRPQ" and $\anexpansion$ be an "expansion" of $\gamma$.
  Then $\seggraph(\anexpansion)$ is a "minor" of
  the "underlying graph" of $\gamma$.%
  \end{propositionrep}

In particular, $\nbseg{\anexpansion} \leq \nbatoms{\gamma}$.
See \Cref{fig:seggraph-of-expansion} for an example.
\begin{proof}
  The "underlying graph" of $\anexpansion$ is obtained from
  the "underlying graph" of $\gamma$ by:
  \begin{enumerate}
    \item contracting some edges---corresponding to "atom refinements"
      for the word $\varepsilon$,
         \item potentially removing isolated vertices,\footnote{This can happen "eg" in the "atom refinement" of $x \atom{a^*} x$ when dealing with the empty word.} , and
    \item "refining" some edges---corresponding to "atom refinements"
      for words of length at least 2.
  \end{enumerate}
  Let $G'$ be the "underlying graph" of $\gamma$ to which we applied 
  all operation of type (1) and (2). By construction,
  $G'$ is a "minor" of $\underlying{\gamma}$.
  Notice then that if $H$ is a graph obtained by "refining" one edge of $G'$,
  then $\seggraph(H)$ and $\seggraph(G')$ are isomorphic.
  By trivial induction, it follows that $\seggraph(\anexpansion)$
  is isomorphic to $\seggraph(G')$.
  In turn, by \Cref{fact:segment-graph-is-contraction}, $\seggraph(G')$
  is an "edge contraction" of $G'$, which concludes the proof
  since the latter is a "minor" of $\underlying{\gamma}$.
\end{proof}
The next proposition provides a helpful tool to prove lower bounds on the
number of "atoms"---but also on the structure---required to express a query.
\AP\phantomintro{Semantical Structure}

\begin{theorem}[\reintro{Semantical Structure}]
  \AP\label{thm:structure-theorem}
  Let $\Gamma$ be a "UCRPQ".
  Let $\anexpansion$ be a "hom-minimal" "expansion" of $\Gamma$,
  and $\Delta$ be any "UCRPQ" "equivalent" to $\Gamma$.
  Then there exists some $\delta \in \Delta$ "st" the "segment graph" $\seggraph(\core(\anexpansion))$
  of the "core" of $\anexpansion$ is a "minor" of the "underlying graph" of $\delta$. 
\end{theorem}
\begin{proof}%
  Let $\Gamma$ be a fixed "UCRPQ", and let $\Delta$ be a "equivalent" "UCRPQ".
  Let $\anexpansion_\Gamma$ be a "hom-minimal" "expansion" of $\Gamma$.
  Since $\Gamma \contained \Delta$, there exists an "expansion" $\anexpansion_\Delta$ of $\Delta$
  "st" $\anexpansion_\Delta \homto \anexpansion_\Gamma$. Likewise, since $\Delta \contained \Gamma$, there exists an "expansion" $\anexpansion'_\Gamma \in \Exp(\Gamma)$ "st"
  $\anexpansion'_\Gamma \homto \anexpansion_\Delta$. Overall, we have
  $\anexpansion'_\Gamma \homto \anexpansion_\Gamma$ and so,
  by "hom-minimality" of $\anexpansion_\Gamma$, it is "hom-equivalent" to $\anexpansion'_\Gamma$. In turn, this implies that $\anexpansion_\Delta$ is "hom-equivalent" to $\anexpansion_\Gamma$, 
  and thus there exists an "embedding" of
  $\core(\anexpansion_\Gamma)$ into $\anexpansion_\Delta$.
  Note moreover that such an "embedding" must send variables of
  in-degree 0 (resp. out-degree 0) to nodes of in-degree 0 (resp. out-degree 0)
  and so by \Cref{coro:embedding-segments},
  there is an "embedding" from $\seggraph(\core(\anexpansion_\Gamma))$
  into $\seggraph(\anexpansion_\Delta)$.
  Letting $\delta$ be the "disjunct" of $\Delta$ of which $\anexpansion_\Delta$ is an "expansion",
  \Cref{prop:seggraph-of-expansion} implies that $\seggraph(\anexpansion_\Delta)$ is a
 "minor" of the "underlying graph" of $\delta$.
  Hence, $\seggraph(\core(\anexpansion_\Gamma))$ is a subgraph of a "minor", and hence a "minor", of the "underlying graph" of $\delta$.
\end{proof}

\begin{corollaryrep}[of \Cref{thm:structure-theorem}]\AP\label{coro:strong-min}
  Every "strongly minimal" "UCRPQ" is "minimal".
\end{corollaryrep}
\begin{proof}
  The number of edges of $\seggraph(\core(\anexpansion))$ equals $\nbseg{\core(\anexpansion)}$, and a "minor" can only decrease the number of edges.
\end{proof}
In fact, it can be seen that the assumption that $\anexpansion$ is "hom-minimal" in
\Cref{thm:structure-theorem} is necessary as otherwise the statement would be false (see \Cref{rk:structure-theorem} for details).
Also, note in particular that \Cref{thm:structure-theorem} implies
$\nbseg{\core(\anexpansion)} \leq \nbatoms{\Delta}$. 
But it can also be used to obtain lower bounds on, for instance,
the tree-width of $\Delta$, and hence the one-way semantic tree-width\footnote{Defined in
\cite[\S 1, p. 7]{FM2023semantic}.} of $\Gamma$, 
or more generally to prove that $\Gamma$ cannot be "equivalent" to a "UCRPQ" whose
underlying graphs all belong to a "minor"-closed class of graphs.

\begin{toappendix}
  \begin{proposition}
    \AP\label{prop:union-segments}
    Let $\gamma$ be a "CRPQ".
    The set of "atoms" of any path $x_0 \atom{a_1} \cdots \atom{a_n} x_n$ of $\gamma$
    where either $x_0 = x_n$ or where both $x_0$ and $x_n$ are "external" 
    is a finite union of "segments" of $\gamma$.
  \end{proposition}

  \begin{proof}
    The statement deals with the set of "atoms" of the path---and not the path itself---,
    so "wlog", up to a circular permutation of the path, we assume that (\adforn{75})
    if $x_0 = x_n$
    then either $x_0$ and $x_n$ are "external", or all $x_i$'s ($i \in \lBrack 1,n-1\rBrack$)
    are "internal".

    We prove the statement by induction on the length of the path.
    We identify three cases:
    \begin{enumerate}
      \item each $x_i$ ($i \in \lBrack 1,n-1 \rBrack$) is both
        "internal" and distinct from all $x_j$'s ($j\in \lBrack 0,n\rBrack$);
      \item $x_0$ and $x_n$ are "external" and there exists
        $i \in \lBrack 1,n-1 \rBrack$ "st" $x_i$ is "external";
      \item $x_0 = x_n$ and there exists $k \in \lBrack 1,n-1 \rBrack$
            "st" $x_k = x_0 \mathrel{(=} x_n)$.
    \end{enumerate}
    Next, we show that this covers all possible cases.

    If we are not in the first case, then either some
    $x_i$ ($i \in \lBrack 1,n-1 \rBrack$) is "external" or is equal to some $x_j$ ($j\in \lBrack 0,n\rBrack$). For the former, by (\adforn{75}) we get that $x_0$ and $x_n$ are necessarily
    "external", and we fall in case 2.
    
    For the former, we know that all $x_i$'s 
    ($i \in \lBrack 1,n-1 \rBrack$) are "internal", and there exists $i \in \lBrack 1,n-1 \rBrack$
    and $j\in \lBrack 0,n\rBrack$ "st" $x_i = x_j$. Up to renaming the variables, we get 
    that $i<j$ and $(i,j) \neq (0,n)$. 
    Recall that all $x_k$'s ($k \in \lBrack 1,n-1 \rBrack$) are "internal",
    and so from $x_i = x_j$ we get by trivial induction that $x_0 = x_{j-i}$.
    In particular, $x_0$ must be "internal" and so $x_0 = x_n$.
    Letting $k \defeq j-i$ we have $k \in \lBrack 1,n-1 \rBrack$ "st" 
    $x_k = x_0 = x_n$, and so we are in case 3.

    We can now proceed with the induction.
    \begin{enumerate}
      \item For the base case, we have a path where each $x_i$ ($i \in \lBrack 1,n-1 \rBrack$)
        is both "internal" and distinct from all $x_j$'s ($j\in \lBrack 0,n\rBrack$).
        By definition this path is a "segment".
      \item In the second case, $x_0$ and $x_n$ are "external", and there exists
        some $i \in \lBrack 1,n-1 \rBrack$
        "st" $x_i$ is "external". We use the induction hypothesis on the paths
        $x_0 \atom{a_1} \cdots \atom{a_i} x_i$ and $x_i \atom{a_{i+1}} \cdots \atom{a_n} x_n$
        and the conclusion follows.
      \item In the last case, there exists $k \in \lBrack 1,n-1 \rBrack$
          "st" $x_0 = x_k = x_n$, and the conclusion follows from applying the
          induction on $x_0 \atom{a_1} \cdots \atom{a_k} x_k$ and
          $x_k \atom{a_{k+1}} \cdots \atom{a_n} x_n$.
    \end{enumerate}
    This concludes the induction and the proof.
  \end{proof}

  \begin{lemma}
    \AP\label{lemma:hom-segments}
    Let $\gamma,\delta$ be "CRPQs".
    If $f\colon \delta \to \gamma$ is a "homomorphism" that sends "external variables"
    of $\delta$ on "external variables" of $\gamma$, then there is a
    function from nodes of $\seggraph(\delta)$ to nodes of $\seggraph(\gamma)$  
    that sends an edge of $\seggraph(\delta)$ to a path of $\seggraph(\gamma)$.
  \end{lemma}

  \begin{proof}
    By \Cref{prop:union-segments}, the image
    $f[\sigma]$ by $f$ of any "segment" $\sigma$ of $\delta$ is a union of
    "segments" of $\gamma$.
  \end{proof}

  \begin{corollary}
    \AP\label{coro:embedding-segments}
    Let $\gamma, \delta$ be two "CRPQs" such that there is an "embedding" from $\delta$ to $\gamma$
    "st" every node of in-degree 0 (resp. out-degree 0) is sent on a node of in-degree 0
    (resp. out-degree 0). 
    Then $\seggraph(\delta)$ is a "minor" of $\seggraph(\gamma)$.
  \end{corollary}

  \begin{proof}
    Let $f\colon \delta \to \gamma$ be such an "embedding".
    \begin{claim}
      \AP\label{claim:hom-preserves-non-internality}
      If $x \in \delta$ is "external", then $f(x)$ is "external".
    \end{claim}
    Indeed, if $x$ has in-degree at least 2, or out-degree at least 2,
    so does $f(x)$ since $f$ is an "embedding".
    Otherwise, either $x$ has either in-degree 0 or out-degree 0,
    and so does $f(x)$ using the assumption we made on $f$.

    We then use \Cref{lemma:hom-segments}.
    Note that since $f$ is an "embedding", the segments of $\gamma$ occurring in
    $f[\sigma]$ must be distinct from the segments occurring in $f[\sigma']$
    for any "segment" $\sigma' \neq \sigma$.
    Overall, we have an injective map from nodes of $\seggraph(\delta)$ 
    to nodes of $\seggraph(\gamma)$, which sends an edge to a path,
    and moreover these paths are pairwise disjoint.
    This shows that $\seggraph(\delta)$ is a "minor" of $\seggraph(\gamma)$.
  \end{proof}

\end{toappendix}

\begin{toappendix}
\begin{remark}
  \AP\label{rk:structure-theorem}
  Note that the assumption that $\anexpansion$ is "hom-minimal" in
  \Cref{thm:structure-theorem} is necessary: consider the "CRPQ"
  $\gamma(x,y) \defeq x \atom{a^+} y \land x \atom{(aa)^+} y$.
  For $n,m \in \Np$, let $\anexpansion_{n,m}(x,y) \defeq x \atom{a^n} y \land x \atom{a^{2m}} y$.
  There are two cases:
  \begin{enumerate}
    \item \AP\label{rk:structure-theorem:1} If $n \neq 2m$, then $\nbseg{\core(\anexpansion_{n,m})} = 2$
      but $\anexpansion_{n,m}$ is not "hom-minimal" since $\anexpansion_{2m,m} \homto \anexpansion_{n,m}$ but $\anexpansion_{n,m} \nothomto \anexpansion_{2m,m}$.
    \item If $n = 2m$, then $\nbseg{\core(\anexpansion_{n,m})} = 1$
    and $\anexpansion_{n,m}$ is "hom-minimal".
  \end{enumerate}
  Hence, using \Cref{thm:structure-theorem}, we can only get a lower bound of
  one "atom" (and not two) on the size of any "UCRPQ" "equivalent" to $\gamma(x,y)$,
  which is consistent with the fact that $\gamma(x,y) \semequiv \gamma'(x,y)$
  where $\gamma'(x,y) \defeq x \atom{(aa)^+} y$.
  If \Cref{thm:structure-theorem} would allow for non "hom-minimal" queries we would obtain, by \Cref{rk:structure-theorem:1} a lower bound of 2 "atoms", which is false.
\end{remark}
\end{toappendix}

However, this is a sound but unsurprisingly not a complete characterization of "minimality".

\begin{proposition}
  There are "minimal" "(U)CRPQs@CRPQ" which are not "strongly minimal". 
\end{proposition}

\begin{proof}
  The "Boolean" "CRPQ" $\gamma() = x \atom{a^+} x$. Since it has one atom, it must be "minimal", but it has no "hom-minimal" "expansions".
\end{proof}

Finally, we show that checking strong minimality is at least as hard as checking containment.

\begin{propositionrep}
  \label{prop:lowerbound-strong-minimality}
  Testing whether a "(U)CRPQ@UCRPQ" is "strongly minimal" is "ExpSpace"-hard.
\end{propositionrep}
\begin{proof}
  We use the same reduction as in \Cref{prop:lowerbound-non-redundant}.
  If $\delta$ is "strongly minimal", then it is "minimal", and so it is
  "non-rendundant", and hence
  by the proof of \Cref{prop:lowerbound-non-redundant},
  $x' \atom{K} y' \not\contained \bigwedge_i x \atom{L_i} y$.
  
  We prove the converse implication: assume that
    $x' \atom{K} y' \not\contained \bigwedge_i x \atom{L_i} y$.
  Then there exists $\#u\# \in K$ "st" $x' \atom{\#u\#} y'$
  does not "satisfy" $\bigwedge_i x \atom{L_i} y$.
  Consider the "expansion" $\anexpansion$ obtained by replacing $K$ with $\#u\#$, and replacing 
  $L_i$ with $@_i$ for each $i$.
  Then $\anexpansion$ is a "core" because of the fresh letters $\#$ and $@_i$.
  Moreover, we claim that it has to be "hom-minimal".
  Indeed, assume that some other "expansion" $\anexpansion'$ is "st" $\anexpansion'\homto \anexpansion$.
  Then because of $\#$, the "atom expansion" of $x' \atom{K} y'$ in $\anexpansion'$ 
  must be mapped on $x' \atom{\# u \#} y'$ in $\anexpansion$.
  Then, by definition of $u$, the "atom expansions" of $\bigwedge_i x \atom{L_i} y$
  cannot be mapped on $x' \atom{\# u \#} y'$, and so they must be mapped on
  $\bigwedge_i x \atom{@_i} y$, and hence $\anexpansion' = \anexpansion$.
  Hence, this shows that $\delta$ is "strongly minimal".
  Overall, we showed that $x' \atom{K} y' \not\contained \bigwedge_i x \atom{L_i} y$
  "iff" $\delta$ is "strongly minimal", which concludes our reduction.
\end{proof}

\section{An Upper Bound for Minimization of CRPQs}
\AP\label{sec:upperCRPQ}
In this section we show that the "minimization problem" for "CRPQs" is decidable, in particular, it belongs to the class "2ExpSpace".

\begin{theorem}
\AP\label{thm:2expspace-min-crpqs}
The "minimization problem" for "CRPQs" is in "2ExpSpace".
\end{theorem}

The proof of \Cref{thm:2expspace-min-crpqs} is based on a key lemma stated below. Intuitively, this lemma tells us that if a "CRPQ" $\gamma$ is "equivalent" to another "CRPQ" $\alpha$, then $\gamma$ is also equivalent to a "CRPQ" $\beta$, where $\beta$ has the same “shape” than $\alpha$ but the sizes of the NFAs appearing in $\beta$ are bounded by the size of $\gamma$. In particular, if $\gamma$  is equivalent to a  "CRPQ" with at most $k$ atoms, then it is equivalent to a "CRPQ" with at most $k$ atoms and NFAs of bounded sizes, and hence the search space in the minimization problem becomes finite. A careful analysis of this idea yields our "2ExpSpace" upper bound.

\begin{lemmarep}
\AP\label{lemma:crpq-size-bound}
Let $\gamma$ and $\alpha$ be two "CRPQs" such that $\alpha \contained \gamma$. Then there exists a "CRPQ" $\beta$ satisfying $\alpha \contained \beta\contained \gamma$ such that:
\begin{enumerate}
\item The "underlying graphs" of $\alpha$ and $\beta$ coincide. %
\item The size of each NFA appearing in $\beta$ is bounded by $f(\size{\gamma})$, where $f$ is a double-exponential function. %
\end{enumerate}
\end{lemmarep}

\begin{proof}
\proofcase{Construction of $\beta$.}
The idea is to define the "CRPQ" $\beta$ as the "CRPQ" obtained from $\alpha$ by replacing each regular language $L$ by a suitable regular language $\tilde{L}$, which depends also on $\gamma$. 
Consider the following equivalence relation on $\A^*$, where $\A$ is the underlying alphabet. Given $u,v\in \A^*$, we write $u \intro*\weakequiv{\gamma} v$ if for every NFA $\+A$ appearing on $\gamma$, and pair of states $p,q$ of $\+A$:
 $$\text{$u$ is accepted by $\subaut{\+A}{p}{q}$ $\iff$  $v$ is accepted by $\subaut{\+A}{p}{q}$}.$$
 Recall that  $\subaut{\+A}{p}{q}$ denote the ""sublanguage"" of $\+A$ recognized  when considering $\set{p}$ as the set of initial states and $\set{q}$ as the set of final states.
 For $u\in \A^*$, we define its ""$\gamma$-type"" to be:
  \[\type{\gamma}{u} \defeq \{
  ([u_1]_{\weakequiv{\gamma}},\dots,[u_\ell]_{\weakequiv{\gamma}}) \mid \text{$\ell\leq (\nbvar{\gamma} +1)$, and $u_1,\dots,u_\ell\in \A^*$ satisfy that $u=u_1\cdots u_\ell$} \}.
  \]
The idea is that $\type{\gamma}{u}$ encodes all the possible ways $u$ can be broken into $\ell\leq (\nbvar{\gamma}+1)$ subwords: we are not interested in the particular subwords $u_i$, but in their equivalence classes $[u_i]_{\weakequiv{\gamma}}$ with respect to $\weakequiv{\gamma}$.

We define the sought "CRPQ" $\beta$ to be the "CRPQ" obtained from $\alpha$ by replacing each regular language $L$ by $\tilde{L}$, where:
  \[\tilde{L} \defeq \bigcup_{u\in L} \{z\in \A^*\mid \type{\gamma}{u}\subseteq \type{\gamma}{z}\}.
  \]
  
  \medskip
  
\proofcase{Correctness of the construction.} By definition, the "underlying graphs" of $\alpha$ and $\beta$ are the same, and hence condition (1) holds. It remains to verify  $\alpha \contained \beta\contained \gamma$  and condition (2).

\proofsubcase{$\alpha \contained \beta\contained \gamma$.} 
Note that $u\in \{z\in \A^*\mid \type{\gamma}{u}\subseteq \type{\gamma}{z}\}$ always holds, and hence $L\subseteq \tilde{L}$. It follows that $\alpha \contained \beta$. 

We check $\beta\contained \gamma$ using Proposition~\ref{prop:cont-char-exp-st}. Assume $\beta$ is of the form $\bigwedge_{j=1}^m x_j \atom{\tilde{L}_j} y_j$. Consider an "expansion" $\anexpansion_\beta\in \Exp(\beta)$ defined by replacing each atom $x_j \atom{\tilde{L}_j} y_j$ by the path $x_j \atom{z_j} y_j$, for  $z_j\in \tilde{L}_j$. There must exist $u_j\in L_j$ such that $\type{\gamma}{u_j}\subseteq \type{\gamma}{z_j}$. Let $\anexpansion_{\alpha}\in \Exp(\alpha)$ be the expansion of $\alpha=\bigwedge_{j=1}^m x_j \atom{L_j} y_j$ obtained from replacing each atom $x_j \atom{L_j} y_j$ by the path $x_j \atom{u_j} y_j$. As $\alpha \contained \gamma$, then there exists an expansion $\eta_\alpha\in \Exp(\gamma)$ such that $\eta_\alpha\homto \anexpansion_\alpha$. Assume $\gamma$ is of the form $\bigwedge_{i=1}^r t_i \atom{M_i} s_i$ and $\eta_\alpha$ is of the form $\bigwedge_{i=1}^r t_i \atom{w_i} s_i$, where $w_i\in M_i$. Let $f$ be a homomorphism from $\eta_\alpha$ to $\anexpansion_\alpha$. We can decompose each path $x_j \atom{u_j} y_j$ in $\anexpansion_\alpha$ into
	\[
		x_j \atom{u_{j,1}} h_{j,1} \atom{u_{j,2}} \cdots \atom{u_{j,\ell-1}} h_{j,\ell-1}\atom{u_{j,\ell}} y_j
	\]
	where each $h_{j,p}$ is the image via $f$ of some variable in $\vars(\gamma)=\{s_1,t_1,\dots,s_r,t_r\}$ and each path $x \atom{u_{j,p}} x’$ satisfies that all of its internal variables are not  images via $f$ of variables in $\vars(\gamma)$. Note that $\ell\leq  (\nbvar{\gamma} +1)$. We say that the paths of the form $x \atom{u_{j,p}} x’$ are the ""basic paths"" of  $\anexpansion_\alpha$ with respect to $f$. Since $\type{\gamma}{u_j}\subseteq \type{\gamma}{z_j}$, each path $x_j \atom{z_j} y_j$ in $\anexpansion_\beta$ can be decomposed into
	\[
		x_j \atom{z_{j,1}} g_{j,1} \atom{z_{j,2}} \cdots \atom{z_{j,\ell-1}} g_{j,\ell-1}\atom{z_{j,\ell}} y_j
	\]
	where $z_{j,p}\weakequiv{\gamma} u_{j,p}$. Again, we say that the paths $x \atom{z_{j,p}} x’$ are the "basic paths" of  $\anexpansion_\beta$ with respect to $f$. We conclude by showing that there is $\eta_\beta\in \Exp(\gamma)$ such that $\eta_\beta\homto \anexpansion_\beta$. Intuitively, $\anexpansion_\alpha$ and $\anexpansion_\beta$ have ``equivalent’’ decompositions in terms of "basic paths". Hence, it is possible to turn $\eta_\alpha$ into another expansion $\eta_\beta\in \Exp(\gamma)$ by replacing "basic paths" of $\anexpansion_\alpha$  by their corresponding "basic paths" in $\anexpansion_\beta$. By doing so, we indeed obtain an expansion $\eta_\beta$ of $\gamma$ such that $\eta_\beta\homto \anexpansion_\beta$.

	Formally, observe that each path $t_i \atom{w_i} s_i$ of expansion $\eta_\alpha$ is mapped via $f$ to path in $\anexpansion_\alpha$ that can be decomposed into
		\[
		k_0 \atom{w_{i,1}} k_1 \atom{w_{i,2}} \cdots \atom{w_{i,n-1}} k_{n-1}\atom{w_{i,n}} k_n
	\]
	where $w_i=w_{i,1}\cdots w_{i,n}$, $k_0=f(t_i)$, $k_n=f(s_i)$ and each path $k_{q-1} \atom{w_{i,q}} k_q$ is a "basic path" of $\anexpansion_\alpha$ w.r.t.\ $f$. Each path $k_{q-1} \atom{w_{i,q}} k_q$ has a corresponding "basic path" $k’_{q-1} \atom{w_{i,q}’} k_q’$ of $\anexpansion_\beta$ w.r.t.\ $f$ defined in the natural way: if $k_o\in \vars(\alpha)=\{x_1,y_1,\dots,x_m,y_m\}$, then $k_o’=k_o$; if $k_o=h_{j,p}$, then $k_o’=g_{j,p}$; and if $w_{i,q}= u_{j,p}$, then $w_{i,q}’= z_{j,p}$. In particular, we have that $w_{i,q} \weakequiv{\gamma} w_{i,q}’$ and that the following path belongs to $\anexpansion_\beta$:
	\[
		k_0’ \atom{w_{i,1}’} k_1’ \atom{w_{i,2}’} \cdots \atom{w_{i,n-1}’} k_{n-1}’\atom{w_{i,n}’} k_n’
	\]
	It follows that the CQ $\eta_{\beta}$ obtained from $\gamma$ by replacing each atom $t_i \atom{M_i} s_i$ by the path  $t_i \atom{w_i’} s_i$, where $w’_{i}=w_{i,1}’\cdots w_{i,n}’$ satisfies that $\eta_{\beta}\homto \anexpansion_{\beta}$. It remains to show that  $\eta_{\beta}$ is actually an expansion of  $\gamma$, that is, each $w’_i\in M_i$. Suppose $M_i$ is represented by the NFA $\+A_i$. Since $w_i\in M_i$ and $w_i=w_{i,1}\cdots w_{i,n}$, there is a sequence $e_0,\dots,e_n$ of states of $\+A_i$ such that $e_0$ is initial, $e_n$ is final, and $w_{i,q}$ is accepted by $\subaut{\+A_i}{e_{q-1}}{e_q}$. As $w_{i,q}’\weakequiv{\gamma} w_{i,q}$, we have that $w_{i,q}’$ is also accepted by $\subaut{\+A_i}{e_{q-1}}{e_q}$, and hence $w’_{i}=w_{i,1}’\cdots w_{i,n}’$ is accepted by $\+A_i$.
	
\medskip
\proofsubcase{Bounding the size of NFAs in $\beta$.} 
First note that for any equivalence class $C$ of $\weakequiv{\gamma}$ the language $\{z\in \A^*\mid [z]_{\weakequiv{\gamma}} = C\}$ is accepted by an NFA $\+A_{C}$ of single-exponential size. Indeed, the class $C$ can be described by a set of triples
	\[
		\+T_C\subseteq \+T=\{(\+A,p,q)\mid \text{$\+A$ is an NFA in $\gamma$, and $p,q$ are states of $\+A$}\}
	\]
in such a way that $v\in C$ iff $v$ is accepted by $\subaut{\+A}{p}{q}$, for every $(\+A,p,q)\in \+T_C$, and $v$ is accepted by $\subaut{\+A}{p}{q}^{\complement}$, for every $(\+A,p,q)\not\in \+T_C$, where $\+A^{\complement}$ denotes the complement NFA of $\+A$. Hence $\+A_{C}$ can be written as the following intersection of NFAs:
\[
		\+A_{C} = \bigcap_{(\+A,p,q)\in \+T_C} \subaut{\+A}{p}{q} \cap  \bigcap_{(\+A,p,q)\not\in \+T_C} \subaut{\+A}{p}{q}^{\complement} 
	\]
	
The number of states in $\subaut{\+A}{p}{q}^{\complement}$ is at most $2^{r}$ and $|\+T|\leq r^2 \nbatoms{\gamma}$, where $r$ is the maximum number of states in an NFA of $\gamma$. Hence the number of states of $\+A_{C}$ is at most $2^{r^3\nbatoms{\gamma}}$, that is, single-exponential on $\gamma$.

Fix $u\in\A^*$. We claim that $\{z\in \A^*\mid \type{\gamma}{u}\subseteq \type{\gamma}{z}\}$ can be accepted by an NFA $\+A_u$ of at most double-exponential size on $\gamma$. It is easy to see that for every tuple $\bar{c} =(C_1,\dots,C_\ell)\in \type{\gamma}{u}$, the language $\{z\in \A^*\mid \bar{c} \subseteq \type{\gamma}{z}\}$ can be described by an NFA $\+A_{\bar{c}}$ of single-exponential size: we guess a decomposition $z=z_1\cdots z_\ell$ of the input word $z$ and check that each $z_i$ is accepted by $\+A_{C_i}$. The number of states of $\+A_{\bar{c}}$ is $O(\sum_{i=1}^{\ell} s_i)$, where $s_i$ is the number of states of $\+A_{C_i}$. Since $\ell\leq (\nbvar{\gamma} +1)$, this is $O(\nbvar{\gamma}2^{r^3\nbatoms{\gamma}})$. Now, the NFA $\+A_u$ is simply the intersection of all the NFAs in $\{\+A_{\bar{c}}\mid \bar{c}\in \type{\gamma}{u}\}$. The number of possible equivalence classes of $\weakequiv{\gamma}$ is at most $2^{|\+T|}\leq 2^{r^2 \nbatoms{\gamma}}$, and then the number of possible tuples of the form $\bar{c}=(C_1,\dots,C_\ell)$ is $N=O(2^{r^2 \nbatoms{\gamma}(\nbvar{\gamma} +1)})$. It follows that the number of states of $\+A_u$ is at most $O(\nbvar{\gamma}2^{r^3\nbatoms{\gamma} N})$, i.e., double-exponential on $\gamma$.

Finally, note that the number of possible $\type{\gamma}{u}$ is at most double-exponential on $\gamma$, more precisely, $2^N$. We conclude that every language $\tilde{L}$ in $\beta$ can be represented by an NFA $\+A_{\tilde{L}}$ which is the union of at most double-exponential many NFAs of the form $\+A_u$, each of these having at most double-exponential many states. We conclude that the size of $\+A_{\tilde{L}}$ can be bounded by a double-exponential function $f(\size{\gamma})$. 
\end{proof}

\begin{proofsketch}
The idea is to define the "CRPQ" $\beta$ as the "CRPQ" obtained from $\alpha$ by replacing each regular language $L$ by a suitable regular language $\tilde{L}$, which depends also on $\gamma$. Consider the following equivalence relation on $\A^*$, where $\A$ is the underlying alphabet. Given $u,v\in \A^*$, we write $u \intro*\weakequiv{\gamma} v$ if for every NFA $\+A$ appearing on $\gamma$, and pair of states $p,q$ of $\+A$:
 $$\text{$u$ is accepted by $\subaut{\+A}{p}{q}$ $\iff$  $v$ is accepted by $\subaut{\+A}{p}{q}$}.$$
 Recall that  $\subaut{\+A}{p}{q}$ denote the "sublanguage" of $\+A$ recognized  when considering $\set{p}$ as the set of initial states and $\set{q}$ as the set of final states. %
 For $u\in \A^*$, we define its \AP""$\gamma$-type"" to be:
  \[\AP\intro*\type{\gamma}{u} \defeq \{
  ([u_1]_{\weakequiv{\gamma}},\dots,[u_\ell]_{\weakequiv{\gamma}}) \mid \text{$\ell\leq (\nbvar{\gamma} +1)$, and $u_1,\dots,u_\ell\in \A^*$ satisfy that $u=u_1\cdots u_\ell$} \}.
  \]
The idea is that $\type{\gamma}{u}$ encodes all the possible ways $u$ can be broken into $\ell\leq (\nbvar{\gamma}+1)$ subwords: we are not interested in the particular subwords $u_i$, but in their equivalence classes $[u_i]_{\weakequiv{\gamma}}$ with respect to $\weakequiv{\gamma}$.

We define the sought "CRPQ" $\beta$ to be the "CRPQ" obtained from $\alpha$ by replacing each regular language $L$ by $\tilde{L}$, where:
  \[\tilde{L} \defeq \bigcup_{u\in L} \{z\in \A^*\mid \type{\gamma}{u}\subseteq \type{\gamma}{z}\}.
  \]
  
  \medskip
  
  By definition, the "underlying graphs" of $\alpha$ and $\beta$ are the same, and hence condition (1) holds. It remains to verify  $\alpha \contained \beta\contained \gamma$  and condition (2). Note that $u\in \{z\in \A^*\mid \type{\gamma}{u}\subseteq \type{\gamma}{z}\}$ always holds, and hence $L\subseteq \tilde{L}$. It follows that $\alpha \contained \beta$. 

Showing $\beta\contained \gamma$ is more involved. By Proposition~\ref{prop:cont-char-exp-st}, we need to prove that for every "expansion" $\anexpansion_\beta\in \Exp(\beta)$, there exists an "expansion" $\eta_\beta\in \Exp(\gamma)$, such that $\eta_\beta\homto \anexpansion_\beta$. Assume that $\beta$ and $\gamma$ are of the form $\bigwedge_{j=1}^m x_j \atom{\tilde{L}_j} y_j$ and $\bigwedge_{i=1}^r t_i \atom{M_i} s_i$, respectively. In particular, $\alpha$ must be of the form $\bigwedge_{j=1}^m x_j \atom{L_j} y_j$. Suppose $\anexpansion_\beta$ is of the form  $\bigwedge_{j=1}^m x_j \atom{z_j} y_j$, where $z_j\in \tilde{L}_j$. By construction, $\anexpansion_\beta$ has a corresponding expansion $\anexpansion_\alpha$ of $\alpha$: since each $z_j\in \tilde{L}_j$, there must be a $u_j\in L_j$ such that $\type{\gamma}{u_j}\subseteq \type{\gamma}{z_j}$, and then we can take $\anexpansion_\alpha$ as $\bigwedge_{j=1}^m x_j \atom{u_j} y_j$. In turn, since $\alpha\contained \gamma$, there exists an "expansion" $\eta_\alpha\in \Exp(\gamma)$ such that $\eta_\alpha\homto \anexpansion_\alpha$ via a homomorphism $f$. 

The idea is to modify $\eta_\alpha$ into another expansion $\eta_\beta$ with  $\eta_\beta\to \anexpansion_\beta$, as desired. Note that $f$ maps the external variables of $\eta_\alpha$ to external or internal variables in $\anexpansion_\alpha$. This determines a subdivision for each path $x_j \atom{u_j} y_j$ of $\anexpansion_\alpha$ into smaller or `basic paths’, whose endpoints correspond to external variables of $\xi_\alpha$ or images of the external variables  of $\eta_\alpha$ via $f$. The number of these paths is hence bounded by $\nbvar{\gamma} +1$. Since $\type{\gamma}{u_j}\subseteq \type{\gamma}{z_j}$, then each path $x_j \atom{z_j} y_j$ in $\anexpansion_\beta$ can also be subdivided in an equivalent way than  $x_j \atom{u_j} y_j$. Overall, the decomposition of $\anexpansion_\alpha$  into basic paths can be `simulated’ in $\anexpansion_\beta$. This gives us a homomorphism from $\eta_\beta$ to $\anexpansion_\beta$, where $\eta_\beta$ is obtained from $\eta_\alpha$ in the following way: for each path $t_i \atom{w_i} s_i$ in $\eta_\alpha$, where $w_i\in M_i$, the image of the path via $f$ induces a subdivision of $t_i \atom{w_i} s_i$ into basic paths  of $\anexpansion_\alpha$. We can replace each of these basic paths by its equivalent basic path in $\anexpansion_\beta$. As the label of these paths are equivalent w.r.t to the relation $\weakequiv{\gamma}$, membership in $M_i$ is maintained after this transformation. Hence $\eta_\beta$ is indeed an expansion of $\gamma$.

For condition (2), it is easy to see that every equivalence class $C$ of the relation $\weakequiv{\gamma}$ can be accepted by an NFA $\+A_C$ of single-exponential size on $\size{\gamma}$. Also, for each word $u\in \A^*$, and  each tuple $\bar{c} =(C_1,\dots,C_\ell)\in \type{\gamma}{u}$, there is a single-exponential size NFA $\+A_{\bar{c}}$  accepting the language $\{z\in \A^*\mid \bar{c} \in \type{\gamma}{z}\}$. By intersecting these NFAs, we obtain an NFA $\+A_{u}$ accepting the language $\{z\in \A^*\mid \type{\gamma}{u}\subseteq \type{\gamma}{z}\}$. It is easy to see that the number of tuples of the form $(C_1,\dots,C_\ell)$ is at most single-exponential, and then the size of $\+A_{u}$ is at most double-exponential. It follows that  the number of possible  $\type{\gamma}{u}$ is at most double-exponential, and hence the languages  $\tilde{L}$ in $\beta$ can be described by a union of double-exponential many NFAs, each of size at most double-exponential. Overall, each $\tilde{L}$ can be described by an NFA of at most double-exponential size on $\size{\gamma}$.
 \end{proofsketch}

\begin{proof}[Proof of \Cref{thm:2expspace-min-crpqs}]
The "2ExpSpace" algorithm proceeds as follows. Let $\gamma$ be a "CRPQ"  over $\A$ and $k\in\N$. Let $f$ be the double-exponential function from \Cref{lemma:crpq-size-bound}. We enumerate all possible directed multigraphs with at most $k$ edges. For each of these graphs, we enumerate all the possible "CRPQs" $\beta$ obtained by labelling each edge with an NFA of size bounded by $f(\size{\gamma})$. If for some of these $\beta$, we have that $\beta\semequiv \gamma$ then we accept the instance, otherwise we reject it. \Cref{lemma:crpq-size-bound} ensures that this algorithm is correct.

It remains to show that  $\beta\semequiv \gamma$ can be carried out in "2ExpSpace". We use a proposition from \cite[Proposition 3.11, p. 17]{FM2023semantic} stating that containment $\Gamma\contained \Delta$, for "UCRPQs" $\Gamma$ and  $\Delta$, can be solved in non-deterministic space $O(\size{\Gamma} + \size{\Delta}^{c\cdot \nbatoms{\Delta}})$, where $c$ is a constant. This implies that $\beta\contained \gamma$ can be checked in space $O(\size{\beta} + \size{\gamma}^{c\cdot \nbatoms{\gamma}})$, and then within "2ExpSpace". The containment $\gamma\contained \beta$ can be solved in space $O(\size{\gamma} + \size{\beta}^{c\cdot k})$, and hence also in "2ExpSpace".
\end{proof}

\section{Minimization of UCRPQs via Approximations}
\label{sec:ucrpqs}

\AP\label{sec:upperUCRPQ}
We now focus on minimizations by \emph{finite unions} of "CRPQs". This is a different problem than the one seen in the previous section, "ie", there is no obvious reduction in either direction between the "minimization problem" for "CRPQs" and the "minimization problem" for "UCRPQs", and indeed we will solve this problem using an altogether different approach.

\subsection{Unions Allow Further Minimization}

As it turns out, having unions may help in minimizing the (maximum) number of atoms of a query as the next proposition shows.
\begin{propositionrep}
  \AP\label{prop:unionsmatter}
  There exist "CRPQs" which are "minimal" among "CRPQs" but not among "UCRPQs".
\end{propositionrep}
\begin{proof}
  The following example is inspired from \cite[Example 1.2]{FigueiraM23}.
  Consider the following "CRPQs":\leavevmode
  \begin{center}
    \includegraphics{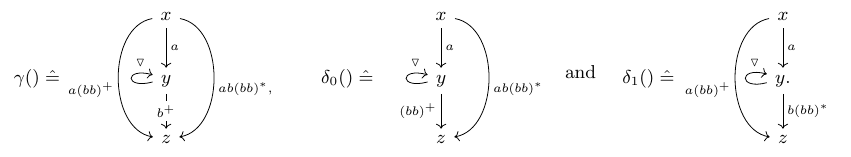}
  \end{center}
  It is easy to see that $\gamma \semequiv \delta_0 \lor \delta_1$ 
  by doing a case disjunction on the parity of the path between $y$ and $z$---the even case
  is handled by $\delta_0$, the odd case by $\delta_1$.
  Hence, $\gamma$ is not "minimal" among "UCRPQs".
  
  We want to show that $\gamma$ is "minimal" among "CRPQs".
  Let $\zeta()$ be a "CRPQ" that is "equivalent" to $\gamma$, and assume by contradiction
  that it has at most four "atoms".
  Given natural numbers $l,m,r \in \Np$ where $l$ is even and $r$ is odd,
  consider the "expansion"
  \begin{center}
    \includegraphics{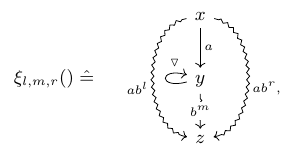}
  \end{center}
  where squiggly arrows represent paths of atoms.
  \begin{claim}
    \AP\label{claim:example-union-matter-1}
    The "expansion" $\anexpansion_{l,m,r}$ is "hom-minimal" "iff" $l = m$ or $m=r$.
  \end{claim}
  Indeed, if $l \neq m$ and $m \neq r$, then
  if $m$ is even then $\anexpansion_{m,m,r} \homto \anexpansion_{l,m,r}$
  but $\anexpansion_{l,m,r} \nothomto \anexpansion_{m,m,r}$ (since $l \neq m$)
  and dually if $m$ is odd then $\anexpansion_{l,m,m} \homto \anexpansion_{l,m,r}$
  but $\anexpansion_{l,m,r} \nothomto \anexpansion_{l,m,m}$ (since $m \neq r$).
  In both cases, $\anexpansion_{l,m,r}$ is not "hom-minimal".
  
  Conversely, if $l = m$ or $m=r$, assume "wlog", by symmetry, that $l=m$.
  Let $l',m',r' \in \Np$ "st" $l'$ is even and $r'$ is odd,
  and assume that there is a "homomorphism" $f\colon \anexpansion_{l',m',r'} \to \anexpansion_{l,m,r}$. Because of the $\marking$-self-loop, $f(y) = y$ and hence $f(x) = x$. It then follows that
  $f(z) = z$ because we must have both an $a(bb)^+$- and an $ab(bb)^*$-path
  from $f(x)$ and $f(z)$. 
  Moreover, we must have $m=m'$,
  and since $l'$ is even, we must have $l' = l = m$,
  and dually, since $r'$ is odd, we must have $r' = r$.
  It follows that $\tup{l',m',r'} = \tup{l,m,r}$, and hence $\anexpansion_{l,m,r}$ is "hom-minimal".

  \begin{claim}
    \AP\label{claim:example-union-matter-2}
    If $m=r$, then $\core(\anexpansion_{l,m,r})$ equals
    \begin{center}
      \includegraphics{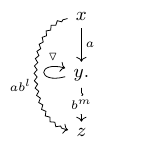}
    \end{center}
  \end{claim}
  Putting \Cref{claim:example-union-matter-1,claim:example-union-matter-2} and
  \Cref{thm:structure-theorem} together, we get that 
  $\seggraph(\core(\anexpansion_{2,1,1}))$ is a "minor" of $\zeta$.
  Since $\zeta$ was assumed to have at most four "atoms",
  it follows that $\zeta()$ must be exactly of the form
  \begin{center}
    \includegraphics{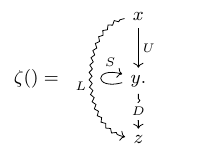}
  \end{center}
  Using the containment $\zeta \contained \gamma$,
  since the only directed simple cycle in an "expansion" of $\gamma$ is labelled
  by $\marking$, we get $S = \{\marking\}$.
  Similarly, using again that $\zeta \contained \gamma$ and that any "expansion" of $\gamma$
  must have both an $a(bb)^+$- and an $ab(bb)^*$-path, it can be shown that $U = \{a\}$.

  Now let $l,m,r \in \Np$ with $l$ even and $r$ odd.
  From $\gamma \contained \zeta$, we get that each $\anexpansion_{l,m,r}$ "satisfies"
  $\zeta$ and so there is an "evaluation map" $f\colon \zeta \to \anexpansion_{l,m,r}$.
  Clearly $f(y) = y$ and then using the
  $a(bb)^+$- and $ab(bb)^*$-paths, we get that $f(x) = x$ and $f(z) = z$. 
  It then follows that
  \begin{itemize}
    \item $b^m \in D$, and
    \item either $ab^l \in L$, or $ab^r \in L$, or $a \marking^k b^m$ for some $k\in\N$
  \end{itemize}
  From the first point we get $b^+ \subseteq D$, and in fact using $\zeta\contained \gamma$,
  $D = b^+$.

  Now let $w \in L$ and $m \in \N$, and consider the "expansion"
  \begin{center}
  \includegraphics{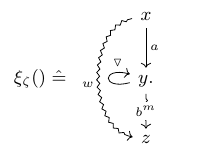}
  \end{center}
  of $\zeta$. Since $\zeta \contained \gamma$, there is
  an "evaluation map" $f\colon \gamma \to \anexpansion_\zeta$.
  Taking $m$ to be any even value, we get that $w$ must be in $ab(bb)^*$,
  say $w = ab^r$ for some odd $r \in \Np$.
  Taking $m$ to be any odd value, we get dually that $w$ must be in $a(bb)^+$,
  say $w = ab^l$ for some even $l \in \Np$.
  Hence, we get $\{ab^l, ab^r\} \subseteq L$.

  We can now pick $m$ to be even and
  $w = ab^l$, and then $\anexpansion_\zeta$ should be a model of $\gamma$.
  But $\anexpansion_\zeta$ does cannot satisfy the atom $x \atom{ab(bb)^*} z$.
  Contradiction. Hence, $\gamma$ cannot be equivalent to a single "CRPQ"
  with at most four "atoms".
\end{proof}

\begin{proofsketch}
  The following example is inspired from \cite[Example 1.2]{FigueiraM23}.
  Consider these "CRPQs":\leavevmode
  \begin{center}
    \includegraphics{figs/8.pdf}
  \end{center}
  It can be seen that $\gamma \semequiv \delta_0 \lor \delta_1$ 
  by doing a case analysis on the parity of the path between $y$ and $z$---the even case
  is handled by $\delta_0$, the odd case by $\delta_1$.
  Hence, $\gamma$ is not "minimal" among "UCRPQs".

  The hard part is to argue that $\gamma$ is "minimal" among
  "CRPQs": we consider a "CRPQ" $\zeta$ that is equivalent to $\gamma$ and,
  by contradiction we assume that it has at most four "atoms".
  Using the "Semantical Structure" theorem, we deduce that the following graph must be "minor"
  of the "underlying graph" of $\zeta$. Since $\zeta$ has at most four "atoms",
  we conclude that $\zeta$ must exactly be of the following shape 
  \begin{center}
      \includegraphics{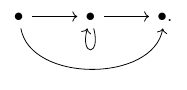}
    \end{center}
    We then use the equivalence with $\delta$ to reach a contradiction.
\end{proofsketch}

\subsection{Maximal Under-Approximations}

\AP\label{sec:maximal-under-approxiamtions}
Our approach exploits having unions at our disposal, enabling the possibility of defining and computing "maximal under-approximations" for "UCRPQs" having some given underlying shape. This will lead to an "ExpSpace" upper bound for "UCRPQ" minimization.

For a graph class $\+C$ (remember that these are directed multigraphs) we denote by \AP$\intro*\CQ[\+C]$, $\intro*\CRPQ[\+C]$, $\intro*\UCRPQ[\+C]$, $\intro*\infUCRPQ[\+C]$, the class of "CQs", of "CRPQs", of finite unions of "CRPQs" 
and of infinite unions of "CRPQs" $\gamma$ such that $\underlying{\gamma} \in \+C$.
Given a graph class $\+C$ and a "UCRPQ" $\Gamma$, we define $\AppInf{\Gamma}{\+C}$ to be the (infinite) set of all $\CRPQ[\+C]$ queries which are "contractions" of "CQs" "contained" in $\Gamma$. More formally:\footnote{Note that the "contraction" of a "CQ" is a "CRPQ": this is why $\AppInf{\Gamma}{\+C}$ is an infinite union of "CRPQs" and not of "CQs".}
\[
  \intro*\AppInf{\Gamma}{\+C} \defeq \{
    \alpha \in \CRPQ[\+C] \mid \exists \xi \in \Exp(\Gamma), \exists \eta \in \CQ \text{ s.t. }
    \xi \homto \eta \text{ and } \alpha \text{ is a "contraction" of } \eta
  \}.
\]
\begin{proposition}
  \AP\label{prop:max-under-approx}
  If $\+C$ is a graph class closed under taking "minors", then
  $\AppInf{\Gamma}{\+C}$ is a \AP""maximal under-approximation"" of $\Gamma$ by $\infUCRPQ[\+C]$ queries, in the sense that:
  \begin{enumerate}
    \item $\AppInf{\Gamma}{\+C} \in \infUCRPQ[\+C]$,
    \item $\AppInf{\Gamma}{\+C} \contained \Gamma$, and
    \item for any $\Delta \in \infUCRPQ[\+C]$, if $\Delta \contained \Gamma$, then $\Delta \contained \AppInf{\Gamma}{\+C}$.
  \end{enumerate}
\end{proposition}

\begin{proof}
  The first point is trivial. For the second one, observe that
  if $\alpha \in \CRPQ[\+C]$ is "st" there exist $\xi \in \Exp(\Gamma)$ and
  $\eta \in \CQ$ "st" $\xi \homto \eta$ and $\alpha$ is a "contraction" of $\eta$,
  then $\eta \contained \xi$ and $\eta$ is "semantically equivalent" to $\alpha$ by \Cref{fact:produce-fully-contracted}, and hence $\alpha \contained \xi$, from which it follows that $\AppInf{\Gamma}{\+C} \contained \Gamma$.

  For the third point, let $\Delta \in \infUCRPQ[\+C]$ "st" $\Delta \contained \Gamma$. Let $\delta \in \Delta$ and $\xi_{\delta}$ be an "expansion" of $\delta$. Since $\Delta \contained \Gamma$, there exist $\gamma \in \Gamma$ and $\xi_{\gamma} \in \Exp(\gamma)$ "st" $\xi_{\gamma} \homto \xi_{\delta}$. From the fact that $\underlying{\delta} \in \+C$ and that $\xi_{\delta} \in \Exp(\delta)$, it follows that there exists a "contraction" $\alpha_{\delta}$ of $\xi_{\delta}$ "st" $\underlying{\alpha_{\delta}}$ is a "minor" of $\underlying{\delta}$, and thus belongs to $\+C$.
  Hence, it follows that $\alpha_{\delta} \in \AppInf{\Gamma}{\+C}$.
  So $\xi_{\delta} \contained \AppInf{\Gamma}{\+C}$ and thus
  $\Delta \contained \AppInf{\Gamma}{\+C}$.
\end{proof}

\begin{toappendix}
\begin{remark}
  If $\+C$ is closed under taking "subgraphs" and "subdivisions", then
  \[\AppInf{\gamma}{\+C} \semequiv \{
    \alpha \in \CQ[\+C] \mid \exists \xi \in \Exp(\gamma), 
    \xi \surjto \alpha
  \}.\]
\end{remark}
\end{toappendix}

Our "ExpSpace" upper bound relies on the following technical lemma:

\begin{lemma}
  \AP\label{lemma:approximation-for-finclass}
  If $\+C$ is finite and closed under taking "minors", then $\AppInf{\Gamma}{\+C}$ is equivalent to a query $\Delta\in \UCRPQ[\+C]$ 
with exponentially many "CRPQs", each "CRPQ" of $\Delta$ being of size polynomial in $\size{\Gamma} + \max{\set{\nbatoms{G} \mid G \in \+C}}$.  %
Further, the membership $\delta \in^? \Delta$ can be tested "NP". 
In particular, this "UCRPQ" $\Delta$ can be computed in exponential time from $\Gamma$.
\end{lemma}

\begin{proof}
We start by observing that $\AppInf{\Gamma}{\+C}$ admits an equivalent and more flexible definition in terms of "refinements". This definition will allow us to effectively compute our desired equivalent query $\Delta$, by considering "refinements" of bounded lengths.

  For each $m \in \N\cup\{+\infty\}$, we define the $\infUCRPQ[\+C]$
  \[\AP\intro*\App{\Gamma}{\+C}{m} \defeq \{
    \alpha \in \CRPQ[\+C] \mid \exists \rho \in \Refin[\leq m](\Gamma), 
    \exists \eta \in \CRPQ \text{ s.t. }
    \rho \homto \eta \text{ and } \alpha \text{ is a "contraction" of } \eta
  \}.\]

We have that $\AppInf{\Gamma}{\+C}\semequiv \App{\Gamma}{\+C}{+\infty}$. Indeed, $\AppInf{\Gamma}{\+C}\contained \App{\Gamma}{\+C}{+\infty}$ as the former is a subset of the latter. On the other hand, $\App{\Gamma}{\+C}{+\infty}\contained \AppInf{\Gamma}{\+C}$ follows from Proposition \ref{prop:max-under-approx} and the fact that $\App{\Gamma}{\+C}{+\infty}\contained \Gamma$ by construction.  
  
Let $\nbatoms{\+C}\defeq \max{\set{\nbatoms{G} \mid G \in \+C}}$ and $r_{\Gamma}$ be the maximum number of states over the NFAs describing the languages appearing in $\Gamma$. 
  We claim that $\App{\Gamma}{\+C}{+\infty} \semequiv \App{\Gamma}{\+C}{O(\nbatoms{\Gamma}\cdot r_{\Gamma}\cdot \nbatoms{\+C})}$.
  The right-to-left containment is trivial, so it suffices to show
$\App{\Gamma}{\+C}{+\infty} \contained \App{\Gamma}{\+C}{O(\nbatoms{\Gamma}\cdot r_{\Gamma}\cdot \nbatoms{\+C})}$.

  We define an \AP""explicit approximation"" of $\Gamma$ over $\+C$ as a tuple
  $\vec{\alpha} = \tup{\rho, \eta, \alpha, h, \orig, \cont}$ consisting of the following:
  \begin{itemize}
    \item queries $\rho \in \Refin(\Gamma)$, $\eta \in \CRPQ$ and
      $\alpha \in \CRPQ[\+C]$,
    \item a "homomorphism" $h\colon \rho \homto \eta$,
    \item witnesses that $\alpha$ is a "contraction" of $\eta$, in the form of
      a function $\orig \colon \vars(\alpha) \to \vars(\eta)$,
      saying from which variable of $\eta$ a variable of $\alpha$ originates,
      and a function $\cont \colon \atoms(\eta) \to \atoms(\alpha)$
      saying onto which atom of $\alpha$ an atom of $\eta$ is contracted ("ie", the functions $\orig,\cont$ must meet the expected properties).
  \end{itemize}
  We say that an "explicit approximation" $\vec{\alpha_1}$ is \AP""contained@@approx"" in
  an "explicit approximation" $\vec{\alpha_2}$ if $\alpha_1 \contained \alpha_2$.
  For any $m$, to prove that $\App{\Gamma}{\+C}{+\infty} \contained \App{\Gamma}{\+C}{m}$,
  it suffices to prove that any "explicit approximation" $\vec{\alpha_1}=\tup{\rho_1, \eta_1, \alpha_1, h_1, \orig_1, \cont_1}$ is "contained@@approx"
  in an "explicit approximation" $\vec{\alpha_2}=\tup{\rho_2, \eta_2, \alpha_2, h_2, \orig_2, \cont_2}$ such that $\rho_2 \in \Refin[\leq m](\Gamma)$.

  \proofcase{1\textsuperscript{st} step: Bounding the size of $\eta$.}
  We define the \AP""contraction length"" of an "explicit approximation" $\vec{\alpha}$
  as the size of the longest path in $\eta$ whose atoms are all sent on the same
  atom of $\alpha$ via $\cont$. In symbols, this is $\max{\set{|\cont^{-1}(\mu)| : \mu \in \atoms(\alpha)}}$.
  We show that any "explicit approximation" $\vec{\alpha_1}$ is "contained@@approx"
  in an "explicit approximation" $\vec{\alpha_2}$ of bounded "contraction length".

  Consider a path $x_0 \atom{} x_1 \atom{} \cdots \atom{} x_k$ of $\eta_1$ whose
  atoms are all sent on the same atom of $\alpha_1$ via $\cont_1$. 
  If this path is very long, in particular greater than $\nbvar{\Gamma}$, it must contain an "internal variable" $x_i$ such that all of its $h_1$-preimages are "internal variables" of $\rho_1$. 
  Then we will be able to "contract" $x_i$ as well as the "internal variables" of the preimage, obtaining an "explicit approximation" which  contains  $\vec{\alpha_1}$ (details in \Cref{app:lemma:approximation-for-finclass}). 
Overall, this shows that any "explicit approximation" is "contained@@approx"  in an "explicit approximation" of "contraction length" at most $O(\nbvar{\Gamma})\leq O(\nbatoms{\Gamma})$ .%

  \begin{toappendix}
    \subsection{Missing details to the proof of \Cref{lemma:approximation-for-finclass}}
    \AP\label{app:lemma:approximation-for-finclass}
    Let $\vec{\alpha_1}$ be an "explicit approximation".
    Consider a path $x_0 \atom{} x_1 \atom{} \cdots \atom{} x_k$ of $\eta_1$ whose
    atoms are all sent on the same atom of $\alpha_1$ via $\cont_1$. 
    
Assume that for some $x_i$, with $i\in \lBrack 1,k-1\rBrack$, all variables in $h_1^{-1}(x_i)$ are "internal" in $\rho_1$. If $x_{i-1}\atom{L} x_i$ and $x_i\atom{L’} x_{i+1}$ are the only atoms containing $x_i$ in $\eta_1$, then for a variable $z\in h_1^{-1}(x_i)$, the only atoms in  containing $z$ in $\rho_1$ must have the form $w\atom{L} z$ and $z\atom{L’} w’$. Let $\eta_2$ be the query resulting from $\eta_1$ by contracting the internal variable $x_i$ and replacing $L\cdot L’$ by $K$, where $K$ is defined as follows. Since $L$ and $L’$ appear consecutively in internal paths of the refinement $\rho_1$, there must be an  NFA $\+A$ in $\gamma$ and three states $p,q,r$ such that $L=\subaut{\+A}{p}{q}$ or $L=\{a\}$ with $a\in \subaut{\+A}{p}{q}$, and   $L’=\subaut{\+A}{q}{r}$ or $L’=\{a\}$ with $a\in \subaut{\+A}{q}{r}$. We define $K=\subaut{\+A}{p}{r}$. Note that in any case, $L\cdot L’\subseteq K$. Similarly, define $\rho_2$ to be the query resulting from $\rho_1$ by contracting each internal variable $z\in h_1^{-1}(x_i)$ and replacing $L\cdot L’$ by $K$. Note that $\rho_2$ is still a refinement of $\gamma$ and that the homomorphism  $h_1\colon \rho_1 \to \eta_1$ induces a "homomorphism" $f_2\colon \rho_2 \to \eta_2$. Define $\alpha_2$ be the contraction of $\eta_2$ obtained by contracting all the remaining internal variables as  the contraction $\alpha_1$ is obtained from  $\rho_1$. Since $\alpha_1\contained \alpha_2$ as $L\cdot L’\subseteq K$, this defines an  "explicit approximation" $\vec{\alpha_2}$ that contains $\vec{\alpha_1}$. Note that in the case $h_1^{-1}(x_i)=\emptyset$, we can take $K=L\cdot L’$, and $\alpha_1\equiv \alpha_2$.

If $k-1>\nbvar{\Gamma}$, then path $x_0 \atom{} x_1 \atom{} \cdots \atom{} x_k$ contains a variable satisfying the condition above, and hence we can apply the simplification.  Overall, this shows that any "explicit approximation" is "contained@@approx"
  in an "explicit approximation" of "contraction length" at most
  $O(\nbvar{\Gamma})\leq O(\nbatoms{\Gamma})$.

\end{toappendix}

  \proofcase{2\textsuperscript{nd} step: bounding the size of $\rho$.}
  We now show that we can bound the \AP""refinement length"" of an "explicit approximation",
  namely the maximal length of an "atom refinement" in $\rho$.
  Let $\vec{\alpha_1}$ be an "explicit approximation" of "contraction length" at most
  $O(\nbatoms{\Gamma})$. Then $\eta_1$ has at most $O(\nbatoms{\Gamma}\cdot\nbatoms{\+C})$
  "atoms". It follows then, by the pigeonhole principle, that we can bound the
  "refinement length" of $\rho_1$ by $O(\nbatoms{\Gamma}\cdot r_{\Gamma}\cdot\nbatoms{\alpha})$. Indeed, if the length of an "atom refinement" of $\rho_1$  is greater than this bound, there are two  atoms in the refinement $x\atom{L_i} y$ and  $x’\atom{L_j} y’$, with $i<j$, mapped to the same atom via $h_1$ and whose corresponding NFA states $q_i$ and $q_j$ in the definition of refinements are the same. We can then remove the path between $y$ and $y’$.
  In conclusion, this shows that $\App{\Gamma}{\+C}{+\infty} \contained \App{\Gamma}{\+C}{O(\nbatoms{\Gamma}\cdot r_{\Gamma}\cdot \nbatoms{\+C})}$.

  \proofcase{Conclusion: Expressing \& computing $\App{\Gamma}{\+C}{O(\nbatoms{\Gamma}\cdot r_{\Gamma} \cdot \nbatoms{\+C})}$ as a "UCRPQ".} 
  In order to compute $\App{\Gamma}{\+C}{O(\nbatoms{\Gamma}\cdot r_{\Gamma} \cdot \nbatoms{\+C})}$ we can enumerate the finitely many $m$-refinements $\rho$ of $\Gamma$, where $m=O(\nbatoms{\Gamma}\cdot r_{\Gamma} \cdot \nbatoms{\+C})$, and the finitely many "CRPQs"  $\eta$ with at most $O(\nbatoms{\Gamma} \cdot \nbatoms{\+C})$ atoms such that $\rho\homto \eta$. The only issue here is that we have infinitely many possibilities to choose languages labelling the atoms that are not in the "homomorphic image" of $\rho \homto \eta$. However, we can  choose the most general language $\A^*$ obtaining a query equivalent to $\App{\Gamma}{\+C}{m}$. 
  Note that each "CRPQ" has at most $O(\nbatoms{\Gamma} \cdot \nbatoms{\+C})$ atoms and its languages are concatenations of
   $O(\nbatoms{\Gamma} \cdot \nbatoms{\+C})$ "sublanguages" of $\Gamma$ or $\A^*$, and
  so they can be described by NFAs of polynomial size on $\size{\Gamma}$ and $\nbatoms{\+C}$. %
\end{proof}

\begin{corollary}
  \label{coro:upperbound-ucrpqs}
  Testing whether a "UCRPQ" is "equivalent" to a "UCRPQ" of at most $k$ atoms is "ExpSpace"-complete.
\end{corollary}

\begin{proof}
 It suffices to test if the "UCRPQ" $\Gamma$ is "equivalent"
  to $\App{\Gamma}{\+C}{m}$ where $\+C$ is the class of all graphs
  with at most $k$ edges and $m=O(\nbatoms{\Gamma}\cdot r_{\Gamma} \cdot \nbatoms{\+C})$ as in the proof of \Cref{lemma:approximation-for-finclass}. 
  The correctness follows from \Cref{prop:max-under-approx} since $\+C$ is
  trivially closed under "minors".
  Each $\alpha \in \App{\Gamma}{\+C}{m}$ has at most $k$ edges,
  and $\App{\Gamma}{\+C}{m}$ contains exponentially many queries,
  so by \cite[Proposition 3.11]{FM2023semantic} (see also proof of \Cref{thm:2expspace-min-crpqs}), it can be solved
  in "ExpSpace".
  Finally, "ExpSpace"-hardness will follow from \Cref{thm:minimization-lowerbound} (\S\ref{sec:lowerbounds}).
\end{proof}

\subsection{CRPQs over Simple Regular Expressions}
\AP Let ""UCRPQ(SRE)"" ("resp" ""CRPQ(SRE)"") be the set of all "UCRPQs" ("resp" "CRPQs") whose languages are expressed via "SREs" (as defined in \Cref{sec:intro}). We show that if we restrict the regular expressions we obtain a much better complexity for the "minimization problem" for "UCRPQs".

\begin{theorem}\AP\label{thm:minimization-SRE}
  The "minimization problem" for "UCRPQ(SRE)" is "PiP2"-complete.\footnote{By this we mean the "minimization problem" for "UCRPQs" whose input instances are "UCRPQ(SRE)".}
\end{theorem}
\begin{proof}
  \changes{
    We first begin with an easy small counterexample property.
  }
  \begin{claim}\AP\label{cl:small-counterexample-SRE}
    \changes{
    Let $\Gamma, \Delta \in \text{"UCRPQ"}$ containing only atoms with expressions of the form \AP""(i)@@sre"" $a^+$, or ""(ii)@@sre"" $a_1 + \dotsb + a_k$. Additionally, $\Delta$ may also have expressions of the form ""(iii)@@sre"" $\A^*$. If $\Gamma \not\contained \Delta$, then there exists $\anexpansion \in \Exp(\Gamma)$ such that (a) $\anexpansion \not\contained \Delta$ and (b) $\nbatoms{\anexpansion} \leq O(\max_{\gamma \in \Gamma}\nbatoms{\gamma} \cdot \max_{\delta \in \Delta}\nbatoms{\delta})$.
    }
  \end{claim}
  \changes{
    The intuition is that if a counterexample includes an atom expansion $x \atom{a^n} y$ of some atom $x \atom{a^+} y$, where $n$ is greater than the maximum number of atoms in $\Delta$ (plus one), then the expansion obtained by replacing $x \atom{a^n} y$ with $x \atom{a^{n-1}} y$ must also be a counterexample. Hence, a minimal counterexample must have all atom expansions bounded by the maximum number of atoms in $\Delta$.
  }
  \begin{proof}
    \changes{
    This fact follows from a standard technique as used in, "eg", \cite{FigueiraGKMNT20}.
    Take any counterexample $\anexpansion \in \Exp(\Gamma)$ as in the statement, and suppose it is of minimal size. By means of contradiction, assume $\nbatoms{\anexpansion} > \max_{\gamma \in \Gamma}\nbatoms{\gamma} \cdot (\max_{\delta \in \Delta}\nbatoms{\delta}+1)$. 
    Then, it contains an "atom expansion" $x \atom{a^m} y$ of size $m>\max_{\delta \in \Delta}\nbatoms{\delta}+1$. 
    Consider removing one atom from such expansion ("ie", replacing $x \atom{a^m} y$ with $x \atom{a^{m-1}} y$), obtaining some expansion $\anexpansion' \in \Exp(\Gamma)$ of smaller size. By minimality $\anexpansion'$ is not a counterexample: in other words there is  $\anexpansion'' \in \Exp(\delta)$ such that $h:\anexpansion'' \homto \anexpansion'$ for some $\delta \in \Delta$ and $h$. Since $x \atom{a^{m-1}} y$ contains more than $\nbatoms{\delta}$ "atoms", there must be some $a$-"atom" of $x \atom{a^{m-1}} y$ which either (1) has no $h$-preimage or (2) every $h$-preimage is in an "atom expansion" of a $\A^*$-"atom" or a $a^+$-"atom" of $\delta$. We can then replace the $a$-atom with two $a$-atoms in $\anexpansion'$ and do similarly in the atom expansions of $\anexpansion''$ in the $h$-preimage, obtaining that $\anexpansion' \contained \delta$. But this is in contradiction with our hypothesis, hence any minimal counterexample is of size smaller or equal to $\max_{\gamma \in \Gamma}\nbatoms{\gamma} \cdot (\max_{\delta \in \Delta}\nbatoms{\delta}+1)$.
    }
  \end{proof}
  \changes{
  Given a "UCRPQ(SRE)" $\Gamma$, the construction of \Cref{lemma:approximation-for-finclass} yields its "maximal under-approximation" by "UCRPQs" of at most $k$ atoms as a "UCRPQ" $\Delta_{\textnormal{App}}$ whose every regular expression is a concatenation of expressions of the form "(i)@@sre", "(ii)@@sre" and "(iii)@@sre" above.
  }
  It suffices then to test $\gamma \contained \Delta_{\textnormal{App}}$ for every "CRPQ(SRE)" $\gamma$ in $\Gamma$. Due to \Cref{cl:small-counterexample-SRE}
  \changes{
    (and observing that equivalent queries without concatenations can be obtained in polynomial time)
  } 
  its negation $\gamma \not\contained \Delta_{\textnormal{App}}$ can be tested by guessing a polynomial sized "expansion" $\anexpansion$ of $\gamma$ and then testing $\anexpansion \not\contained \Delta_{\textnormal{App}}$.
  In turn, $\anexpansion \contained \Delta_{\textnormal{App}}$ can be tested in "NP" by 
  \changes{
    \cite[Theorem 4.2]{FigueiraGKMNT20}.\footnote{Simply by (1) guessing a polynomial size "CRPQ" $\delta$, (2) testing $\delta \in \Delta_{\textnormal{App}}$ in "NP" by \Cref{lemma:approximation-for-finclass}, and (3) guessing a (small) expansion $\anexpansion'$ of $\delta$ and testing $\anexpansion' \homto \anexpansion$.}
  }
  This yields a "PiP2" algorithm for testing $\gamma \contained \Delta_{\textnormal{App}}$, and thus also for $\Gamma \contained \Delta_{\textnormal{App}}$.
  "PiP2"-hardness follows from \Cref{coro:lowerbounds}.
\end{proof}

\section{Lower Bounds}
\AP\label{sec:lowerbounds}
In this section we give some underlying ideas for showing lower bounds for the "minimization problems". All the proofs can be found in the Appendix.

\begin{toappendix}
\subsection{Equivalence with a Single Atom}
\end{toappendix}

\subsection{Equivalence with a Single Atom}

"Containment of CRPQs" is "ExpSpace"-complete \cite{Florescu:CRPQ,four-italians}. Somewhat surprisingly,
Figueira \cite[Lemma 8]{figueira_containment_2020}
showed that there exists a finite alphabet $\A$
"st" the problem remained "ExpSpace"-hard even if restricted to
instances with a simple shape.
We start by strengthening this result to fit our needs.

\begin{propositionrep}[{Variation on \cite[Lemma 8]{figueira_containment_2020}}]
	\AP\label{prop:variation-figueira}
	There is a fixed alphabet over which the
	"con\-tainment problem" for "Boolean CRPQs" is "ExpSpace"-hard restricted
	to instances of the form%
	\[\gamma_1() = x \atom{K} y \qquad \contained^? \qquad \bigwedge_{j \in \lBrack 1,p\rBrack} x \atom{L_j} y = \gamma_2() \text{, where:}\]
	\begin{enumerate}
		\item no language among $K$ or the $L_i$'s is empty or contains the empty word $\varepsilon$, and
		\item there is no $i$ such that $\A^* L_i \A^* = \A^* \bigl( \bigcap_j L_j\bigr) \A^*$.
	\end{enumerate}
\end{propositionrep}
	\begin{proof}
	By inspecting \cite[Proof of Lemma 8, pp. 15--17]{figueira_containment_2020}, it can be noticed that
	actually the first condition is satisfied by Figueira's reduction---using his notation,
		neither $E$ nor $G_i\cup F_C\cup F_H$ with $i \in \lBrack 0,n\rBrack$
		are empty.
		
		Moreover, we claim that the reduction can be made so that the second condition also holds.
		Notice first that the $2^n$-tiling problem is still "ExpSpace"-complete
		if we restrict it to instances with $n > 1$ and such that all instances admit one tiling which is
		``locally valid'' but not valid---namely a tiling which satisfies all vertical and horizontal constraints,
		but not the initial and final tiles conditions.
		This can be achieved "eg" by
		adding a new tile $t$ "st" $(t,t)$ is both a valid horizontal and vertical configuration,
		but $t$ cannot be adjacent to any other tile.
		Then, the second condition amounts to showing that there is no $i$ "st"
		\[
			\mathbb{A}^* (G_i \cup F_C \cup F_H) \mathbb{A}^* =
			\mathbb{A}^* \Bigl( \bigcap_{0 \leq j \leq n} (G_j \cup F_C \cup F_H)\Bigr) \mathbb{A}^*
			\]
			which is equivalent, by elementary manipulations, to saying that for all $i \in \lBrack 0,n\rBrack$
			\[
				\mathbb{A}^* G_i \mathbb{A}^*
				\not\subseteq 
				\mathbb{A}^* \Bigl( \bigcap_{0 \leq j \leq n} G_j \Bigr) \mathbb{A}^*
				\cup \mathbb{A}^* (F_C \cup F_H) \mathbb{A}^*.
				\]
				For $i=0$, this holds because we can consider a valid encoding of a tiling
				which respects all constraints except that one vertical constraint is violated.
				For $i \in \lBrack 1,n\rBrack$, we consider the encoding of a tiling which is locally valid.
				Then, it has no vertical error, no horizontal error, and no encoding error, so
				it does not belong to the right-hand side.
				However, it belongs to $\mathbb{A}^* G_i \mathbb{A}^*$ for any $i \in \lBrack 1,n\rBrack$ since it contains a subword encoding two cells separated by exactly one row. Hence,
				the second condition also holds.
	\end{proof}

We shall use these hard instances to show that the "minimization problem" for "CRPQs" is hard.
\begin{thmrep}
	\AP\label{thm:minimization-lowerbound}
	The "minimization problem" for "CRPQs" is "ExpSpace"-hard. Further, there is a fixed alphabet "st" the problem of, given a "Boolean CRPQ"
	on this alphabet with only four variables is "equivalent"
	to a "Boolean CRPQ" with a single "atom" is "ExpSpace"-hard.
\end{thmrep}
\begin{proofsketch}
	We reduce an instance of the problem of \Cref{prop:variation-figueira} to the instance
	$\delta$, where
	\begin{center}
		\includegraphics{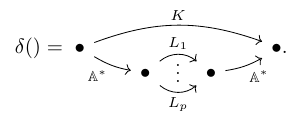}
    \end{center}
	First, it is easy to see that if $\gamma_1 \contained \gamma_2$ then $\delta \semequiv \gamma_1$. Conversely, if $\delta$ is "equivalent" to a "Boolean CRPQ" with at most one atom, then $\gamma_1 \contained \gamma_2$.
	The conditions imposed by \Cref{prop:variation-figueira} are necessary to discard one-atom queries which are self-loops and that whenever  $\gamma_1 \not\contained \gamma_2$ there is an "expansion" of $\delta$ to which any $\delta$-"equivalent" single-atom query "expansion" cannot be mapped.
\end{proofsketch}
\begin{proof}[Proof of \Cref{thm:minimization-lowerbound}]
	We reduce an instance of the problem of \Cref{prop:variation-figueira} to the instance
	$\delta$, where
		\begin{center}
			\includegraphics{figs/1.pdf}
		\end{center}
	
		\begin{claim}
			\AP\label{claim:minimization-lowerbound-1}
			If $\gamma_1 \contained \gamma_2$ then $\delta \semequiv \gamma_1$.
		\end{claim}
		Note first that $\delta \contained \gamma_1$.
		Then, if $\gamma_1 \contained \gamma_2$, then every word of $K$
		contains a factor which belongs to each $L_i$ for $i \in \lBrack 1, p\rBrack$,
		and hence $\gamma_1 \contained \delta$ "ie" $\delta \semequiv \gamma_1$,
		and so $\delta$ is equivalent to a "CRPQ" with a single "atom".
	
		\begin{claim}
			\AP\label{claim:minimization-lowerbound-2}
			Conversely, if $\delta$ is "equivalent" to a "Boolean CRPQ" with at most one atom,
			then $\gamma_1 \contained \gamma_2$.
		\end{claim}
		Let $\zeta$ be the "Boolean CRPQ" with at most one atom which is equivalent to $\delta$. Assume first, by contradiction, that it is a self-loop, "ie" $\zeta() = x \atom{M} x$ for some language $M$.
		Then by assumption on $K$, there exists a word $u\in K$ of size at least one.
		Since none of the $L_i$ are empty, there exists a "canonical database"
		$G_\delta^u$ where the atom
		$\qvar \atom{K} \qvar$ yielded a $u$-labelled path. Since $\delta \contained \zeta$,
		the database $G_\delta^u$ must satisfy $\zeta() = x \atom{M} x$.
		Since every strongly connected component of $G_\delta^u$ is trivial---we assumed that none of the languages
		of $\delta$ contained $\varepsilon$---, it must be that $\varepsilon \in M$, and hence $\zeta$
		is the query which is always satisfied, which contradicts the equivalence $\delta \semequiv \zeta$.
	
		Similarly, it can be shown that $\zeta$ cannot have zero atoms since $\delta$ is non-trivial. 
		Hence, $\zeta$ is exactly of the form $\zeta() = x \atom{M} y$ for some language $M$.
		First, note that from $\zeta \contained \delta$, it follows that
		\begin{equation}
			M \subseteq \A^* \bigl( \bigcap_j L_j \bigr) \A^*.
			\AP\label{eq:L-is-factor-of-M}
		\end{equation}
		Assume then, by contradiction, that there
		exists an $i$ "st" every word of $L_i$ has a factor in $M$,
		"ie" $L_i \subseteq \A^* M \A^*$. Then \eqref{eq:L-is-factor-of-M} implies
		$\A^* L_i \A^* = \A^* \bigl( \bigcap_j L_j \bigr) \A^*$ which contradicts
		the second assumption of \Cref{prop:variation-figueira}.
		Therefore, for every $i$, there is a word $v_i \in L_i$ which contains no factor in $M$.
		
		We are now ready to show that $\gamma_1 \contained \gamma_2$, by first observing that
		it boils down to showing $K \subseteq \A^* (\bigcap_j L_j ) \A^*$. Let $u \in K$.
		Consider the following "canonical database" of $\delta$, where the $v_i$'s are words
		defined as in the paragraph above:
		\begin{center}
			\includegraphics{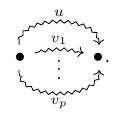}
		\end{center}
		Since $\delta \contained \zeta$, it must contain a path labelled by a word of $M$.
		But no $v_i$ contains a factor in $M$, hence it has to be $u$ that does.
		Hence, $K \subseteq \A^* M \A^*$. Together with
		\Cref{eq:L-is-factor-of-M}, we get $K \subseteq \A^* \bigl( \bigcap_j L_j \bigr) \A^*$,
		which concludes the proof of \Cref{claim:minimization-lowerbound-2}.
	
		Overall, we showed that $\gamma_1 \contained \gamma_2$ "iff"
		$\delta$ is equivalent to a "CRPQ" with at most one atom, which concludes the proof.
\end{proof}

\begin{toappendix}

Note that the assumption of $\A^* L_i \A^* = \A^* \bigl( \bigcap_j L_j\bigr) \A^*$
in \Cref{prop:variation-figueira} in necessary for the reduction to be correct,
otherwise $\delta()$ would be equivalent to
\[\delta'() \defeq x \atom{K} y \land x \atom{\A^* L_i \A^*} y,\]
and so we could have $\A^* L_i \A^* \subsetneq K$, implying that
(1) $K \not\subseteq \A^* L_i \A^*$ and hence $\gamma_1 \not\contained \gamma_2$
but (2) $\delta$ would be equivalent to a "CRPQ" with a single "atom", namely
$\delta''() \defeq x\atom{\A^* L_i \A^*} y$.
	
\subsection{Equivalence with a Single Variable}

\begin{theorem}
	\AP\label{thm:variable-minimization-lowerbound}
	There is a fixed alphabet "st" the problem of, given a "Boolean CRPQ"
	on this alphabet with only five variables is "equivalent"
	to a "Boolean CRPQ" with a single variable is "ExpSpace"-hard.
\end{theorem}

\begin{proof}
	We use same idea as in \Cref{thm:minimization-lowerbound}.
	We reduce the problem of \Cref{prop:variation-figueira} to the instance $\delta$, where
	\begin{center}
		\includegraphics{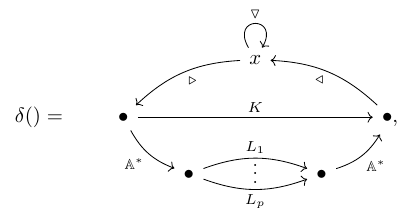}
    \end{center}
	where $\triangleright$, $\marking$ and $\triangleleft$ are new symbols.
	Note that despite being named, variable $x$ is also existentially quantified.
	
	\begin{claim}
		\AP\label{claim:variable-minimization-lowerbound-1}
		If $\gamma_1 \contained \gamma_2$ then $\delta \semequiv \gamma'_1$
		where $\gamma'_1() \defeq x \atom{\triangleright K \triangleleft} x
		\land x \atom{\marking} x$.
	\end{claim}
	
	If $\gamma_1 \contained \gamma_2$ then any word of $K$ contains a
	factor which belongs to $\bigcap_j L_j$ and so $\gamma'_1 \contained \delta$.
	The converse $\delta \contained \gamma'_1$ always holds.

	\begin{claim}
		\AP\label{claim:variable-minimization-lowerbound-2}
		Conversely, if $\delta$ is "equivalent" to a "Boolean CRPQ" with
		a single variable, then $\gamma_1 \contained \gamma_2$.
	\end{claim}

	Let $\zeta() = \bigwedge_{i=0}^n x \atom{M_i} x$ be a single-variable "Boolean CRPQ"
	that is "equivalent" to $\delta$.

	We first claim that there is some $i \in \lBrack 0,n\rBrack$
	"st" $M_i = \{\marking\}$. Every "canonical database" of $\delta$ contains a $\marking$-self loop
	and so from $\zeta \contained \delta$ it follows that any
	"canonical database" of $\zeta$ contains a $\marking$-self loop, which in turns implies
	that $M_i = \{\marking\}$ for some $i$. "Wlog", assume that $M_0 = \{\marking\}$.
	
	Observe that any "evaluation map" from $\zeta$ to a "canonical database"
	of $\delta$ must send $x \in \zeta$ to $x \in \delta$ because of the $\marking$-self loop,
	and conversely, any "evaluation map" from $\delta$ to a "canonical database"
	of $\zeta$ must send $x \in \delta$ to $x \in \zeta$.

	We remove from $\zeta$ all atoms $x \atom{M_i} x$ "st" $i \neq 0$
	and $M_i \cap {\marking^*} \neq\emptyset$. Thanks to the
	$\marking$-self loop, this transformation preserve the semantics of $\zeta$.
	More generally, if $M_i$ contains a word in which the letter `$\marking$' occurs, we get remove 
	the atom associated to $M_i$ altogether. The query obtained $\zeta'$ is clearly
	"st" $\zeta \contained \zeta'$, but dually for any "canonical database" $G_{\zeta'}$ of
	$\zeta'$, extend it to a "canonical database" $G_{\zeta}$ of $\zeta$ by picking, for any
	atom that was removed, any word containing the letter `$\marking$'.
	Since $\zeta \contained \delta$, there is an "evaluation map" from $\delta$
	to $G_{\zeta}$. Now the atoms of $\delta$ except the $\marking$-self loop
	do not use the letter $\marking$, and so it follows that the
	"evaluation map" from $\delta$ to $G_{\zeta}$ actually yields
	an "evaluation map" from $\delta$ to $G_{\zeta'}$.
	Hence, $\zeta' \contained \delta$ and thus $\zeta' \semequiv \delta$.

	The same argument works for "atoms" containing a word that does not
	start with $\triangleleft$, or that does not end $\triangleright$,
	or that contain strictly more than one occurrence of these symbols.
	Overall, it implies that "wlog" $\zeta$ is "equivalent" to
	\[
		x \atom{\marking} x \land \bigwedge_{j=1}^{m} x \atom{\triangleright N_j \triangleleft} x
	\]
	where $m \geq 0$ and the $N_j$'s are languages over $\A$.

	Assume now, by contradiction, that for all $j \in \lBrack 1,m \rBrack$
	"st" $N_j \not\subseteq K \cap \A^* \big(\bigwedge_i L_i\big) \A^*$. Pick for each $j$ a word $u_j$ witnessing this. 
	The canonical database of $\zeta$ induced by these
	words $\langle n_1, \hdots, n_m \rangle$, namely
	\[x \atom{\marking} x \land \bigwedge_{j=1}^m x \atom{\triangleright n_j \triangleleft}\]
	must satisfy $\delta$. But this implies that at least one $n_j$ must belong to
	$K \cap \A^* \big(\bigcap_i L_i\big) \A^*$.
	Contradiction.
	
	In fact, a argument similar to what we claimed before shows that
	we can remove all "atoms" "st" $N_j \not\subseteq K \cap \A^* \big(\bigcap_i L_i\big) \A^*$ without changing the semantics.
	Hence, "wlog", for each $j$, we have $N_j \subseteq K \cap \A^* \big(\bigcap_i L_i\big) \A^*$.

	We then claim that each word of $K$ must belong to all $N_j$.
	Indeed, let $u$ be a word of $K$. Let $v_i$ be a word in
	$L_i \smallsetminus \A^* \big(\bigcap_k L_k\big) \A^*$---recall
	that such words exist by an assumption of \Cref{prop:variation-figueira}---and
	consider the "canonical database" of $\delta$ obtained by "expanding"
	$K$ into $u$ and $L_i$ into $v_i$.
	Now $\delta \contained \zeta$ so this database must satisfy $\zeta$.
	Hence, for each $j$, $N_j$ must contain one word among $u$, $v_1, \hdots,v_m$.
	It cannot be any $v_i$ since otherwise we would have $v_i \in N_j \subseteq
	K \cap \A^* \big(\bigcap_k L_k\big) \A^* \subseteq \A^* \big(\bigcap_k L_k\big) \A^*$, which is a contradiction. And so $u \in N_j$.
	Therefore, we have
	\[
		K \subseteq \bigcap_j N_j \subseteq K \cap \A^* \big(\bigcap_i L_i\big) \A^*
	\]
	from which it follows that $K \subseteq \big(\bigcap_k L_k\big)$ and hence
	$\gamma_1 \contained \gamma_2$.
	
	Overall,
	\Cref{claim:variable-minimization-lowerbound-1,claim:variable-minimization-lowerbound-2}
	imply that $\gamma_1 \contained \gamma_2$ "iff" it is equivalent
	to a "CRPQ" with a single variable, in which case it is actually equivalent to
	\[\gamma'_1() \defeq x \atom{\triangleright K \triangleleft} x
		\land x \atom{\marking} x,\]
	which concludes the correctness of the reduction.
\end{proof}

\subsection{Variable Minimization is Harder than Containment}
A \AP""class of (Boolean) CRPQs"" is a function \AP$\intro*\classCRPQ$ mapping an alphabet $\A$ 
to a set $\classCRPQ_{\A}$ of "Boolean CRPQs", which is closed under variable renaming and alphabetic
renaming of the languages.

The \AP""disjoint conjunction"" \AP$\intro*\disconj$ of two "CRPQs" consists of
the conjunction of the queries, up to renaming so that their variable sets are disjoint---of two queries is 
the "class@@CRPQ" still belong to the "class@@CRPQ".
We say that the "class@@CRPQ" is \AP""closed under disjoint conjunction""
if $\gamma \in \classCRPQ_\A$ and $\delta \in \classCRPQ_\B$
imply $\gamma \disconj \delta \in \classCRPQ_{\A \cup \B}$.

Lastly, we say that the class is \AP""closed under variable marking""
if \emph{one of} the three following properties holds: 
\begin{description}
	\itemAP[\intro*\axiomVarMarkingLoop] for any $\gamma \in \classCRPQ_{\A}$, if $y$ is a variable of $\gamma$,
		if $a\not\in\A$, then $\gamma' \defeq \gamma \land y \atom{a} y$
		is in $\classCRPQ_{\A\sqcup\{a\}}$, or
	\itemAP[\intro*\axiomVarMarkingOut] for any $\gamma \in \classCRPQ_{\A}$, if $y$ is a variable of $\gamma$,	
		if $a\not\in\A$, then $\gamma' \defeq \gamma \land y \atom{a} y'$
		is in $\classCRPQ_{\A\sqcup\{a\}}$,
		where $y'$ is a new variable not occurring in $\gamma$, or
	\itemAP[\intro*\axiomVarMarkingIn] for any $\gamma \in \classCRPQ_{\A}$, if $y$ is a variable of $\gamma$,	
		if $a\not\in\A$, then $\gamma' \defeq \gamma \land y' \atom{a} y$
		is in $\classCRPQ_{\A\sqcup\{a\}}$,
		where $y'$ is a new variable not occurring in $\gamma$.
\end{description}

We will sometimes write $\gamma \in \classCRPQ$ to mean that $\gamma\in\classCRPQ_{\A}$
for some alphabet $\A$.

\begin{fact}
	Any "class@@CRPQ" defined by restricting the class of languages allowed to
	label the "atoms" is both "closed under disjoint conjunction" and "closed under variable marking", assuming that languages of
	the form $\{a\}$ are allowed, where $a$ is a single letter.
\end{fact}

\begin{theorem}
	\AP\label{thm:reduction-containment-to-variable-minimization}
	For any "class of CRPQs closed under disjoint conjunction" and
	"closed under variable marking"
	$\classCRPQ$, there is a polynomial-time reduction
	from the "containment problem" for "Boolean" queries of $\classCRPQ$ to the "CRPQ" 
	"minimization problem" restricted to queries of $\classCRPQ$. 
	The same bound applies if we add the constraint that the target "CRPQ" must also belong to $\classCRPQ$.
\end{theorem}

Say that a "CRPQ" is \AP""degenerate@@CRPQ"" if it contains an atom labelled the language $\{\varepsilon\}$.
Equivalently, it is "non-degenerate@@CRPQ" if it has at least one "canonical database"
which is "non-degenerate@@db".

\begin{fact}
	\AP\label{fact:produce-non-degenerate}
	One can turn a "degenerate CRPQ" into a "non-degenerate@@CRPQ" one by
	iteratively identifying variables adjacent to an atom $\atom{\{\varepsilon\}}$.
	This can be implemented in polynomial time.
\end{fact}

\begin{proof}[Proof of \Cref{thm:reduction-containment-to-variable-minimization}]
	We assume for now that $\classCRPQ$ satisfies
	the axiom $\axiomVarMarkingLoop$.
	Given an instance $\gamma_1() \contained^{?} \gamma_2()$ of the "containment problem" for "Boolean" queries of $\classCRPQ$,
	we assume "wlog" that $\gamma_1$ is "non-degenerate@@CRPQ" using \Cref{fact:produce-non-degenerate},
	and we reduce it to the instance $\langle \delta_1 \disconj \gamma_2, \nbvar{\delta_1} \rangle$,
	where $\delta_1$ is defined as:
	\[
		\delta_1() \defeq \gamma_1 \land \bigwedge_{x \in \vertex{\gamma_1}} x \atom{\marking_x} x
	\]
	where $\marking_x$ is a fresh letter for each $x \in \vertex{\gamma_1}$.
	The reduction works clearly in logarithmic-space,
	and clearly $\delta_1 \disconj \gamma_2 \in \classCRPQ$ since
	$\classCRPQ$ is "closed under disjoint conjunction" and $\axiomVarMarkingLoop$.
	Moreover, for it to be correct
	we need to show that $\gamma_1 \contained \gamma_2$ "iff" $\delta_1 \disconj \gamma_2$
	is "equivalent" to a "CRPQ" with at most $\nbvar{\delta_1}$ variables.

	\begin{claim}
		\AP\label{claim:reduction-containment-to-variable-minimization-1}
		If $\gamma_1 \contained \gamma_2$ then $\delta_1 \disconj \gamma_2 \semequiv \delta_1$.
	\end{claim}

	\begin{proof}
		Indeed, $\gamma_1 \contained \gamma_2$ implies $\delta_1 \contained \gamma_1 \contained \gamma_2$ and so $\delta_1 \disconj \gamma_2 \semequiv \delta_1$.
	\end{proof}

	Actually this property is an ``if and only if''. For the converse, we will prove a stronger statement.

	\begin{claim}
		\AP\label{claim:reduction-containment-to-variable-minimization-2}
		If $\delta_1 \disconj \gamma_2$ is "equivalent" to a "CRPQ"
		with at most $\nbvar{\delta_1}$ variables, then $\gamma_1 \contained \gamma_2$.
	\end{claim}

	\begin{proof}
		Let $\zeta$ be a "CRPQ" with at most $\nbvar{\delta_1}$ variables
		that is equivalent to $\delta_1 \disconj \gamma_2$.
		
		We claim first that for each $x \in \vertex{\zeta}$ there is a
		unique variable in $\zeta$ with a $\marking_x$-self-loop.
		Indeed, consider any "canonical database" $Z$ of $\zeta$:
		since $\zeta \contained \delta_1 \disconj \gamma_2$, there exists
		a "canonical database" $D_1$ of $\delta_1$ and $G_2$ of $\gamma_2$ "st"
		$D_1 \oplus G_2 \homto Z$ where $\oplus$ denotes the disjoint union.
		Since $D_1$ contains a $\marking_x$-self loop for each $x\in \vertex{\gamma_1}$,
		so does $Z$. Since this property holds for every $Z$, it follows that
		$\zeta$ must have a self-loop atom labelled by the singleton language $\{\marking_x\}$
		for each $x\in \vertex{\gamma_1}$.

		Now observe that no variable of $\zeta$ can be labelled by two $\marking_x$-self-loops
		with $x\in\vertex{\gamma_1}$. Indeed, $\gamma_1$ is "non-degenerate@@CRPQ", and so $\delta_1$
		is also "non-degenerate@@CRPQ", and so there
		exists a "canonical database" $D_1$ of $\delta_1$ which is "non-degenerate@@db".
		Then, pick any canonical database $G_2$ of $\gamma_2$.
		$D_1 \oplus G_2$ is a "canonical database" of $\delta_1 \disconj \gamma_2$, which is 
		equivalent to $\zeta$, so there is an "evaluation map"
		from $\zeta$ to $D_1 \oplus G_2$. If a variable of $\zeta$ had both a
		$\marking_x$- and a $\marking_y$-self-loop for $x \neq y \in \vertex{\gamma_1}$,
		then so would either $D_1$ or $G_2$. $G_2$ contains no such letters, and so it would have
		to be $D_1$. This contradicts the definition of $D_1$.
		Hence, no variable of $\zeta$ can be labelled by two $\marking_x$-self-loops.
		Together with the previous paragraph and the fact that 
		$\zeta$ has at most $\nbvar{\delta_1} = \nbvar{\gamma_1}$ variables, it follows that
		we can assume "wlog"---up to renaming the variables of $\zeta$---that
		$\vertex{\zeta} = \vertex{\gamma_1}$ and for each $x \in \vertex{\gamma_1}$,
		$x \atom{\marking_x} x$ is an "atom" of $\zeta$. Moreover, this is the only self-loop in
		$\zeta$ labelled by $\{\marking_x\}$, and for any self-loop atom $x \atom{L} x$
		we cannot have $\marking_y \in L$ for any $y \neq x \in \vertex{\gamma_1}$. 
		
		We are now ready to prove that $\gamma_1 \contained \gamma_2$.
		Let $G_1$ be a "canonical database" of $\gamma_1$, and let $D_1$ be the associated
		"canonical database" of $\delta_1$---it is obtained by adding an $\marking_x$-self-loop
		on every $x \in \vertex{\gamma_1}$.
		Pick any canonical database $G_2$ of $\gamma_2$. Since $\delta_1 \disconj \gamma_2 \contained \zeta$, there exists a "canonical database" $Z$ of $\zeta$
		"st" $Z \homto D_1 \oplus G_2$.
		But then, since $\zeta \contained \delta_1 \disconj \gamma_2$, there exists
		$D'_1$ and $G'_2$, which are "canonical databases" of $\delta_1$ and $\gamma_2$,
		respectively, "st"
		\[
			D'_1 \oplus G'_2 \homto
			Z \homto
			D_1 \oplus G_2.
		\]
		Restrict this homomorphism to $G'_2$: we obtain 
		\[
			G'_2 \homto
			Z \homto
			D_1 \oplus G_2.
		\]
		Now note that, because of the previous paragraph,
		the "homomorphism" $Z \homto D_1 \oplus G_2$ must map $x \in \vertex{Z}$  
		to $x \in \vertex{D_1}$---because of the $\marking_x$-self-loop.
		Since $D_1 \oplus G_2$ is a disjoint union, it follows that image of this
		"homomorphism" is actually included in $D_1$, and so obtain a "homomorphism"
		\[
			G'_2 \homto Z \homto D_1.
		\]
		Now of course $\marking_x$-self-loop will occur in the image of any "homomorphism"
		$Z \homto D_1$. However, in the composition $G'_2 \homto Z \homto D_1$,
		since $G'_2$ does not use any letter of the form $\marking_x$, $x\in \vertex{\gamma_1}$,
		we conclude we actually get a "homomorphism"
		\[
			G'_2 \homto G_1,
		\]
		which concludes the proof that $\gamma_1 \contained \gamma_2$.
	\end{proof}
	
	\Cref{claim:reduction-containment-to-variable-minimization-1,claim:reduction-containment-to-variable-minimization-2} imply that
	$\gamma_1 \contained \gamma_2$ "iff" $\delta_1 \disconj \gamma_2$ is "equivalent" to a "CRPQ"
	with at most $\nbvar{\delta_1}$ variables, which concludes the reduction under the assumption
	that $\classCRPQ$ satisfies $\axiomVarMarkingLoop$.

	To conclude, note that
	if $\classCRPQ$ satisfies either $\axiomVarMarkingOut$ or $\axiomVarMarkingIn$ then exactly the same proof works, except that the definition of $\delta_1$
	should be changed: variables will be marked using outgoing and incoming edges, respectively.
	Lastly, since $\delta_1 \in \classCRPQ$, then we have as a by-product of our proof
	that $\delta_1 \disconj \gamma_2$ is equivalent to a "CRPQ" with at most $k$ atoms
	"iff" it is equivalent to a "CRPQ" of $\classCRPQ$ with at most $k$ atoms.
	It follows that this reduction also works it we add the constraint that
	$\delta$ must be in $\classCRPQ$.
\end{proof}

\end{toappendix}

\begin{toappendix}
	\subsection{Minimization is Harder than Containment}
\end{toappendix}
\subsection{Minimization is Harder than Containment}
We show that, under some technical conditions, the containment problem can be reduced to the "minimization problem". This allows to transfer known lower bounds from the "containment problem" of "CRPQ" classes. Due to space constraints we only state the key definitions and lemmas.
\changes{We first introduce an intermediary technical property called ``canonization'',
which ensures the feasibility of the reduction.}

\paragraph*{Canonization.}
We say that an "expansion" $\anexpansion$ of a "CRPQ" $\gamma$ is \AP""non-degenerate@@db""
if no "atom refinement" in $\anexpansion$ was obtained using the empty word.
\begin{toappendix}
	\begin{fact}
		\AP\label{fact:nb-seg-expansion}
		If $\anexpansion$ is a "non-degenerate expansion" of $\gamma$ and
		$\gamma$ is "fully contracted", then $\nbseg{\anexpansion} = \nbatoms{\gamma}$.
	\end{fact}
\end{toappendix}
A \AP""class of (Boolean) CRPQs"" is a function \AP$\intro*\classCRPQ$ mapping an alphabet $\A$ 
to a set $\classCRPQ_{\A}$ of "Boolean CRPQs", which is closed under variable renaming and alphabetic
renaming of the languages.
Given a "class of CRPQs" $\classCRPQ$, the \AP""$\classCRPQ$-canonization problem""
is the functional problem taking an alphabet $\A$
and two "Boolean CRPQs" $\langle \gamma_1,\gamma_2\rangle$ in $\classCRPQ_{\A}$,
and outputting an alphabet \AP$\intro*\alphabetmarking$, two other "Boolean CRPQs" $\langle \delta_1, \delta_2 \rangle$, in $\classCRPQ_{\A\sqcup \alphabetmarking}$,
such that:
\begin{description}
	\itemAP[\intro*\axiomCanonMonotonicity{}:] \phantomintro*\axiomsCanon
		$\gamma_1 \contained \gamma_2$ "iff" $\delta_1 \contained \delta_2$,
	\itemAP[\intro*\axiomCanonCore{}:]
		there exists a "non-degenerate@@cdb" 
		$D_1 \in \Exp(\delta_1)$ "st", for every $D'_1 \in \Exp(\delta_1)$ and $f: D'_1 \homto D_1$ we have (i) $D'_1$ is "non-degenerate@@cdb", (ii) $f$ is "strong onto", and (iii)  $f(x)=x$ for every $x\in \vars(\delta_1)$,
	\itemAP[\intro*\axiomCanonNonRed{}:] for each $D_1 \in \Exp(\delta_1)$,
		for each $x,y \in \vars(\delta_1)$, there cannot be an "atom refinement" in $D_1$
		from $x$ to $y$ and another path from $x$ to $y$ in $D_1$, disjoint from the "atom refinement" that share the same label,
	\itemAP[\intro*\axiomCanonContracted{}:] $\delta_1$ is "fully contracted",
	\itemAP[\intro*\axiomCanonContainment{}:] $\gamma_2 \contained \delta_2$, and 
	\itemAP[\intro*\axiomCanonMarking{}:] each connected component of $\delta_1$ must contain at least one "atom"
	labelled by a language $L$ "st" every word of $L$ must contain at least one letter from $\alphabetmarking$.
\end{description}

The \reintro{$\classCRPQ$-strong canonization problem} is defined similarly, except that we replace `there exists' with `for all' in \axiomCanonCore{}---see \axiomStrongCanonCore{} 
for a formal definition.

\begin{toappendix}
The \AP""$\classCRPQ$-strong canonization problem"" is defined similarly to the "$\classCRPQ$-canonization problem", except that $\axiomCanonCore{}$ is replaced by the axiom
\begin{description}
	\itemAP[\intro*\axiomStrongCanonCore{}:] %
	for every "non-degenerate@@cdb" 
		$D_1 \in \Exp(\delta_1)$, every $D'_1 \in \Exp(\delta_1)$, and $f: D'_1 \homto D_1$ we have (i) $D'_1$ is "non-degenerate@@cdb", (ii) $f$ is "strong onto", and (iii)  $f(x)=x$ for every $x\in \vars(\delta_1)$,
\end{description}
\end{toappendix}

We show that assuming we can solve the "$\classCRPQ$-canonization problem" ("resp" "$\classCRPQ$-strong canonization problem") this problem, then the "CRPQ"
(resp. "UCRPQ") "minimization problem"
restricted to "CRPQs" of $\classCRPQ$ ("resp" to "UCRPQs" whose "disjuncts" are all in $\classCRPQ$) is harder than the "containment problem" over $\classCRPQ$.

\begin{toappendix}
In the following statement, a \AP""$\classCRPQ$-canonization oracle"" (resp. \AP""$\classCRPQ$-strong canonization oracle"") is an oracle to any algorithm solving the "$\classCRPQ$-canonization problem" (resp. "$\classCRPQ$-strong canonization problem").
\end{toappendix}

\begin{lemmarep}
	\AP\label{lem:reduction-containment-to-minimization}
	For any class $\classCRPQ$ of "CRPQs" closed under "disjoint conjunction",
	there is a polynomial-time algorithm
	using a "$\classCRPQ$-canonization oracle" ("resp" "$\classCRPQ$-strong canonization oracle")
	from the "containment problem" for "Boolean" queries of $\classCRPQ$ to the
	"CRPQ" (resp. "UCRPQ") "minimization problem" restricted to "CRPQs" in $\classCRPQ$
	(resp. "UCRPQs" whose "disjuncts" are all in $\classCRPQ$).
	The reduction also applies under the restriction that the target query must also be in
	$\classCRPQ$.
\end{lemmarep}

\begin{proof}
	\proofcase{Minimization in the class of "CRPQs".}
	Let $\gamma_1 \contained^{?} \gamma_2$ be an instance of the "containment problem" for "Boolean" queries of $\classCRPQ$.
	We apply the "$\classCRPQ$-canonization oracle" to obtain 
	a pair $\langle \delta_1, \delta_2 \rangle$ as in the axioms $\axiomsCanon$.
	We then map the instance $\gamma_1 \contained^{?} \gamma_2$ to 
	$\langle \delta_1 \disconj \gamma_2, \nbatoms{\delta_1} \rangle$.

	The reduction works in logarithmic space with a "$\classCRPQ$-canonization oracle",
	and clearly $\delta_1 \disconj \delta_2 \in \classCRPQ$ since
	both $\delta_1$ and $\delta_2$ are in $\classCRPQ$ and $\classCRPQ$ is "closed under disjoint conjunction".
	We need to show that $\gamma_1 \contained \gamma_2$ "iff" $\delta_1 \disconj \gamma_2$
	is "equivalent" to a "CRPQ" with at most $\nbatoms{\delta_1}$ "atoms".

	\begin{claim}
		\AP\label{claim:reduction-containment-to-minimization-1}
		If $\gamma_1 \contained \gamma_2$, then $\delta_1 \disconj \delta_2 \semequiv \delta_1$
		and so $\delta_1 \disconj \delta_2$ is "equivalent" to a "CRPQ" with at most $\nbatoms{\delta_1}$ "atoms".
	\end{claim}

	This follows from $\axiomCanonMonotonicity$.
	The converse hold for the same reason, but we actually need a stronger property.

	\begin{claim}
		\AP\label{claim:reduction-containment-to-minimization-2}
		If $\delta_1 \disconj \delta_2$ is "equivalent" to a "CRPQ"
		with at most $\nbatoms{\delta_1}$ "atoms", then $\gamma_1 \contained \gamma_2$.
	\end{claim}

	We write $\zeta$ as $\zeta_+ \disconj \zeta_-$ where $\zeta_+$ is
	the "disjoint conjunction" of all connected components of $\zeta$ containing
	an "atom" whose language contains a word containing a `$\alphabetmarking$'-letter,
	and $\zeta_-$ is the "disjoint conjunction" of all other
	components. We want to show that $\zeta_-$ is actually empty.

	Let $D_1$ be a "canonical database" of $\delta_1$ as in \axiomCanonCore{}.
	Then pick any "canonical database" $G_2$ of $\gamma_2$.
	By $\axiomCanonContainment$, there exists $D_2 \cdb \delta_2$ "st" $D_2 \homto G_2$.
	Then $D_1 \oplus D_2 \cdb \delta_1 \disconj \delta_2$, so from $\delta_1 \disconj \delta_2 \semequiv \zeta$
	it follows that there exists $Z_+ \cdb \zeta_+$, $Z_- \cdb \zeta_-$, $D'_1 \cdb \delta_1$ and $D'_2 \cdb \delta_2$
	such that we have "homomorphisms"
	\[
		D'_1 \oplus D'_2
		\xrightarrow{f}
		Z_+ \oplus Z_-
		\xrightarrow{g}
		D_1 \oplus D_2
		\xrightarrow{h}
		D_1 \oplus G_2.
	\]
	By $\axiomCanonMarking$, every connected component of $D'_1$ must contain at least one edge labelled
	by a letter of $\alphabetmarking$, and so the "homomorphism" $f_{\restriction{D'_1}}\colon
	D'_1 \homto Z_+ \oplus Z_-$ is actually a "homomorphism" from $D'_1$ to $Z_+$.
	Note then that $(h \circ g)_{\restriction{Z_+}}$ maps $Z_+$ onto $D_1 \oplus G_2$ but
	since $G_2$ contains no letter $\alphabetmarking$, the image of this "homomorphism" is included in $D_1$.
	Overall, we have "homomorphisms"
	\[D'_1 \xrightarrow{i} Z_+ \xrightarrow{j} D_1.\] 
	By \axiomCanonCore{}, $j\circ i$ be "strong onto" and for every $x\in \vars(\delta_1)$,
	$j(i(x)) = x$.
	
	We claim that for each $x$ "external" in $D'_1$, then $i(x)$ is "external".
	First, since $x \in \vars(\delta_1)$ is "external", then $x \in \vars(\delta_1)$
	and so $j(i(x)) = x$. Then $j \circ i$ is "strong onto" and so $j$ be also be
	"strong onto". It follows that the in-degree (resp. out-degree) of $i(x)$
	is lower bounded by the in-degree (resp. out-degree) $j(i(x)) = x$.
	So, if $x$ has in-degree or out-degree at least 2, so does $i(x)$.
	Moreover, if $x$ has in-degree (resp. out-degree 0), then 
	so must $i(x)$ because otherwise, $j(i(x)) = x$ should also have an incoming edge.
	Overall, by letting $i[D'_1]$ be the image of $D'_1$ by $i$,
	we get that the natural embedding $i[D'_1] \homto Z_+$ satisfies the assumption
	of \Cref{coro:embedding-segments} and so $\nbseg{i[D'_1]} \leq \nbseg{Z_+}$.
	
	Now observe that $i\colon D'_1 \to i[D'_1]$ is injective on $\vars(\delta_1)$
	because $j(i(x))$ for any $x\in \vars(\delta_1)$.
	Moreover, by \axiomCanonNonRed{}, $i$ cannot identity an "atom" with another path of "atoms" and so $D'_1$ is actually isomorphic to $i[D'_1]$, from which we
	get $\nbseg{i[D'_1]} = \nbatoms{\delta_1}$
	and so $\nbseg{Z_+} \geq \nbatoms{\delta_1}$.
	By \Cref{prop:seggraph-of-expansion} $\nbatoms{\zeta_+} \geq \nbseg{Z_+}$, and so $\zeta_+$ has at least $\nbatoms{\delta_1}$
	"atoms", but since we assumed that $\zeta$ has at most $\nbatoms{\delta_1}$ "atoms",
	it follows that $\zeta_-$ is trivial and $\zeta \semequiv \zeta_+$.
	
	We are now ready to prove that $\delta_1 \contained \delta_2$.
	Let $D_1 \cdb \delta_1$. Pick any $G_2 \cdb \gamma_2$.
	By $\axiomCanonContainment$, there exists $D_2 \cdb \delta_2$ "st" $D_2 \homto G_2$.
	Then since $\delta_1 \disconj \delta_2 \semequiv \zeta_+$,
	there exists $Z_+ \cdb \zeta_+$, $D'_1 \cdb \delta_1$ and $D'_2 \cdb \delta_2$ "st"
	\[
		D'_1 \oplus D'_2
		\homto
		Z_+
		\homto
		D_1 \oplus D_2
		\homto
		D_1 \oplus G_2.
	\]
	Because of $\alphabetmarking$, the "homomorphism" $Z_+ \homto D_1 \oplus G_2$ must in fact be
	a homomorphism $Z_+ \homto D_1$, and so by composition we obtain a "homomorphism"
	$D'_1 \oplus D'_2 \homto D_1$, which can be restricted to $D'_2$, yielding $D'_2 \homto D_1$.
	Hence, $\delta_1 \contained \delta_2$, and so by $\axiomCanonMonotonicity$, $\gamma_1 \contained \gamma_2$.

	Putting \Cref{claim:reduction-containment-to-minimization-1,claim:reduction-containment-to-minimization-2}
	together shows that the reduction is correct. We now prove that this reduction also works for the other
	variations of the problem.

	\proofcase{Minimization in the class of "UCRPQs".}
	If we allow $\zeta$ to be a "UCRPQ", then \Cref{claim:reduction-containment-to-minimization-1} still holds,
	and we need to adapt \Cref{claim:reduction-containment-to-minimization-2}. Assume that this ``small''
	"UCRPQ" is of the form $\zeta_1 \lor \zeta_2 \lor \cdots \lor \zeta_k$ where each $\zeta_i$ has
	at most $\nbatoms{\delta_1}$ "atoms".
	We say that a "disjunct" $\zeta_i$ is \AP""relevant@@prooflowerbound"" when it has
	at least one "canonical database" $Z_i$ appearing in a pattern of the form
	\[
		D'_1 \oplus D'_2
		\homto
		Z_i
		\homto
		D_1 \oplus D_2
		\homto
		D_1 \oplus G_2.
	\]
	for some $D_1,D'_1 \cdb \delta_1$, $D_2,D'_2 \cdb \delta_2$ and $G_2 \cdb \gamma_2$.
	Using $\axiomStrongCanonCore$ on $D_1$,
	the same proof as in the case of "CRPQs" apply. We can then conclude that "wlog"
	$\zeta_i \semequiv (\zeta_i)_+$ for all "relevant@@prooflowerbound" "disjunct".
	The proof of $\delta_1 \contained \delta_2$---and hence $\gamma_1 \contained \gamma_2$---then goes through as before,
	which concludes the proof.

	\proofcase{If $\zeta$ is restricted to be in $\classCRPQ$.}
	Then \Cref{claim:reduction-containment-to-minimization-2} still holds.
	To adapt \Cref{claim:reduction-containment-to-minimization-1}, it suffices to remark that
	$\delta_1 \in \classCRPQ$.
\end{proof}

\begin{proof}[Proof sketch]
	\changes{
		The construction reduces some restriction of the "containment problem"
		to some variant of the "minimization problem". The main idea is,
		given an instance $\gamma_1 \contained^? \gamma_2$, to build "CRPQs"
		$\delta_1$ and $\delta_2$ with some desirable properties "st" the following are equivalent: (i)
		$\gamma_1 \contained \gamma_2$, (ii) $\delta_1 \contained \delta_2$,
		(iii) $\delta_1 \disconj \delta_2 \semequiv \delta_1$,
		where $\reintro*\disconj$ denotes the \reintro{disjoint conjunction} ("ie", the conjunction of "atoms" of both queries making sure that variables are disjoint), and
		(iv) $\delta_1 \disconj \delta_2$ is equivalent to a
		"CRPQ" whose size is at most the size of $\delta_1$.
		Of course, (ii) $\Leftrightarrow$ (iii) always holds, as well as (iii) $\Rightarrow$ (iv).}

	\changes{All the difficulty lies in guaranteeing the converse property: (iv) $\Rightarrow$ (iii).
		We will use the constructions $\gamma_i \mapsto \delta_i$, given
		by the "canonization problem", to enforce it
		while respecting (i) $\Leftrightarrow$ (ii). 
		The main idea of this approach is to add some `marking' (with fresh alphabet symbols) of either the variables or the atoms of $\gamma_1$ in $\delta_1$, ensuring that
		$\delta_1$ has some strong structure implying that---loosely speaking---any query equivalent to $\delta_1 \disconj \delta_2$ must contain $\delta_1$ as a subquery.
		Using the assumption that $\delta_1 \disconj \delta_2$ is equivalent to a
		"CRPQ" whose size is at most the size of $\delta_1$, we conclude that in fact $\delta_1 \disconj \delta_2 \semequiv \delta_1$, "ie" $\delta_1 \contained \delta_2$ and so $\gamma_1 \contained \gamma_2$. 
	}
\end{proof}

\begin{toappendix}
\begin{lemma}
	\AP\label{lem:canonization-CRPQs}
	The "strong canonization problem" can be solved in non-deterministic logarithmic space for the class of all
	"CRPQs", or more generally for all "classes of CRPQs" defined by restricting the
	underlying multigraph class, provided that this class is "closed under disjoint union".
\end{lemma}

\begin{proof}
	Given a pair $\langle\gamma_1, \gamma_2\rangle$ of "Boolean" queries
	we assume "wlog" that $\gamma_1$ is "fully contracted" using
	\Cref{fact:produce-fully-contracted}---which works in non-deterministic logarithmic space---,
	and we reduce it to the pair $\langle \delta_1, \delta_2 \rangle$,
	where $\delta_1$ is defined as
	\[
		\delta_1() \defeq \bigwedge_{\alpha = x \atom{L} y \in \gamma_1} x \atom{\triangleright_\alpha L \triangleleft_\alpha} y
	\]
	where $\triangleright_\alpha$ and $\triangleleft_\alpha$ are fresh letters for each "atom" $\alpha$ of $\gamma_1$,
	and 
	\[
		\delta_2() \defeq \bigwedge_{x \atom{L} y \in \gamma_2} x \atom{\erasingmorphism[-1]{\alphabetmarking}[L]} y,
	\]
	where $\alpha_1, \hdots, \alpha_k$ are all the "atoms" of $\gamma_1$,
	$\alphabetmarking \defeq \{\triangleright_{\alpha_1},\triangleleft_{\alpha_1}, \hdots,
	\triangleright_{\alpha_k}, \triangleleft_{\alpha_k}\}$ and
	\AP$\intro*\erasingmorphism{\alphabetmarking}\colon (\A\cup \alphabetmarking)^* \to \A^*$
	is the monoid morphism that maps letters of $\A$ to themselves and 
	letters of $\alphabetmarking$ to the empty word.

	In other words, $\delta_1$ is similar to $\gamma_1$ except that it must read the special
	symbol $\triangleright_{\alpha}$ before satisfying "atom" $\alpha \in \gamma_1$,
	and read symbol $\triangleleft_{\alpha}$ after.
	On other hand, $\delta_2$ is obtained from $\gamma_2$ by relaxing the constraints:
	instead of having to read a path labelled by some language $L$,
	we now must read a path such that, when we ignore these new symbols $\triangleright_{\alpha}$
	and $\triangleleft_{\alpha}$, then it belongs to $L$.

	We now need to prove that properties \axiomsCanon{} and \axiomStrongCanonCore{} hold.
	\begin{claim}
		\AP\label{claim:canonization-multigraph-monotonic-l-to-r}
		If $\gamma_1 \contained \gamma_2$ then $\delta_1 \disconj \delta_2 \semequiv \delta_1$.
	\end{claim}
		
	Showing that $\delta_1 \disconj \delta_2 \semequiv \delta_1$
	amounts to showing $\delta_1 \contained \delta_2$.
	Let $D_1$ be a "canonical database" of $\delta_1$.
	Consider the "canonical database" of $\gamma_1$ obtained
	by removing every edge of the form $x \atom{\triangleright_\alpha} y$
	or $x \atom{\triangleleft_\alpha} y$, and merging variables $x$ and $y$.
	Since $\gamma_1 \contained \gamma_2$, there exists a "canonical database"
	$G_2$ of $\gamma_2$ and a homomorphism $f\colon G_2 \to G_1$.
	We then define a "canonical database" $D_2$ of $\delta_2$ together
	with a homomorphism $g\colon D_2 \to D_1$ as follows:
	given an "atom refinement"
	\[
		x_0 \atom{b_1} x_1 \atom{b_2} \cdots \atom{b_n} x_n
		\text{ in } G_2.
	\]
	we look at its image
	\[
		f(x_0) \atom{b_1} f(x_1) \atom{b_2} \cdots \atom{b_n} f(x_n)
		\text{ in } G_1.
	\]
	Now some $f(x_i)$'s might be variables of $\gamma_1$ and hence
	this might not a path in $G_1$.
	We let $i_1 < \hdots < i_k$ denote the indices $i$ "st" $x_i \in \vertex{\gamma_1}$, so
	that we can split the path in $G_1$ into multiples "atom refinements" 
	of "atoms" of $\gamma_1$:
	\begin{align*}
		\underbrace{f(x_0) \atom{b_1} f(x_1) \atom{b_2} \cdots \atom{b_{i_1}}}_{
			\text{end of an "atom refinement" of $\alpha_0$}
		}
		f(x_{i_1})
		\underbrace{\atom{b_{i_1+1}} f(x_{i_{1}+1}) \atom{b_{i_1+2}} \cdots \atom{b_{i_2}}}_{
			\text{"atom refinement" of $\alpha_1$}
		}
		f(x_{i_2})
		\atom{b_{i_2+1}} f(x_{i_{2}+1}) \atom{b_{i_2+2}} \cdots
		\\[1em]
		\cdots \atom{b_{i_k}}
		f(x_{i_k})
		\underbrace{\atom{b_{i_k+1}} f(x_{i_{k}+1}) \atom{b_{i_k+2}} \cdots \atom{b_n} f(x_n)}_{
			\text{beginning of an "atom refinement" of $\alpha_k$}
		}
	\end{align*}
	in $G_1$.
	For $i \in\lBrack 1,k\rBrack$, we let $f(x_{i_j})^r$ (resp. $f(x_{i_j})^l$) denote the unique variable of
	$D_1$ "st" there is an edge from $f(x_{i_j})$ to $f(x_{i_j})^r$
	labelled by $\triangleright_{\alpha_i}$ (resp. from $f(x_{i_j})^l$ to $f(x_{i_j})$
	labelled by $\triangleleft_{\alpha_{i-1}}$), we obtain a path
	\begin{align*}
		\underbrace{f(x_0) \atom{b_1} f(x_1) \atom{b_2} \cdots \atom{b_{i_1}}}_{
			\text{end of an "atom refinement" of $\alpha_0$}
		}
		f(x_{i_1})^l \atom{\triangleleft_{\alpha_{0}}} f(x_{i_1}) \atom{\triangleright_{\alpha_{1}}} f(x_{i_1})^r
		\\[1em]
		\underbrace{\atom{b_{i_1+1}} f(x_{i_{1}+1}) \atom{b_{i_1+2}} \cdots \atom{b_{i_2}}}_{
			\text{"atom refinement" of $\alpha_1$}
		}
		f(x_{i_2})^l 
		\atom{\triangleleft_{\alpha_{1}}} f(x_{i_2}) \atom{\triangleright_{\alpha_{1}}} f(x_{i_2})^r 
		\atom{b_{i_2+1}} f(x_{i_{2}+1}) \atom{b_{i_2+2}} \cdots
		\\[1em] 
		\cdots \atom{b_{i_k}}
		f(x_{i_k})^l \atom{\triangleleft_{\alpha_{k-1}}} f(x_{i_k}) \atom{\triangleright_{\alpha_{k}}} f(x_{i_k})^r
		\underbrace{\atom{b_{i_k+1}} f(x_{i_{k}+1}) \atom{b_{i_k+2}} \cdots \atom{b_n} f(x_n)}_{
			\text{beginning of an "atom refinement" of $\alpha_k$}
		}
	\end{align*}
	in $D_1$. Hence, we build $D_2$ by replacing each "atom refinement"
	$
		x_0 \atom{b_1} x_1 \atom{b_2} \cdots \atom{b_n} x_n
		\text{ in } G_2
	$
	by
	\begin{align*}
		x_0 \atom{b_1} x_1 \atom{b_2} \cdots \atom{b_{i_1}}
		x_{i_1}^l \atom{\triangleleft_{\alpha_0}} x_{i_1} \atom{\triangleright_{\alpha_{1}}} x_{i_1}'
		\atom{b_{i_1+1}} x_{i_{1}+1} \atom{b_{i_1+2}} \cdots \atom{b_{i_2}}
		x_{i_2}^l \atom{\triangleleft_{\alpha_1}} x_{i_2} \atom{\triangleright_{\alpha_{2}}} x_{i_2}'\\
		\atom{b_{i_2+1}} x_{i_{2}+1} \atom{b_{i_2+2}} \cdots \atom{b_{i_k}}
		x_{i_k}^l \atom{\triangleleft_{\alpha_{k-1}}} x_{i_k} \atom{\triangleright_{\alpha_{k}}} x_{i_k}'
		\atom{b_{i_k+1}} x_{i_{k}+1} \atom{b_{i_k+2}} \cdots \atom{b_n} x_n,
	\end{align*}
	where $x_{i_1}^l, x_{i_1}^r, \hdots, x_{i_k}^l, x_{i_k}^r$ are new fresh variables.
	By construction, $D_2$ comes equipped with a "homomorphism" 
	$g\colon D_2 \to D_1$ which sends $x_i$ to $f(x_i)$,
	$x_{i_j}^l$ to $f(x_{i_j})^l$ and $x_{i_j}^r$ to $f(x_{i_j})^r$. Since $D_2$ is---by
	construction---a "canonical database" of $\delta_2$, this concludes the proof that 
	$\delta_1 \contained \delta_2$.

	\begin{claim}
		\AP\label{claim:canonization-multigraph-monotonic-r-to-l}
		If $\delta_1 \contained \delta_2$ then $\gamma_1 \contained \gamma_2$.
	\end{claim}

	The construction is dual to \Cref{claim:canonization-multigraph-monotonic-l-to-r} and left to the
	reader. Both claims yield $\axiomCanonMonotonicity$.

	We now show that \axiomStrongCanonCore{} holds: pick a "canonical database" $D_1 \cdb \delta_1$.
	For any $D'_1 \cdb \delta_1$, if $D'_1 \homto D_1$, then because of the letters in $\alphabetmarking$,
	it follows that for each "atom" $\alpha_i$ of $\gamma_1$, the "atom refinement" of
	$\alpha_i$ in $D'_1$ must be sent bijectively on the "atom refinement" of $\alpha_i$ in $D_1$,
	and so they are equal. It follows that the homomorphism $D'_1 \to D_1$ must actually be the identity,
	and hence $D_1$ is "maximal".
	The same argument applied to $D'_1 \defeq D_1$ shows that the only "homomorphism" from $D_1$ itself
	is the identity, and so in particular $D'_1$ is a "core".
	Lastly, because of the letters of $\alphabetmarking$, no atom of $\delta_1$ contains the empty word,
	and so in particular $D_1$ must be "non-degenerate@@db". Hence, \axiomStrongCanonCore{} holds.

	Since $\gamma_1$ is "fully contracted", so is $\delta_1$, which proves \axiomCanonContracted{}.
	For any language $L$, we have $L \subseteq \erasingmorphism[-1]{\alphabetmarking}[L]$, and so we have \axiomCanonContainment{}.
	Finally, \axiomCanonNonRed{} and \axiomCanonMarking{} trivially hold.
	
	Together with that fact that
	$\langle \gamma_1,\gamma_2 \rangle \mapsto \langle \delta_1, \delta_2\rangle$ preserves
	the underlying multigraphs, this shows that this is a solution to "strong canonization problem"
	for any "class of CRPQs" defined by restring the underlying class of multigraphs.
	Note also that an NFA for $\erasingmorphism[-1]{\alphabetmarking}[L]$ can be obtained
	from an NFA for $L$ by adding on every state a self-loop labelled by every possible letter of $\alphabetmarking$
	and hence, this algorithm can be implemented in logarithmic space.
\end{proof}

	Note however that if $L$ is a "simple regular expression", then
	$\erasingmorphism[-1]{\alphabetmarking}[L]$ does not need to be.
	Hence, the construction above does not work for "CRPQs" over "simple regular expressions".

\begin{lemma}
	\AP\label{lem:canonization-SREs}
	The "strong canonization problem" can be solved in polynomial time
	for the "class of CRPQs" over "simple regular expressions".
\end{lemma}

	Given a "CRPQ" $\gamma$, we say that an "atom" $x\atom{L} y$ is \AP""locally redundant"" if
	there exists a path of "atoms" $z_0 \atom{L_1} z_1 \atom{L_2} \cdots \atom{L_n} z_n$
	in which $x\atom{L} y$ does not occur, and with $z_0 = x$ and $z_n = x$
	where $L_1 L_2\cdots L_n \subseteq L$.

\begin{proof}[Proof of \Cref{lem:canonization-SREs}]
	Fix a pair $\langle \gamma_1, \gamma_2 \rangle$ of "CRPQs".
	From $\gamma_1$, we start by picking a "locally redundant atom" (if any), and remove it.
	We iterate this process, until we get a "CRPQ" with no "locally redundant atom" $\gamma'_1$.
	By construction, it is "equivalent" to $\gamma_1$.\footnote{Note however that in general $\gamma'_1$ cannot be
	obtained by only keeping all atoms of $\gamma_1$ which are not "locally redundant": for instance, if $\gamma_1() =
	x \atom{L} y \land x \atom{L} y$, then all atoms are "locally redundant".
	Instead, we need to remove such atoms one after the other.}
	Moreover, $\gamma'_1$ can be computed in polynomial time.
	We then refine in $\gamma'_1$ each atom so that each "atom" is either labelled by $a$ or
	$a^+$ for some $a\in \A$.

	We then define $\langle \delta_1, \delta_2 \rangle$, where
	$\delta_2 \defeq \gamma_2$ and
	\[
		\delta_1 = \Bigl(
			\bigwedge_{x \atom{L} y \in \gamma'_1} x \atom{L} y
		\Bigr)
		\land 
		\Bigl(
			\bigwedge_{x \in \vertex{\gamma'_1}} x \atom{\marking_x} x
		\Bigr),
	\]
	and we let $\alphabetmarking \defeq \{\marking_x \mid x \in \vertex{\gamma'_1}\}$.

	Next, we show that \axiomCanonMonotonicity{} holds:
	if $\gamma_1 \contained \gamma_2$ then $\delta_1 \contained \gamma'_1 \semequiv \gamma_1 \contained \gamma_2 = \delta_2$,
	and dually if $\delta_1 \contained \delta_2$ then let $G_1 \cdb \gamma'_1$, and let $D_1$ be the associated "canonical database".
	Since $\delta_1 \contained \delta_2$, there exists $D_2 \cdb \delta_2$ "st" $D_2 \homto D_1$ but since $D_2$ contains no
	letter from $\alphabetmarking$, we actually get a "homomorphism" $D_2 \homto G_1$, and so $\gamma'_1 \contained \delta_2$
	"ie" $\gamma_1 \contained \gamma_2$.

	For \axiomStrongCanonCore{}, by definition of "simple regular expressions",
	no language labelling an "atom" of $\delta_1$ contains the empty word, and hence
	every "canonical database" of $\delta_1$ is "non-degenerate@@db".
	Then, let $f\colon D'_1 \to D_1$ be a "homomorphism" between "canonical databases"
	of $\delta_1$. Because of the letters of $\alphabetmarking$,
	$f$ must send $x \in \vertex{\delta_1} \subseteq \vertex{D_1}$
	onto $x \in \vertex{D'_1}$. We then claim that $f$ is "strong onto".
	Let $\alpha \defeq x \atom{L} y$ be an "atom" of $\delta_1$.
	We consider its "atom refinement" in $D_1$, and we want to show that it is
	\emph{included in} the image of the "atom refinement" of $\alpha$ in $D'_1$:
	\begin{itemize}
		\item if $L = \{\marking_x\}$, this is trivial;
		\item if $L = \{a\}$ for some letter $a$, then since
			$\alpha$ is not "locally redundant" in $\gamma'_1$,
			there are no other $a$-edge from $x$ to $y$ in $D_1$ (or $D'_1$),
			and so the unique $a$-edge from $x$ to $y$ in $D'_1$
			must be sent on the unique $a$-edge from $x$ to $y$ in $D_1$;
		\item if $L = a^+$ for some letter $a$, then the "atom refinement" of $\alpha$ in
			$D'_1$, say $x \atom{a^k} y$ ($k \geq 1$) is sent via $f$ on a path
			from $x$ to $y$ in $D_1$. If the "atom refinement" of $\alpha$ in $D_1$
			is included in this path, we are done; otherwise, when lifting
			this path to $\delta_1$, we would obtain a path of "atoms"
			$x \atom{L_1} \cdots \atom{L_n} y$ "st" $a^k \in L_1\cdots L_n$.
			By definition of "simple regular expressions", all $L_i$'s must
			be either $a$ or $a^+$, and hence in all cases $L_1 \cdots L_n \subseteq a^+$,
			contradicting that $\alpha$ is not "locally redundant" in $\gamma'_1$. 
	\end{itemize}
	Thus, we have \axiomStrongCanonCore{}.

	Similarly, \axiomCanonNonRed{} holds because all "atoms" of $\gamma'_1$ are labelled
	by $a$ or $a^+$ and we removed "locally redundant atoms".
	Thanks to the self-loops, $\delta_1$ is "fully contracted" and so \axiomCanonContracted{} holds.
	Moreover, \axiomCanonContainment{} holds trivially since $\gamma_2 = \delta_2$,
	and so does \axiomCanonMarking{} by definition of $\delta_1$.
\end{proof}
\end{toappendix}

Motivated by \Cref{lem:reduction-containment-to-minimization}
we show in the appendix that several reasonable classes admit
a polynomial-time algorithm for the "strong canonization problem"---see \Cref{lem:canonization-CRPQs,lem:canonization-SREs}.

\begin{corollaryrep}
	\AP\label{coro:lowerbounds}
	The "CRPQ" and "UCRPQ" "minimization problems" are:
	\begin{enumerate}
		\item \AP\label{expspace-h:pw1}"ExpSpace"-hard, even
		if restricted to queries of path-width at most 1,
		\item\AP\label{pspace:forest} "PSpace"-hard when restricted to "forest-shaped CRPQs",\footnote{By \AP""forest-shaped CRPQs"" we mean queries whose "underlying graph" has no undirected cycle.}
		\item\AP\label{pip2:crpqsre} "PiP2"-hard when restricted to "CRPQs" over "simple regular expressions".
	\end{enumerate}
	All hardness results are under polynomial-time reductions.
\end{corollaryrep}
\begin{proof}
	From \Cref{lem:reduction-containment-to-minimization,lem:canonization-CRPQs,lem:canonization-SREs} we can derive the stated hardness results when combined with known hardness results for the "containment problem":
	\Cref{expspace-h:pw1} follows from the "ExpSpace" lower bound of \cite[Lemma 8]{figueira_containment_2020} (or its strengthening \Cref{prop:variation-figueira}).
	\Cref{pspace:forest} follows from the trivial "PSpace" lower bound from regular language containment which is also the lower bound for one-atom "CRPQs".
	\Cref{pip2:crpqsre} follows from the known "PiP2"-lower bound for "CRPQ(SRE)" queries implied by \cite[Theorem 4.2]{FigueiraGKMNT20}.
\end{proof}

\section{Discussion}

Several open problems are left by our work, more prominently, the complexity gap for "minimization" of "CRPQs". Below we discuss further avenues for future research.

\paragraph{Variable minimization}
\AP\label{sec:discussion}
\AP\label{sec:varmin}
Another approach for an algorithm for query answering of a "(U)CRPQs" $\gamma$ on a "graph database" $G$ is by first guessing a variable assignment $f: \vars(\gamma) \to \vertex{G}$ and then checking, for each "atom" $x \atom{L} y$ of $\gamma$, that there is a path in $G$ from $f(x)$ to $f(y)$ with label in $L$. 
This implementation approach  privileges minimizing the number of \emph{variables} as opposed to the number of \emph{atoms} of a "(U)CRPQ", and gives rise to the corresponding \reintro{variable-minimization problem(s)}. From a practical perspective, as already mentioned in the Introduction, systems commonly evaluate CRPQs via join algorithms. Recent \emph{worst-case optimal} joins algorithms work by ordering the variables and assigning potential values to these, and hence the number of variables may also be a relevant parameter in these cases~\cite{cucumides-icdt23,milleniumDB24}. 
\decisionproblem{""Variable-minimization problem"" for "CRPQs" ("resp" for "UCRPQ")}
{A finite alphabet $\A$, a  "CRPQ" ("resp" "UCRPQ") $\gamma$ over $\A$ and $k \in \N$.}
{Does there exists a "CRPQ" ("resp" "UCRPQ") $\delta$ over $\A$ with at most $k$ variables 
("resp" whose every "CRPQ" has at most $k$ variables) "st" $\gamma \semequiv \delta$?}
As before, a "(U)CRPQ" is \AP""variable minimal"" if there is no "equivalent" "(U)CRPQ" smaller in the number of variables.
It is worth observing that for "conjunctive queries" (and for "tree patterns") minimizing the number of variables or minimizing the number of atoms is equivalent: a query is minimal in the number of variables "iff" it is minimal in the number of atoms.
However, for "CRPQs" and "UCRPQs" it is not: $\gamma(x,y) = x \atom{a} y \land x \atom{a+b} y$ is "variable minimal", but it is not (atom) "minimal" since it is "equivalent" to $\gamma'(x,y) = x \atom{a} y$. We further conjecture that there are (atom) "minimal" "CRPQs" which are not "variable minimal".
\begin{conjecture}
  There exist (atom) "minimal" "CRPQs" which are not "variable minimal".
\end{conjecture}

By adaptations of the algorithms of \Cref{sec:upperCRPQ,sec:upperUCRPQ} we can derive some upper bounds, which are likely to be sub-optimal.
\begin{theorem}
  The "variable-minimization problem" for "CRPQs" is in "4ExpSpace" and for "UCRPQs" in "2ExpSpace". Both problems are "ExpSpace"-hard.
\end{theorem}
\begin{proof}
  For the case of "CRPQ" "variable-minimization problem", it suffices to observe, in the proof of \Cref{lemma:crpq-size-bound}, that since each NFA of the proof has size double-exponential, there there cannot be more than a triply-exponential number of distinct NFA, and hence that the underlying multigraph of queries to be considered has a triply-exponential number of edges. Thus, there are `only' a quadruply-exponential number of such triply-exponential queries. For each such query we test if it is equivalent to the original query in "4ExpSpace".

  For the case of "UCRPQ" "variable-minimization problem", let $\+C_k$ be the (infinite) set of multigraphs having at most $k$ vertices. Via the same argument as in the proof of \Cref{lemma:approximation-for-finclass}, 
  we obtain that each $\CRPQ[\+C_k]$ in the union of "CRPQs" expressing
  $\AppInf{\Gamma}{\+C_k}=\App{\Gamma}{\+C_k}{O(\nbatoms{\Gamma}\cdot r_{\Gamma}\nbatoms{\+C})}$
  has atoms consisting in a concatenation of
  at most $O(\nbatoms{\Gamma}\cdot\nbatoms{\+C})$ "sublanguages" of $\Gamma$. Hence, there cannot be more than $2^{O(\nbatoms{\Gamma}\cdot\nbatoms{\+C})}$ distinct atoms between two variables, and $\App{\Gamma}{\+C_k}{O(\nbatoms{\gamma}\cdot r_{\Gamma}\cdot \nbatoms{\+C})} \equiv \App{\Gamma}{\+C'_k}{O(\nbatoms{\Gamma}\cdot r_{\Gamma}\nbatoms{\+C})}$ for $C'_k$ being the finite subclass of $\+C_k$ having graphs with no more than $2^{O(\nbatoms{\Gamma}\cdot\nbatoms{\+C})}$ parallel edges.
  Hence, there is a double-exponential number of exponential queries to test for "equivalence" with $\Gamma$, which yields a "2ExpSpace" upper bound.

  The lower bounds follows by a similar idea as \Cref{thm:minimization-lowerbound} and can be found in \Cref{thm:variable-minimization-lowerbound} of the Appendix.
\end{proof}

\medskip

\paragraph{Tree patterns}
We believe that the techniques of \Cref{sec:upperUCRPQ} should also yield a method to compute "maximal under-approximations" for unions of tree patterns, as well as a "PiP2" upper bound for the minimization problem of unions of tree patterns, contrasting with
the "SigmaP2"-completeness of minimization of tree patterns 
proven by Czerwiński, Martens, Niewerth \& Parys \cite[Theorem 3.1]{min-tree-patterns}.
We leave the details for a future long version of this article.

\begin{toappendix}
\subsection{Tree patterns}
\AP\label{sec:apdx-tree-patterns}
A \AP""tree pattern""---see "eg" \cite[\S 2.2]{min-tree-patterns}---over node variables
$\A$ is a directed tree, whose nodes have a label from $\A\sqcup\{*\}$, and whose
edges are partition into simple edges and transitive edges.

\begin{figure}
	\centering
	\begin{minipage}{.45\linewidth}
		\centering
		\includegraphics{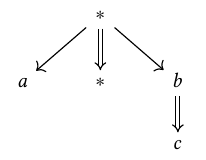}
		\subcaption{A "tree pattern" $\tau$. Double arrows represent transitive edges.}
	\end{minipage}
	\hfill
	\begin{minipage}{.45\linewidth}
		\centering
                 \includegraphics{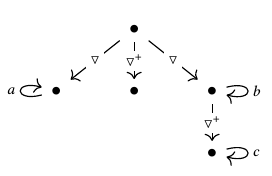}
		\subcaption{Its "encoding" $\Encoding(\tau)$ as a "CRPQ".}
	\end{minipage}
	\caption{
		\AP\label{fig:encoding-tree-pattern}
		"Encoding" a "tree pattern" into a "CRPQ".
	}
\end{figure}

We \AP""encode"" a "tree pattern" $\tau$ into a "CRPQ" over $\A\sqcup\{\marking\}$, denoted
by \AP$\intro*\Encoding(\tau)$, obtained as follows:
\begin{itemize}
	\item we start from the underlying tree of the "tree pattern",
		and replace simple edges by an "atom" $\qvar \atom{\marking} \qvar$,
		and transitive edges by an "atom" $\qvar \atom{\marking^+} \qvar$;
	\item for any node $x$ with a node label $a \in \A$, we add an "atom" $x \atom{a} x$;
	\item for any node with a wildcard label $*$, we do not add any "atom".
\end{itemize}
See \Cref{fig:encoding-tree-pattern} for an example. Note that this "encoding" is injective.

\begin{proposition}
	Given two "tree patterns" $\tau_1$, $\tau_2$ over $\A$, the following are equivalent:
	\begin{itemize}
		\item $\tau_1 \contained \tau_2$ as "tree patterns"---in the sense of
			\cite[Definition 2.2]{min-tree-patterns}, and
		\item $\Encoding(\tau_1) \contained \Encoding(\tau_2)$ as "CRPQs".
	\end{itemize}
\end{proposition}

\begin{proof}[Proof sketch]
	This follows from the characterization of containment for "tree patterns"
	using ``canonical tree models''---see \cite[\S4.2]{min-tree-patterns}\footnote{Note
	however that the authors assume the set of labels to be infinite, and label `$*$'-nodes
	by $z$-nodes where $z$ is a new label: this assumption can be removed by
	allowed unlabelled nodes in the model.}---and the characterization of "containment" for "CRPQs"
	via "canonical databases" ("aka" "expansions")---see \Cref{prop:cont-char-exp-st}.
\end{proof}

\end{toappendix}

We do not fully understand the relation between tree pattern minimization and "CRPQ minimization", 
we conjecture that if a tree pattern is minimal among tree patterns, then its "encoding" as a 
"CRPQ"---see \Cref{sec:apdx-tree-patterns} for a definition---should also be minimal,
up to "contracting internal variables", among "CRPQs" but we have failed so far to prove this.

\begin{acks}
  Diego Figueira is partially supported by ANR AI Chair INTENDED, grant ANR-19-CHIA-0014. Miguel Romero is funded by the National Center for Artificial
Intelligence CENIA FB210017, Basal ANID.
\end{acks}

\bibliographystyle{ACM-Reference-Format}
\bibliography{long,biblio}


\begin{thebibliography}{25}

%%% the Web address.

\ifx \showCODEN    \undefined \def \showCODEN     #1{\unskip}     \fi
\ifx \showDOI      \undefined \def \showDOI       #1{#1}\fi
\ifx \showISBNx    \undefined \def \showISBNx     #1{\unskip}     \fi
\ifx \showISBNxiii \undefined \def \showISBNxiii  #1{\unskip}     \fi
\ifx \showISSN     \undefined \def \showISSN      #1{\unskip}     \fi
\ifx \showLCCN     \undefined \def \showLCCN      #1{\unskip}     \fi
\ifx \shownote     \undefined \def \shownote      #1{#1}          \fi
\ifx \showarticletitle \undefined \def \showarticletitle #1{#1}   \fi
\ifx \showURL      \undefined \def \showURL       {\relax}        \fi
% invisible to TeX
\providecommand\bibfield[2]{#2}
\providecommand\bibinfo[2]{#2}
\providecommand\natexlab[1]{#1}
\providecommand\showeprint[2][]{arXiv:#2}

\bibitem[Barcel{\'{o}} et~al\mbox{.}(2019)]%
        {BarceloF019}
\bibfield{author}{\bibinfo{person}{Pablo Barcel{\'{o}}}, \bibinfo{person}{Diego
  Figueira}, {and} \bibinfo{person}{Miguel Romero}.}
  \bibinfo{year}{2019}\natexlab{}.
\newblock \showarticletitle{Boundedness of Conjunctive Regular Path Queries}.
  In \bibinfo{booktitle}{\emph{International Colloquium on Automata, Languages
  and Programming (ICALP)}} \emph{(\bibinfo{series}{LIPIcs},
  Vol.~\bibinfo{volume}{132})}. \bibinfo{publisher}{Schloss Dagstuhl -
  Leibniz-Zentrum f{\"{u}}r Informatik}, \bibinfo{pages}{104:1--104:15}.
\newblock


\bibitem[Barcel{\'{o}} et~al\mbox{.}(2016)]%
        {BarceloRV16}
\bibfield{author}{\bibinfo{person}{Pablo Barcel{\'{o}}},
  \bibinfo{person}{Miguel Romero}, {and} \bibinfo{person}{Moshe~Y. Vardi}.}
  \bibinfo{year}{2016}\natexlab{}.
\newblock \showarticletitle{Semantic Acyclicity on Graph Databases}.
\newblock \bibinfo{journal}{\emph{{SIAM} Journal on computing}}
  \bibinfo{volume}{45}, \bibinfo{number}{4} (\bibinfo{year}{2016}),
  \bibinfo{pages}{1339--1376}.
\newblock
\urldef\tempurl%
\url{https://doi.org/10.1137/15M1034714}
\showDOI{\tempurl}


\bibitem[Calvanese et~al\mbox{.}(2000)]%
        {four-italians}
\bibfield{author}{\bibinfo{person}{Diego Calvanese},
  \bibinfo{person}{Giuseppe~De Giacomo}, \bibinfo{person}{Maurizio Lenzerini},
  {and} \bibinfo{person}{Moshe~Y. Vardi}.} \bibinfo{year}{2000}\natexlab{}.
\newblock \showarticletitle{Containment of Conjunctive Regular Path Queries
  with Inverse}. In \bibinfo{booktitle}{\emph{Principles of Knowledge
  Representation and Reasoning (KR)}}. \bibinfo{publisher}{Morgan Kaufmann},
  \bibinfo{pages}{176--185}.
\newblock


\bibitem[Chandra and Merlin(1977)]%
        {DBLP:conf/stoc/ChandraM77}
\bibfield{author}{\bibinfo{person}{Ashok~K. Chandra} {and}
  \bibinfo{person}{Philip~M. Merlin}.} \bibinfo{year}{1977}\natexlab{}.
\newblock \showarticletitle{Optimal Implementation of Conjunctive Queries in
  Relational Data Bases}. In \bibinfo{booktitle}{\emph{Symposium on Theory of
  Computing (STOC)}}. \bibinfo{publisher}{{ACM}}, \bibinfo{pages}{77--90}.
\newblock
\urldef\tempurl%
\url{https://doi.org/10.1145/800105.803397}
\showDOI{\tempurl}


\bibitem[Cucumides et~al\mbox{.}(2023)]%
        {cucumides-icdt23}
\bibfield{author}{\bibinfo{person}{Tamara Cucumides}, \bibinfo{person}{Juan
  Reutter}, {and} \bibinfo{person}{Domagoj Vrgo\v{c}}.}
  \bibinfo{year}{2023}\natexlab{}.
\newblock \showarticletitle{{Size Bounds and Algorithms for Conjunctive Regular
  Path Queries}}. In \bibinfo{booktitle}{\emph{International Conference on
  Database Theory (ICDT)}} \emph{(\bibinfo{series}{Leibniz International
  Proceedings in Informatics (LIPIcs)})}. \bibinfo{publisher}{Leibniz-Zentrum
  f{\"u}r Informatik}, \bibinfo{pages}{13:1--13:17}.
\newblock


\bibitem[Czerwi\'{n}ski et~al\mbox{.}(2018)]%
        {min-tree-patterns}
\bibfield{author}{\bibinfo{person}{Wojciech Czerwi\'{n}ski},
  \bibinfo{person}{Wim Martens}, \bibinfo{person}{Matthias Niewerth}, {and}
  \bibinfo{person}{Pawe\l{} Parys}.} \bibinfo{year}{2018}\natexlab{}.
\newblock \showarticletitle{Minimization of Tree Patterns}.
\newblock \bibinfo{journal}{\emph{J. ACM}} \bibinfo{volume}{65},
  \bibinfo{number}{4}, Article \bibinfo{articleno}{26} (\bibinfo{year}{2018}),
  \bibinfo{numpages}{46}~pages.
\newblock


\bibitem[Feier et~al\mbox{.}(2024)]%
        {FeierGM24}
\bibfield{author}{\bibinfo{person}{Cristina Feier}, \bibinfo{person}{Tomasz
  Gogacz}, {and} \bibinfo{person}{Filip Murlak}.}
  \bibinfo{year}{2024}\natexlab{}.
\newblock \showarticletitle{Evaluating Graph Queries Using Semantic Treewidth}.
  In \bibinfo{booktitle}{\emph{International Conference on Database Theory
  (ICDT)}} \emph{(\bibinfo{series}{LIPIcs}, Vol.~\bibinfo{volume}{290})}.
  \bibinfo{publisher}{Schloss Dagstuhl - Leibniz-Zentrum f{\"{u}}r Informatik},
  \bibinfo{pages}{22:1--22:20}.
\newblock


\bibitem[Figueira(2020)]%
        {figueira_containment_2020}
\bibfield{author}{\bibinfo{person}{Diego Figueira}.}
  \bibinfo{year}{2020}\natexlab{}.
\newblock \showarticletitle{Containment of {UC2RPQ:} The Hard and Easy Cases}.
  In \bibinfo{booktitle}{\emph{International Conference on Database Theory
  (ICDT)}} \emph{(\bibinfo{series}{Leibniz International Proceedings in
  Informatics (LIPIcs)}, Vol.~\bibinfo{volume}{155})}.
  \bibinfo{publisher}{Leibniz-Zentrum f{\"u}r Informatik},
  \bibinfo{pages}{9:1--9:18}.
\newblock
\urldef\tempurl%
\url{https://doi.org/10.4230/LIPICS.ICDT.2020.9}
\showDOI{\tempurl}


\bibitem[Figueira et~al\mbox{.}(2020)]%
        {FigueiraGKMNT20}
\bibfield{author}{\bibinfo{person}{Diego Figueira}, \bibinfo{person}{Adwait
  Godbole}, \bibinfo{person}{S. Krishna}, \bibinfo{person}{Wim Martens},
  \bibinfo{person}{Matthias Niewerth}, {and} \bibinfo{person}{Tina Trautner}.}
  \bibinfo{year}{2020}\natexlab{}.
\newblock \showarticletitle{Containment of Simple Conjunctive Regular Path
  Queries}. In \bibinfo{booktitle}{\emph{Principles of Knowledge Representation
  and Reasoning (KR)}}. \bibinfo{pages}{371--380}.
\newblock


\bibitem[Figueira et~al\mbox{.}(2024)]%
        {FigueiraAnanthaAl24}
\bibfield{author}{\bibinfo{person}{Diego Figueira}, \bibinfo{person}{S.
  Krishna}, \bibinfo{person}{Om~Swostik Mishra}, {and} \bibinfo{person}{Anantha
  Padmanabha}.} \bibinfo{year}{2024}\natexlab{}.
\newblock \showarticletitle{Boundedness for Unions of Conjunctive Regular Path
  Queries over Simple Regular Expressions}. In
  \bibinfo{booktitle}{\emph{Principles of Knowledge Representation and
  Reasoning (KR)}}.
\newblock


\bibitem[Figueira and Morvan(2023)]%
        {FigueiraM23}
\bibfield{author}{\bibinfo{person}{Diego Figueira} {and}
  \bibinfo{person}{R{\'{e}}mi Morvan}.} \bibinfo{year}{2023}\natexlab{}.
\newblock \showarticletitle{Approximation and Semantic Tree-Width of
  Conjunctive Regular Path Queries}. In \bibinfo{booktitle}{\emph{International
  Conference on Database Theory (ICDT)}} \emph{(\bibinfo{series}{LIPIcs},
  Vol.~\bibinfo{volume}{255})}. \bibinfo{publisher}{Leibniz-Zentrum f{\"u}r
  Informatik}, \bibinfo{pages}{15:1--15:19}.
\newblock


\bibitem[Figueira and Morvan(2025)]%
        {FM2023semantic}
\bibfield{author}{\bibinfo{person}{Diego Figueira} {and}
  \bibinfo{person}{R{\'e}mi Morvan}.} \bibinfo{year}{2025}\natexlab{}.
\newblock \showarticletitle{{Semantic Tree-Width and Path-Width of Conjunctive
  Regular Path Queries}}.
\newblock \bibinfo{journal}{\emph{Logical Methods in Computer Science (LMCS)}}
  \bibinfo{volume}{21}, \bibinfo{number}{1} (\bibinfo{date}{March}
  \bibinfo{year}{2025}).
\newblock
\urldef\tempurl%
\url{https://doi.org/10.46298/lmcs-21(1:21)2025}
\showDOI{\tempurl}


\bibitem[Figueira et~al\mbox{.}(2025)]%
        {thispaper}
\bibfield{author}{\bibinfo{person}{Diego Figueira}, \bibinfo{person}{Rémi
  Morvan}, {and} \bibinfo{person}{Miguel Romero}.}
  \bibinfo{year}{2025}\natexlab{}.
\newblock \showarticletitle{Minimizing Conjunctive Regular Path Queries}.
\newblock \bibinfo{journal}{\emph{Proceedings of the {ACM} on Management of
  Data (PACMMOD)}} \bibinfo{volume}{3}, \bibinfo{number}{2 (PODS)}
  (\bibinfo{year}{2025}).
\newblock


\bibitem[Figueira and Romero(2023)]%
        {FigueiraRomero23}
\bibfield{author}{\bibinfo{person}{Diego Figueira} {and}
  \bibinfo{person}{Miguel Romero}.} \bibinfo{year}{2023}\natexlab{}.
\newblock \showarticletitle{Conjunctive Regular Path Queries under Injective
  Semantics}. In \bibinfo{booktitle}{\emph{ACM Symposium on Principles of
  Database Systems (PODS)}}. \bibinfo{publisher}{ACM Press},
  \bibinfo{pages}{231--240}.
\newblock


\bibitem[Flesca et~al\mbox{.}(2008)]%
        {FFM08}
\bibfield{author}{\bibinfo{person}{S. Flesca}, \bibinfo{person}{F. Furfaro},
  {and} \bibinfo{person}{E. Masciari}.} \bibinfo{year}{2008}\natexlab{}.
\newblock \showarticletitle{On the minimization of XPath queries}.
\newblock \bibinfo{journal}{\emph{J. ACM}} \bibinfo{volume}{55},
  \bibinfo{number}{1} (\bibinfo{year}{2008}).
\newblock


\bibitem[Florescu et~al\mbox{.}(1998)]%
        {Florescu:CRPQ}
\bibfield{author}{\bibinfo{person}{Daniela Florescu}, \bibinfo{person}{Alon
  Levy}, {and} \bibinfo{person}{Dan Suciu}.} \bibinfo{year}{1998}\natexlab{}.
\newblock \showarticletitle{Query Containment for Conjunctive Queries with
  Regular Expressions}. In \bibinfo{booktitle}{\emph{ACM Symposium on
  Principles of Database Systems (PODS)}}. \bibinfo{publisher}{ACM Press},
  \bibinfo{pages}{139--148}.
\newblock
\urldef\tempurl%
\url{https://doi.org/10.1145/275487.275503}
\showDOI{\tempurl}


\bibitem[Francis et~al\mbox{.}(2023a)]%
        {DBLP:conf/pods/FrancisGGLMMMPR23}
\bibfield{author}{\bibinfo{person}{Nadime Francis},
  \bibinfo{person}{Am{\'{e}}lie Gheerbrant}, \bibinfo{person}{Paolo
  Guagliardo}, \bibinfo{person}{Leonid Libkin}, \bibinfo{person}{Victor
  Marsault}, \bibinfo{person}{Wim Martens}, \bibinfo{person}{Filip Murlak},
  \bibinfo{person}{Liat Peterfreund}, \bibinfo{person}{Alexandra Rogova}, {and}
  \bibinfo{person}{Domagoj Vrgo\v{c}}.} \bibinfo{year}{2023}\natexlab{a}.
\newblock \showarticletitle{{GPC:} {A} Pattern Calculus for Property Graphs}.
  In \bibinfo{booktitle}{\emph{ACM Symposium on Principles of Database Systems
  (PODS)}}. \bibinfo{publisher}{ACM Press}, \bibinfo{pages}{241--250}.
\newblock
\urldef\tempurl%
\url{https://doi.org/10.1145/3584372.3588662}
\showDOI{\tempurl}


\bibitem[Francis et~al\mbox{.}(2023b)]%
        {DBLP:conf/icdt/FrancisGGLMMMPR23}
\bibfield{author}{\bibinfo{person}{Nadime Francis},
  \bibinfo{person}{Am{\'{e}}lie Gheerbrant}, \bibinfo{person}{Paolo
  Guagliardo}, \bibinfo{person}{Leonid Libkin}, \bibinfo{person}{Victor
  Marsault}, \bibinfo{person}{Wim Martens}, \bibinfo{person}{Filip Murlak},
  \bibinfo{person}{Liat Peterfreund}, \bibinfo{person}{Alexandra Rogova}, {and}
  \bibinfo{person}{Domagoj Vrgo\v{c}}.} \bibinfo{year}{2023}\natexlab{b}.
\newblock \showarticletitle{A Researcher's Digest of {GQL} (Invited Talk)}. In
  \bibinfo{booktitle}{\emph{International Conference on Database Theory
  (ICDT)}} \emph{(\bibinfo{series}{Leibniz International Proceedings in
  Informatics (LIPIcs)}, Vol.~\bibinfo{volume}{255})}.
  \bibinfo{publisher}{Leibniz-Zentrum f{\"u}r Informatik},
  \bibinfo{pages}{1:1--1:22}.
\newblock
\urldef\tempurl%
\url{https://doi.org/10.4230/LIPICS.ICDT.2023.1}
\showDOI{\tempurl}


\bibitem[Guti{\'{e}}rrez-Basulto et~al\mbox{.}(2022)]%
        {GGIM-kr22}
\bibfield{author}{\bibinfo{person}{V{\'{i}}ctor Guti{\'{e}}rrez-Basulto},
  \bibinfo{person}{Albert Gutowski}, \bibinfo{person}{Yazm{\'{i}}n
  Ib{\'{a}}{\~{n}}ez-Garc{\'{i}}a}, {and} \bibinfo{person}{Filip Murlak}.}
  \bibinfo{year}{2022}\natexlab{}.
\newblock \showarticletitle{{Finite Entailment of UCRPQs over ALC Ontologies}}.
  In \bibinfo{booktitle}{\emph{Principles of Knowledge Representation and
  Reasoning (KR)}}. \bibinfo{pages}{184--194}.
\newblock


\bibitem[Guti{\'{e}}rrez-Basulto et~al\mbox{.}(2024)]%
        {GGIM-pods24}
\bibfield{author}{\bibinfo{person}{V{\'{i}}ctor Guti{\'{e}}rrez-Basulto},
  \bibinfo{person}{Albert Gutowski}, \bibinfo{person}{Yazm{\'{i}}n
  Ib{\'{a}}{\~{n}}ez-Garc{\'{i}}a}, {and} \bibinfo{person}{Filip Murlak}.}
  \bibinfo{year}{2024}\natexlab{}.
\newblock \showarticletitle{Containment of Graph Queries Modulo Schema}. In
  \bibinfo{booktitle}{\emph{ACM Symposium on Principles of Database Systems
  (PODS)}}. \bibinfo{publisher}{ACM Press}, \bibinfo{pages}{1--26}.
\newblock


\bibitem[{International Organization for Standardization (ISO)}(2024a)]%
        {isoGQL}
\bibfield{author}{\bibinfo{person}{{International Organization for
  Standardization (ISO)}}.} \bibinfo{year}{2024}\natexlab{a}.
\newblock \bibinfo{title}{ISO/IEC 39075:2024 GQL}.
\newblock
  \bibinfo{howpublished}{\url{https://www.iso.org/standard/76120.html}}.
\newblock
\newblock
\shownote{Released in April 2024 by ISO/IEC}.


\bibitem[{International Organization for Standardization (ISO)}(2024b)]%
        {isoPGQ}
\bibfield{author}{\bibinfo{person}{{International Organization for
  Standardization (ISO)}}.} \bibinfo{year}{2024}\natexlab{b}.
\newblock \bibinfo{title}{ISO/IEC 9075-16:2023 Part 16: Property Graph Queries
  (SQL/PGQ)}.
\newblock
  \bibinfo{howpublished}{\url{https://www.iso.org/standard/79473.html}}.
\newblock
\newblock
\shownote{Released in June 2023 by ISO/IEC}.


\bibitem[Karalis et~al\mbox{.}(2024)]%
        {eswc-crpqs24}
\bibfield{author}{\bibinfo{person}{Nikolaos Karalis},
  \bibinfo{person}{Alexander Bigerl}, \bibinfo{person}{Liss Heidrich},
  \bibinfo{person}{Mohamed~Ahmed Sherif}, {and} \bibinfo{person}{Axel-Cyrille
  Ngonga~Ngomo}.} \bibinfo{year}{2024}\natexlab{}.
\newblock \showarticletitle{Efficient Evaluation of Conjunctive Regular Path
  Queries Using Multi-way Joins}. \bibinfo{publisher}{Springer},
  \bibinfo{pages}{218--235}.
\newblock


\bibitem[Kimelfeld and Sagiv(2008)]%
        {KS08}
\bibfield{author}{\bibinfo{person}{Benny Kimelfeld} {and}
  \bibinfo{person}{Yehoshua Sagiv}.} \bibinfo{year}{2008}\natexlab{}.
\newblock \showarticletitle{Revisiting redundancy and minimization in an XPath
  fragment}. In \bibinfo{booktitle}{\emph{International Conference on Extending
  Database Technology (EDBT)}}. \bibinfo{publisher}{ACM Press},
  \bibinfo{pages}{61--72}.
\newblock


\bibitem[Vrgo\v{c} et~al\mbox{.}(2024)]%
        {milleniumDB24}
\bibfield{author}{\bibinfo{person}{Domagoj Vrgo\v{c}}, \bibinfo{person}{Carlos
  Rojas}, \bibinfo{person}{Renzo Angles}, \bibinfo{person}{Marcelo Arenas},
  \bibinfo{person}{Vicente Calisto}, \bibinfo{person}{Benjam\'{\i}n
  Far\'{\i}as}, \bibinfo{person}{Sebasti\'{a}n Ferrada},
  \bibinfo{person}{Tristan Heuer}, \bibinfo{person}{Aidan Hogan},
  \bibinfo{person}{Gonzalo Navarro}, \bibinfo{person}{Alexander Pinto},
  \bibinfo{person}{Juan Reutter}, \bibinfo{person}{Henry Rosales}, {and}
  \bibinfo{person}{Etienne Toussiant}.} \bibinfo{year}{2024}\natexlab{}.
\newblock \showarticletitle{MillenniumDB: A Multi-modal, Multi-model Graph
  Database}. In \bibinfo{booktitle}{\emph{ACM Symposium on Principles of
  Database Systems (PODS)}}. \bibinfo{publisher}{ACM Press},
  \bibinfo{pages}{496--499}.
\newblock
\showISBNx{9798400704222}


\end{thebibliography}

\end{document}